\documentclass[aps,prx,twocolumn,10pt]{revtex4-1}

\usepackage{mathrsfs}
\usepackage{bbm}
\usepackage{amsmath}
\usepackage{amssymb}
\usepackage{amsthm}
\usepackage{graphicx}
\usepackage{MnSymbol}
\usepackage{mathtools}
\usepackage{stmaryrd}
\usepackage{enumitem}
\usepackage{bbold}
\usepackage[italicdiff]{physics}

\makeatletter
\usepackage{hyperref}
\usepackage{color}
\definecolor{supcol}{RGB}{10,50,180}
\hypersetup{
	colorlinks,
	citecolor=supcol,
	linkcolor=supcol,
	urlcolor=supcol
}

\allowdisplaybreaks

\DeclareMathOperator{\Tan}{Tan}

\newcommand{\mca}{\mathcal}
\newcommand{\mbb}{\mathbb}

\newcommand{\msf}{\mathsf}

\newcommand{\sop}[1]{\llbracket #1 \rrbracket}
\newcommand{\dket}[1]{| #1 \rangle\!\rangle}

\newcommand{\peq}{p^{\rm eq}}

\newcommand{\Var}[1]{{\rm var}[#1]}

\newtheorem{theorem}{Theorem}

\newtheorem{proposition}[theorem]{Proposition}

\newtheorem{lemma}[theorem]{Lemma}
\newtheorem{corollary}[theorem]{Corollary}

\begin{document}
\title{Thermodynamic Unification of Optimal Transport: Thermodynamic Uncertainty Relation, Minimum Dissipation, and Thermodynamic Speed Limits}

\author{Tan Van Vu}
\email{tanvu@rk.phys.keio.ac.jp}

\affiliation{Department of Physics, Keio University, 3-14-1 Hiyoshi, Kohoku-ku, Yokohama 223-8522, Japan}

\author{Keiji Saito}
\email{saitoh@rk.phys.keio.ac.jp}

\affiliation{Department of Physics, Keio University, 3-14-1 Hiyoshi, Kohoku-ku, Yokohama 223-8522, Japan}

\begin{abstract}
Thermodynamics serves as a universal means for studying physical systems from an energy perspective.
In recent years, with the establishment of the field of stochastic and quantum thermodynamics, the ideas of thermodynamics have been generalized to small fluctuating systems. 
Independently developed in mathematics and statistics, the optimal transport theory concerns the means by which one can optimally transport a source distribution to a target distribution, deriving a useful metric between probability distributions, called the Wasserstein distance.
Despite their seemingly unrelated nature, an intimate connection between these fields has been unveiled in the context of continuous-state Langevin dynamics, providing several important implications for nonequilibrium systems.
In this study, we elucidate an analogous connection for discrete cases by developing a thermodynamic framework for discrete optimal transport. 
We first introduce a novel quantity called dynamical state mobility, which significantly improves the thermodynamic uncertainty relation and provides insights into the precision of currents in nonequilibrium Markov jump processes.
We then derive variational formulas that connect the discrete Wasserstein distances to stochastic and quantum thermodynamics of discrete Markovian dynamics described by master equations.
Specifically, we rigorously prove that the Wasserstein distance equals the minimum product of irreversible entropy production and dynamical state mobility over all admissible Markovian dynamics.
These formulas not only unify the relationship between thermodynamics and the optimal transport theory for discrete and continuous cases but also generalize it to the quantum case.
In addition, we demonstrate that the obtained variational formulas lead to remarkable applications in stochastic and quantum thermodynamics, such as stringent thermodynamic speed limits and the finite-time Landauer principle.
These bounds are tight and can be saturated for arbitrary temperatures, even in the zero-temperature limit. Notably, the finite-time Landauer principle can explain finite dissipation even at extremely low temperatures, which cannot be explained by the conventional Landauer principle.
\end{abstract}

\date{\today}
\maketitle

\section{Introduction}\label{sec:intro}
\subsection{Background}

Thermodynamics, which is built upon several axioms, is one of the most successful phenomenological theories for studying energy exchanges in macroscopic systems.
Originally developed for the purpose of understanding the behavior of steam engines, thermodynamics has since been applied to various fields of science and engineering.
The laws of thermodynamics show extraordinary universality and impose fundamental constraints on physical systems.

Beyond the macroscopic regime, the past two decades have witnessed substantial progress in extending the notions of conventional thermodynamics to microscopic systems, resulting in the frameworks of stochastic and quantum thermodynamics \cite{Sekimoto.2010,Seifert.2012.RPP,Vinjanampathy.2016.CP,Goold.2016.JPA,Deffner.2019}.
These comprehensive frameworks provide a means of investigating small nonequilibrium systems subject to significant fluctuations.
Various universal relations have been discovered, including fluctuation theorems \cite{Evans.1993.PRL,Gallavotti.1995.PRL,Crooks.1999.PRE,Jarzynski.2000.JSP,Esposito.2009.RMP,Campisi.2011.RMP}, thermodynamic uncertainty relations \cite{Barato.2015.PRL,Gingrich.2016.PRL,Horowitz.2017.PRE,Horowitz.2020.NP}, thermodynamic speed limits \cite{Shiraishi.2018.PRL,Ito.2018.PRL,Funo.2019.NJP,Ito.2020.PRX,Gupta.2020.PRE,Vo.2020.PRE,Vu.2021.PRL,Yoshimura.2021.PRL,Delvenne.2021.arxiv,Salazar.2022.arxiv,Vo.2022.arxiv}, and refinements of the Landauer principle \cite{Aurell.2011.PRL,Goold.2015.PRL,Proesmans.2020.PRL,Zhen.2021.PRL,Vu.2022.PRL,Lee.2022.arxiv}.
These equalities and inequalities characterize the fundamental limits of small systems and distinguish the possible from the impossible in terms of thermodynamics.
They are not only theoretically important but also lead to practical applications in estimating physically relevant quantities from experimental data, such as free energy \cite{Gore.2003.PNAS} and dissipation \cite{Li.2019.NC,Manikandan.2020.PRL,Vu.2020.PRE,Otsubo.2020.PRE,Kim.2020.PRL,Skinner.2021.PRL,Dechant.2021.PRX}.
In addition, information manipulations such as measurement, feedback, erasure, and copying have been incorporated into thermodynamics, leading to significant developments in several subfields, such as the thermodynamics of information \cite{Sagawa.2012.PTP,Parrondo.2015.NP} and computation \cite{Bennett.1982.IJTP,Wolpert.2019.JPA}.
In parallel, concepts from other fields, such as the resource theory \cite{Chitambar.2019.RMP,Lostaglio.2019.RPP} and information geometry \cite{Salamon.1983.PRL,Ruppeiner.1995.RMP,Amari.2000}, have also been integrated into thermodynamics, generating new avenues of research and offering new tools for analyzing thermodynamic processes \cite{Crooks.2007.PRL,Feng.2008.PRL,Machta.2015.PRL,Rotskoff.2017.PRE.GeoApp,Nicholson.2018.PRE,Scandi.2019.Q,Bryant.2020.PNAS,Abiuso.2020.E,Brandner.2020.PRL}.
Accordingly, the integration of thermodynamics with other disciplines provides new insights into the understanding of nonequilibrium systems.
\begin{figure*}[t]
\centering
\includegraphics[width=1\linewidth]{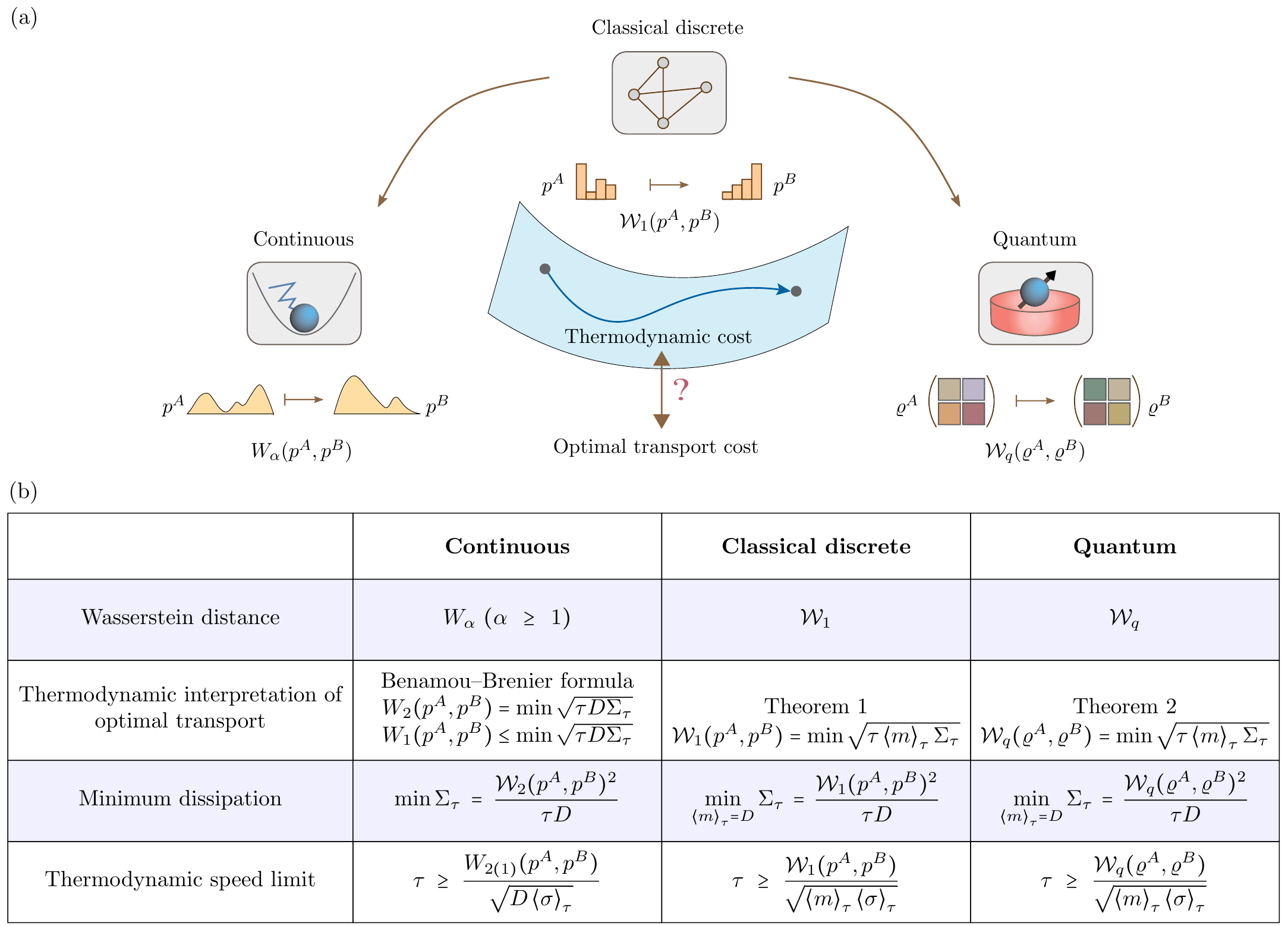}
\protect\caption{(a) Schematic of stochastic and quantum thermodynamics of discrete Markovian dynamics and the optimal transport problem. This study aims to reveal the connection between thermodynamics and the optimal transport theory for discrete cases, including both classical and quantum systems. (b) Our thermodynamic unification of optimal transport and its consequences, including minimum dissipation and thermodynamic speed limits. By introducing dynamical state mobility $m_t$, we provide a unified, discrete generalization of the Benamou--Brenier formula for classical and quantum discrete systems. See the subsection ``Summary of results'' for an explanation.}\label{fig:Cover}
\end{figure*}

A key quantity in thermodynamics is entropy production, which quantifies the degree of irreversibility of thermodynamic processes.
Entropy production plays a central role in the fundamental laws of thermodynamics and provides a quantitative characterization for investigating nonequilibrium processes; a comprehensive review regarding entropy production can be found in Ref.~\cite{Landi.2021.RMP}.
Recently, it has been shown that entropy production must be increased to achieve a high precision of currents \cite{Gingrich.2016.PRL} and fast state transformation \cite{Shiraishi.2018.PRL}.
However, minimizing entropy production is also a particularly relevant issue \cite{Schmiedl.2007.PRL,Chennakesavalu.2022.arxiv} because it is closely related to the energy lost to the environment.
Owing to the critical role of entropy production, great efforts have been made to elucidate its properties and its relationship with other physical quantities \cite{Saito.2016.EPL,Shiraishi.2016.PRL,Neri.2017.PRX,Pigolotti.2017.PRL,Pietzonka.2018.PRL,Manzano.2019.PRL,Falasco.2020.PRL}.

Optimal transport \cite{Villani.2008}, which is developed independently of thermodynamics, is a mature field in mathematics and statistics, and its theory concerns the optimal planning and optimal cost of transporting a distribution.
Specifically, given the individual costs of transporting a unit weight of a resource from one location to another, the optimal transport problem is to determine the optimal means of redistributing the distribution of the resource into the desired distribution to yield the lowest total cost.
Historically, the optimal transport problem was first defined by Monge in 1781 and has since been reformulated in a more general and well-defined form.
Currently, this problem has several theoretical and practical applications in a variety of scientific fields, including statistics and machine learning \cite{Kolouri.2017.SPM}, computer vision \cite{Haker.2004.IJCV}, linguistics \cite{Huang.2016.NIPS}, classical mechanics \cite{Koehl.2019.PRL}, and molecular biology \cite{Schiebinger.2019.Cell}.
It is noteworthy that the solution to this problem not only provides an optimal transport plan between distributions but also defines a useful metric in the space of probability distributions.
Although this metric has been identified by several names in the literature, such as the Monge--Kantorovich distance or earth mover's distance, we refer to it as the Wasserstein distance throughout this paper.

Because the optimal transport theory is concerned with transformations of probability distributions that are commonly used to characterize the state of small systems, whether any connection exists between the two disciplines of optimal transport and stochastic thermodynamics is a natural question.
Indeed, a deep connection between these fields has been elucidated in the context of overdamped Langevin dynamics, revealing that the problem of minimizing entropy production can be mapped to the optimal transport problem \cite{Jordan.1998.JMA,Aurell.2011.PRL,Aurell.2012.JSP,Dechant.2019.arxiv,Nakazato.2021.PRR,Chennakesavalu.2022.arxiv}.
More specifically, the minimum entropy production among all processes that transform the initial into the final distribution can be expressed in terms of the Wasserstein distance between the two distributions.
In addition, the optimal transport plan provides a feasible solution for the optimal control protocol.
The essence of this connection can be intuitively understood through the Benamou--Brenier formula \cite{Benamou.2000.NM}, which is given by the following equality:
\begin{align}
W_2(p^A,p^B)=\min{ \int_0^\tau \sqrt{D\sigma_t}\dd{t}}=\min\sqrt{D\tau\Sigma_\tau}.\label{eq:BB.form.intro}
\end{align}
Here, $W_2(p^A,p^B)$ is the $L^2$-Wasserstein distance [cf. Eq.~\eqref{eq:cWass.dist.def}], $D$ is the diffusion coefficient, $\tau$ is the operational time, $\sigma_t$ is the entropy production rate, $\Sigma_\tau$ is the total entropy production, and the minimum is taken over all the overdamped Langevin processes that transform distribution $p^A$ into $p^B$.
The variational relation \eqref{eq:BB.form.intro} links two apparently unrelated quantities, namely a mathematical metric and a thermodynamic cost.
The metric $W_2(p^A,p^B)$ on the left-hand side of Eq.~\eqref{eq:BB.form.intro} is a static quantity that is determined only by two distributions, whereas the right-hand side represents a dynamical quantity that indicates the thermodynamic cost associated with overdamped Langevin dynamics.
This formula leads to remarkable applications for overdamped Langevin dynamics, such as a thermodynamic speed limit \cite{Dechant.2019.arxiv} and a finite-time Landauer principle of information erasure for classical bits modeled by a continuous double-well potential \cite{Proesmans.2020.PRL,Proesmans.2020.PRE}, to name only two.
The finite-time correction in the Landauer principle indicated by the speed limit expression is consistent with experimental observations \cite{Brut.2012.N}.
Moreover, the bounds obtained from Eq.~\eqref{eq:BB.form.intro} are tight in the sense that, for any two given distributions, there always exists an overdamped Langevin dynamics that transforms the distributions and attains the equality of the bounds.

In contrast to continuous systems, a similar connection between optimal transport and thermodynamics is yet to be unveiled in discrete systems.
We note that stochastic thermodynamics of discrete systems is highly relevant to experiments \cite{Schuler.2005.PRL,Hekking.2013.PRL,Koski.2014.PNAS,Josefsson.2018.NN}.
Even in continuous systems such as biological systems, the dynamics can be represented by effective discrete states \cite{Elowitz.2002.S,Schliwa.2003.N,Stigler.2011.S,Cheong.2011.S}.
In addition, the Landauer principle of information erasure problem is a statement for discrete bit operations \cite{Landauer.1961.JRD}.
Therefore, elucidating the thermodynamic interpretation of the discrete optimal transport problem is essential for an in-depth understanding of the nonequilibrium thermodynamic structure and, particularly, for its application to the state transformation speed.

To reveal this type of relationship for discrete systems, two nontrivial points are worth noting.
First, the formula \eqref{eq:BB.form.intro} cannot be extended directly to discrete cases because no exact correspondence to the diffusion coefficient exists in generic discrete systems.
Even if a proper correspondence to the diffusion coefficient is defined for discrete cases, no guarantee can be given that the discrete Wasserstein distance can be expressed in the same manner as in Eq.~\eqref{eq:BB.form.intro}.
Second, previous studies have shown that without any additional constraint, the distribution of Markov jump processes can always be transformed to the target distribution with arbitrarily small entropy production \cite{Vu.2021.PRL2,Remlein.2021.PRE,Dechant.2022.JPA}.
This implies that entropy production alone is insufficient to characterize the transport cost (i.e., the Wasserstein distance), or equivalently, this implies that another quantity that plays the same role as the diffusion coefficient in continuous cases must be introduced along with entropy production.
These technical remarks are an obstacle to elucidating the relationship between optimal transport and thermodynamics in discrete cases.
Simultaneously, overcoming this obstacle is expected to reveal essential and common thermodynamic structures hidden in nonequilibrium processes.

With this background, we aim to elucidate the deep connection between thermodynamics and optimal transport in discrete cases (see Fig.~\ref{fig:Cover} for illustration).
Specifically, we develop discrete generalizations of the Benamou--Brenier formula in the context of Markovian open classical and quantum dynamics described by the master equations.
Our formulas not only unify the relationship between optimal transport and stochastic thermodynamics for discrete and continuous cases but also generalize to the quantum case.
Moreover, by developing a thermodynamic framework for discrete optimal transport, we can derive fundamental bounds for nonequilibrium systems, including the thermodynamic uncertainty relation, thermodynamic speed limits, and finite-time Landauer principle for both classical and quantum systems.
These bounds are tight and stronger than previously reported results.
Concerning the Landauer principle, finite-time information erasure should generate finite heat dissipation even at zero temperature. 
However, neither the original Landauer bound \cite{Landauer.1961.JRD} nor finite-time corrections that have been proposed thus far for discrete systems \cite{Zhen.2021.PRL,Vu.2022.PRL} can explain heat dissipation at extremely low temperatures. 
By contrast, our expression for the first time can predict heat dissipation even at extremely low temperatures.

\subsection{Summary of results}

The main contributions of this study can be summarized as follows.
\begin{itemize}
	\item[(1)] \emph{Dynamical state mobility and improved thermodynamic uncertainty relation}.---We define a novel kinetic quantity $m_t$ [cf.~Eq.~\eqref{eq:lt.def}], which is essential to our results. The motivation for this definition is derived from Eq.~\eqref{eq:BB.form.intro}, which suggests that a kinetic term is relevant for characterizing the Wasserstein distance. $m_t$ is defined by the microscopic Onsager-like coefficients and reduces to the diffusion coefficient $D$ in the continuous limit; thus, it is referred to as dynamical state mobility. Similar to dynamical activity \cite{Maes.2020.PR}, $m_t$ should be measurable in experiments. Using this kinetic term, we derive an improved thermodynamic uncertainty relation for time-integrated currents in Markov jump processes, which can be expressed as [cf.~Eq.~\eqref{eq:new.TUR}]
	\begin{equation}\label{eq:new.TUR.intr}
		\frac{\ev{J}^2}{\Var{J}}\le\eta\frac{\Sigma_\tau}{2},
	\end{equation}
	where $\ev{J}$ and $\Var{J}$ denote the mean and variance of an arbitrary current $J$, respectively, and $\eta\coloneqq 2\mca{M}_\tau/\mca{A}_\tau\le 1$ is an efficiency defined in terms of dynamical state mobility $\mca{M}_\tau\coloneqq\int_0^\tau m_t\dd{t}$ and dynamical activity $\mca{A}_\tau$.
	The inequality \eqref{eq:new.TUR.intr} indicates that the precision of currents is constrained by the product of the thermodynamic and kinetic costs divided by the timescale.
	Moreover, it provides new insights into the relationship between precision and cost in Markov jump processes; that is, increasing only the thermodynamic cost does not guarantee high precision of currents. Instead, given the same timescale $\mca{A}_\tau$, the product of the thermodynamic and kinetic costs must be increased to achieve high precision.
	Notably, the relation \eqref{eq:new.TUR.intr} is tighter than the conventional thermodynamic uncertainty relation \cite{Gingrich.2016.PRL,Horowitz.2017.PRE}.
	
	\item[(2)] \emph{Variational formulas that connect optimal transport to stochastic and quantum thermodynamics}.---Using the defined state mobility term, we derive variational formulas that relate the discrete Wasserstein distance to the thermodynamic and kinetic costs in Markovian dynamics. More specifically, we prove the following equality for the classical case (cf.~Thm.~\ref{thm:cla.dis.Wass.var}):
	\begin{align}
		\mca{W}_1(p^A,p^B)&=\min{ \int_0^\tau\sqrt{\sigma_t m_t}\dd{t} }=\min\sqrt{\Sigma_\tau\mca{M}_\tau},\label{eq:dis.BB.form.intro}
	\end{align}
	where $\mca{W}_1(p^A,p^B)$ is the discrete $L^1$-Wasserstein distance between two distributions $p^A$ and $p^B$ [cf.~Eq.~\eqref{eq:graph.Wass.dist}], and the minimum is over all admissible Markov jump processes that transform distribution $p^A$ into $p^B$ over a period $\tau$ with a given connectivity.
	The relation \eqref{eq:dis.BB.form.intro} provides a thermodynamic interpretation for the Wasserstein distance, implying that the Wasserstein distance is equal to the minimum product of the thermodynamic and kinetic costs.
	We also analogously generalize the formula \eqref{eq:dis.BB.form.intro} to the quantum case, in which the classical Wasserstein distance $\mca{W}_1(p^A,p^B)$ is replaced with a quantum Wasserstein distance $\mca{W}_q(\varrho^A,\varrho^B)$ between density matrices $\varrho^A$ and $\varrho^B$.
	These formulas can be considered as discrete generalizations of the Benamou--Brenier formula known in continuous cases.
	
	\item[(3)] \emph{Trade-off between irreversibility and state mobility}.---Through the developed variational formulas, we reveal a trade-off relation between the irreversibility and dynamical state mobility in discrete systems, which reads as follows:
		\begin{equation}\label{eq:tradeoff}
			\Sigma_\tau\mca{M}_\tau\ge\mca{W}_1(p^A,p^B)^2.
		\end{equation}
		The inequality \eqref{eq:tradeoff} implies that either the thermodynamic cost $\Sigma_\tau$ or kinetic cost $\mca{M}_\tau$ must be sacrificed (i.e., they cannot be simultaneously small) to evolve the system state.
	
	\item[(4)] \emph{Minimum dissipation and optimal protocol}.---The problem of minimizing entropy production in discrete systems is trivial if no constraints exist on the transition rates. Our results shed new light on this issue. More specifically, the formula \eqref{eq:dis.BB.form.intro} implies that fixing the dynamical state mobility $\mca{M}_\tau$ is a reasonable constraint from which minimum dissipation can be immediately determined through the Wasserstein distance, and the optimal control protocol can be constructed from the optimal transport problem. When additional constraints exist on system dynamics, we show that a lower bound on minimum dissipation can be obtained (see Fig.~\ref{fig:MinDiss} for illustration).
	
	\item[(5)] \emph{Thermodynamic speed limits}.---From the resulting variational formulas, we derive unified and stringent thermodynamic speed limits that place lower bounds on the time required for state transformation for both open classical and quantum systems. The classical bound reads [cf.~Eq.~\eqref{eq:speed.limit}]
	\begin{equation}\label{eq:cla.sl.intro}
		\tau\ge\frac{\mca{W}_1(p^A,p^B)}{\ev{\sqrt{\sigma m}}_\tau}\ge\frac{\mca{W}_1(p^A,p^B)}{\sqrt{\ev{\sigma}_\tau\ev{ m}_\tau}},
	\end{equation}
	where $\ev{x}_\tau$ denotes the time average of a time-dependent variable $x_t$. 
	The quantum bound has the same form, where $\mca{W}_1(p^A,p^B)$ is replaced with $\mca{W}_q(\varrho^A,\varrho^B)$.
	The inequality \eqref{eq:cla.sl.intro} implies that the speed of state transformation is constrained by the time average of the product of the thermodynamic and kinetic costs.
	Because we start from the equality relations, these bounds are tight and can always be saturated for {\it arbitrary} temperatures. 
	In other words, for an arbitrary pair of distributions or density matrices, there always exist Markovian dynamics that saturate the bounds.
	They are also stronger than previously known bounds \cite{Shiraishi.2018.PRL,Vu.2021.PRL}.
	
	\item[(6)] \emph{Finite-time Landauer principle}.---From the variational formulas, we derive finite-time lower bounds for heat dissipation $Q$ incurred in classical and quantum information erasure [cf.~Eqs.~\eqref{eq:cla.Landauer.prin} and \eqref{eq:qua.Landauer.prin}]. The bounds characterize both finite-time and finite-error effects on heat dissipation. 
	Several finite-time Landauer bounds have been derived for discrete systems in previous studies \cite{Zhen.2021.PRL,Vu.2022.PRL}. However, these bounds encounter the same problem as the conventional Landauer bound; that is, they lose the predictive power in the low-temperature regime. By contrast, our new bounds are tight for \emph{arbitrary} temperatures, even in the zero-temperature limit. A further simplified bound including a finite erasure error $\epsilon\ge 0$ reads [cf.~Eq.~\eqref{eq:fin.time.Landauer.error}]
	\begin{equation}
	Q\ge T[\ln d - h(\epsilon)] +\frac{(1-1/d-\epsilon)^2}{\tau\beta\ev{ m}_\tau},
	\end{equation}
	where $d$ is the system's dimension, $T$ is the temperature of the heat bath, $\beta$ is the inverse temperature, and $h(\epsilon)\ge 0$ is a function that vanishes as $\epsilon\to 0$.
	In the perfect-erasure ($\epsilon\to 0$) and quasistatic ($\tau\to \infty$) limits, the above bound reduces to the conventional Landauer bound $Q\ge T\ln d$. 
	
\end{itemize}

\subsection{Relevant literature}
Here, we briefly discuss several relevant studies that have attempted to link the (modified) Wasserstein distances to the thermodynamics of Markov jump processes.

The Benamou--Brenier formula has two facets. It not only provides a thermodynamic interpretation but also reveals a geometric structure for the continuous $L^2$-Wasserstein distance.
More specifically, $W_2$ can be interpreted as a Riemannian metric on the manifold of probability distribution functions.
Although the discrete $L^2$-Wasserstein distance is well defined and widely used in the literature, unfortunately, it does not possess a geometric interpretation, and its connection to thermodynamics also remains unclear.
For this reason, many studies have generalized the Wasserstein distance based on the geometric aspect of the Benamou--Brenier formula for discrete cases \cite{Maas.2011.JFA}. 
This modified Wasserstein distance places a lower bound on irreversible entropy production of Markov jump processes \cite{Vu.2021.PRL,Yoshimura.2022.arxiv}. However, it is system dependent because the transition rates are concretely used to define this distance.

In contrast to the previous direction, in this study, we consider the conventional discrete $L^1$-Wasserstein distance and focus on its thermodynamic interpretation. In this regard, Dechant has obtained some interesting results by relating the discrete $L^1$-Wasserstein distance to the entropy production and dynamical activity of Markov jump processes \cite{Dechant.2022.JPA}.
Here, we consider a different approach by introducing the dynamical state mobility and obtain discrete generalizations of the Benamou--Brenier formula. This approach not only unifies the classical discrete and continuous cases but also extends to the quantum case.
Although we focus on the discrete $L^1$-Wasserstein distance, it is noteworthy that the obtained generalizations of the Benamou--Brenier formula are similar to that for the continuous $L^2$-Wasserstein distance.
This suggests that the discrete $L^1$-Wasserstein distance may play the same role as the continuous $L^2$-Wasserstein distance in continuous cases.

The remainder of the paper is organized as follows. 
Section~\ref{sec:opt.tran.prob.cont.space} presents a review of the optimal transport problem and the relevant existing results in the context of continuous-state overdamped Langevin dynamics. We particularly emphasize the Benamou--Brenier formula of the $L^2$-Wasserstein distances and their connections to stochastic thermodynamics.
In Sec.~\ref{sec:sto.the.dis.sys}, we briefly introduce stochastic thermodynamics of classical Markovian dynamics. Next, we define the novel kinetic term $m_t$ and discuss its relevant properties. We then derive the improved thermodynamic uncertainty relation for Markov jump processes and numerically demonstrate it.
In Sec.~\ref{sec:opt.tran.dis.sys}, we describe the optimal transport problem in discrete cases and explain our first theorem that links the discrete Wasserstein distance to stochastic thermodynamics of Markov jump processes.
The relationship between the obtained and existing results in continuous cases is also discussed.
In Sec.~\ref{sec:quan.gen}, we define a quantum Wasserstein distance and explain our second theorem that generalizes the variational formula to the quantum case.
From the derived variational formulas, in Sec.~\ref{sec:appl}, we describe two applications: the thermodynamic speed limits and the finite-time Landauer principle.
In Sec.~\ref{sec:num.demon}, we numerically demonstrate our findings.
Finally, Sec.~\ref{sec:conc.disc} presents a conclusion with a discussion and outlook.
All detailed mathematical calculations and derivations can be found in the Appendixes.

\section{Review of continuous optimal transport}\label{sec:opt.tran.prob.cont.space}

In this section, we briefly review the optimal transport problem in continuous spaces and discuss the Benamou--Brenier formula, which provides a thermodynamic interpretation of the Wasserstein distances in the context of overdamped Langevin dynamics.

\subsection{Optimal transport problem}

First, we succinctly introduce the classical transport problem on the continuous space $\mbb{R}^d$ with $d\ge 1$ (see Ref.~\cite{Villani.2008} for details).
The problem of optimal transport---that is, determining how a pile of earth can be optimally transported into another pile of the same volume but with a different shape---was originally introduced by Monge.
The optimality here is interpreted to mean that the total transport cost is minimized with respect to a given cost metric.
Suppose that the source and target piles of earth are characterized by probability distribution functions $p^A(x)$ and $p^B(x)$, respectively, on the space $\mbb{R}^d$, and the cost metric is given by $c:\mbb{R}^d\times\mbb{R}^d\mapsto\mbb{R}_{\ge 0}$.
Then, the Monge optimal transport problem is to identify a one-to-one map $\varphi:\mbb{R}^d\to\mbb{R}^d$ that minimizes the objective function
\begin{equation}
\min_{\varphi}\int_{\mbb{R}^d}c(x,\varphi(x))p^A(x)\dd{x},
\end{equation}
where the minimum is over all $\varphi$ satisfying $p^A(x)=p^B(\varphi(x))|\det(\nabla\varphi(x))|$.
However, this formulation presents an issue regarding the non-existence of a valid transport map; that is, the map $\varphi$ might not exist in discrete cases because no mass can be split.
Fortunately, this issue was previously resolved by the relaxation of Kantorovich, which led to a more well-defined problem.
Instead of a transport map $\varphi(x)$, Kantorovich considered a transport plan $\pi(x,y)$ that is a joint probability distribution function and represents a coupling of $p^A$ and $p^B$.
This transport plan allows us to split a single mass and transport it to multiple target locations.
The Kantorovich problem can be formulated as an optimization of the following objective function:
\begin{equation}
\min_{\pi\in\Pi(p^A,p^B)}{\int_{\mbb{R}^d\times\mbb{R}^d}c(x,y)\pi(x,y)\dd{x}\dd{y}},
\end{equation}
where $\Pi(p^A,p^B)$ denotes the coupling set of joint probability distribution functions whose marginal distributions coincide with $p^A$ and $p^B$:
\begin{equation}
\int_{\mbb{R}^d}\pi(x,y)\dd{y}=p^A(x)~\text{and}~\int_{\mbb{R}^d}\pi(x,y)\dd{x}=p^B(y).
\end{equation}

The concept of optimal transport provides a means for defining useful metrics on continuous spaces of probability distribution functions.
By employing the cost metric of the Euclidean norm (i.e., $c(x,y)=\|x-y\|^\alpha$ for a positive number $\alpha\ge 1$), the Kantorovich problem reduces exactly to the $L^\alpha$-Wasserstein distance, which is defined as
\begin{equation}\label{eq:cWass.dist.def}
W_\alpha(p^A,p^B)^\alpha\coloneqq\min_{\pi\in\Pi(p^A,p^B)}{\int_{\mbb{R}^d\times\mbb{R}^d}\|x-y\|^\alpha\pi(x,y)\dd{x}\dd{y}}.
\end{equation}
The $L^\alpha$-Wasserstein distance is a genuine metric and satisfies the triangle inequality.
Applying H{\"o}lder's inequality, we can derive a hierarchical relationship, that is, $W_\alpha\le W_{\alpha'}$ for $\alpha\le\alpha'$.
Of the several that exist, the $L^1$- and $L^2$-Wasserstein distances are particularly relevant from the thermodynamic and geometric perspectives.
In the following, we discuss some remarkable properties of these two distances.

\subsection{Benamou--Brenier formula}
The Wasserstein distance can be expressed in a variational form in several ways.
Interestingly, Benamou and Brenier developed a variational formula for the $L^2$-Wasserstein distance in terms of fluid mechanics \cite{Benamou.2000.NM}.
The Benamou--Brenier formula casts the $L^2$-Wasserstein distance as a minimization problem of a time-integrated cost in terms of probability distribution functions and velocity fields:
\begin{equation}\label{eq:BB.form}
W_2(p^A,p^B)^2=\min_{v_t}{\tau\int_0^\tau\int_{\mbb{R}^d}\|v_t(x)\|^2p_t(x)\dd{x}\dd{t}},
\end{equation}
where the minimum is over all smooth paths $\{v_t\}_{0\le t\le\tau}$ subject to the continuity equation 
\begin{equation}\label{eq:cont.eq}
\dot p_t(x)+\nabla\vdot[v_t(x)p_t(x)]=0
\end{equation}
with the initial and final conditions $p_0(x)=p^A(x)$ and $p_\tau(x)=p^B(x)$, respectively.
Here, $\nabla$ is the del operator, and $\vdot$ denotes the standard Euclidean inner product between vectors.
Note that given any absolutely continuous curve $\{p_t\}$, we can always find a velocity field $\{v_t\}$ that satisfies Eq.~\eqref{eq:cont.eq}.
The formulation \eqref{eq:BB.form} not only enables us to find a numerical scheme for computing $W_2$ but also provides the thermodynamic and geometric interpretations of the $L^2$-Wasserstein distance.

Next, we discuss a thermodynamic interpretation of the $L^2$-Wasserstein distance (see Appendix \ref{app:W2.geo.prop} for a geometric interpretation). Consider an overdamped system on the continuous space $\mbb{R}^d$, which is constantly subject to a time-dependent force $F_t(x)$ and weakly coupled to a single heat bath.
The system state at time $t$ can be characterized by the probability distribution $p_t(x)$, the time evolution of which is described by the Fokker--Planck equation:
\begin{align}
\dot p_t(x)&=-\nabla\vdot[v_t(x)p_t(x)],\\
v_t(x)&=F_t(x)-D\nabla\ln p_t(x).\label{eq:velo.def}
\end{align}
Note that the velocity field $v_t(x)$ of the system in Eq.~\eqref{eq:velo.def} is an exact solution to the continuity equation~\eqref{eq:cont.eq}, which drives the source distribution $p_0(x)$ to the target distribution $p_\tau(x)$.
Based on the framework of stochastic thermodynamics, the irreversible entropy production during period $\tau$ can be calculated as \cite{Seifert.2012.RPP}
\begin{equation}
\Sigma_\tau=\frac{1}{D}\int_0^\tau\int_{\mbb{R}^d}\|v_t(x)\|^2p_t(x)\dd{x}\dd{t}.
\end{equation}
Irreversible entropy production clearly coincides with the time-integrated cost in the integration in Eq.~\eqref{eq:BB.form}, ignoring the scaling factor.
Therefore, we can rewrite the Benamou--Brenier formula as
\begin{equation}\label{eq:BB.form2}
W_2(p^A,p^B)=\min_{v_t}\sqrt{D\tau\Sigma_\tau}.
\end{equation}
From this, the following inequality can be immediately derived:
\begin{equation}\label{eq:Wass.bound.cont}
\Sigma_\tau\ge\frac{W_2(p_0,p_\tau)^2}{D\tau}.
\end{equation}
Inequality \eqref{eq:Wass.bound.cont} refines the second law of thermodynamics by providing a stronger bound on irreversible entropy production solely in terms of the initial and final distributions, given that the operational time and diffusion coefficient are fixed.
The bound can be interpreted as a thermodynamic speed limit:
\begin{equation}\label{eq:W2.csl}
\tau\ge\frac{W_2(p^A,p^B)}{\sqrt{D\ev{\sigma}_\tau}}.
\end{equation}
Moreover, it can also be applied to derive a finite-time Landauer principle \cite{Aurell.2012.JSP,Proesmans.2020.PRL}.
It is noteworthy that the bound can be saturated for any pair of the initial and final distributions and is tight even in the zero-temperature limit.

Because $v_t(x)$ and $F_t(x)$ can be considered a one-to-one correspondence, we can obtain the following equality between irreversible entropy production and the Wasserstein distance:
\begin{equation}\label{eq:Wass.ent.eq.cont}
\min_{F_t}\Sigma_\tau=\frac{W_2(p_0,p_\tau)^2}{D\tau}.
\end{equation}
This relation implies that the minimum entropy production in all overdamped processes that transform one distribution into another can be determined exactly by the Wasserstein distance between the two distributions.
In Refs.~\cite{Benamou.2000.NM,Aurell.2012.JSP}, it has been demonstrated that the minimum in Eq.~\eqref{eq:BB.form} can be achieved with a velocity field of the form $v_t(x)=-\nabla\phi_t(x)$, where $\phi_t(x)$ is a time-dependent potential.
Thus, the minimum entropy production can always be achieved with a conservative force $F_t(x)=-\nabla V_t(x)$, where $V_t(x)=\phi_t(x)-D\ln p_t(x)$ is a time-dependent potential.

\section{Stochastic thermodynamics of discrete systems}\label{sec:sto.the.dis.sys}
In this section, we first briefly introduce the stochastic thermodynamics of classical discrete Markovian dynamics described by the master equation; for a comprehensive review, one can refer to Ref.~\cite{Seifert.2012.RPP}.
We then define a novel physical quantity called dynamical state mobility, discuss its relevant properties, and derive an improved thermodynamic uncertainty relation.

\subsection{Markov jump processes}

We consider a discrete-state system with $N>1$ states, which is weakly attached to single or multiple thermal reservoirs.
Examples of these systems include diffusive processes on a lattice, biomolecular motors, chemical reaction networks, and quantum dots.
The system can be described in terms of a time-dependent probability distribution ${p_t}\coloneqq[p_1(t),\dots,p_N(t)]^\top$, where $p_x(t)$ denotes the probability of finding the system in state $x$ at time $t$.
Assume that the system is modeled by a time-continuous Markov jump process and that the transitions from a state $y$ to a state $x$ occur at a nonnegative rate $w_{xy}(t)$, which can be time dependent according to an external control protocol. 
The time evolution of the probability distribution is described by the master equation:
\begin{equation}\label{eq:Markov.mas.eq}
{\dot p_t}=\msf{W}_t{p_t},
\end{equation}
where dot $\cdot$ denotes the time derivative and $\msf{W}_t=[w_{xy}(t)]$ denotes the matrix of the transition rates with $w_{xx}(t)=-\sum_{y(\neq x)}w_{yx}(t)$.
We consider microscopically reversible dynamics, that is, $w_{xy}(t)>0$ if and only if $w_{yx}(t)>0$.
Hereafter, we assume that the transition rates satisfy the {\it local} detailed balance condition \cite{Seifert.2012.RPP}:
\begin{equation}\label{eq:ldbc}
\ln\frac{w_{xy}(t)}{w_{yx}(t)}=s_{xy}(t),
\end{equation}
where $s_{xy}(t)$ denotes the entropy change in the environment due to the jump from state $y$ to $x$ at time $t$.
If we fix the transition rates at any time, the system relaxes toward a stationary state, which may no longer be an equilibrium state.

In a case wherein the system is attached to a single reservoir at inverse temperature $\beta$ and the transitions between states are induced by the energy difference, the entropy change reads 
\begin{equation}\label{eq:gdbc}
s_{xy}(t)=\beta[\varepsilon_y(t)-\varepsilon_x(t)],
\end{equation} 
where $\varepsilon_x(t)$ denotes the instantaneous energy level of state $x$ at time $t$.
Whenever this occurs, we say that the system satisfies the {\it global} detailed balance condition.
Notably, the thermal state $\peq_x(t)\propto e^{-\beta\varepsilon_x(t)}$ becomes the instantaneous stationary state of the system (i.e., $\msf{W}_t{\peq_t}=0$).
Hereafter, we consider generic dynamics satisfying the local detailed balance condition.
However, dynamics satisfying Eq.~\eqref{eq:gdbc} will be used occasionally for physical interpretation of some quantities.

For convenience, we define the following quantities:
\begin{align}
a_{xy}(t)&\coloneqq w_{xy}(t)p_y(t),\\
j_{xy}(t)&\coloneqq w_{xy}(t)p_y(t)-w_{yx}(t)p_x(t),
\end{align}
which quantify the frequency of jumps and the probability current from state $y$ to $x$ at time $t$, respectively.

\subsection{Entropy production and dynamical activity}

Given the previous setup, we now discuss some relevant thermodynamic quantities.
One central quantity is irreversible entropy production, which quantifies the degree of irreversibility of the thermodynamic process.
Within the framework of stochastic thermodynamics, total entropy production during a period $\tau$ can be defined as
\begin{equation}
\Sigma_\tau\coloneqq\Delta S_{\rm sys}+\Delta S_{\rm env},
\end{equation}
where $\Delta S_{\rm sys}$ and $\Delta S_{\rm env}$ are the changes in the entropy of the system and the environment, respectively, expressed as
\begin{align}
\Delta S_{\rm sys}&=S(p_\tau)-S(p_0),\\
\Delta S_{\rm env}&=\int_0^\tau\sum_{x\neq y} a_{xy}(t) s_{xy}(t)\dd{t}.
\end{align}
Here, system entropy production is quantified via the Shannon entropy $S(p)\coloneqq-\sum_xp_x\ln p_x$, whereas the entropy change of the environment is defined as the sum of entropic contributions from each transition between states.
Simple calculations show that the entropy production rate is always nonnegative:
\begin{equation}\label{eq:ent.prod.def}
\sigma_t\coloneqq\dot\Sigma_t=\sum_{x>y}[a_{xy}(t)-a_{yx}(t)]\ln\frac{a_{xy}(t)}{a_{yx}(t)}\ge 0.
\end{equation}
The non-negativity of irreversible entropy production corresponds to the second law of thermodynamics.
The equality of this zero bound is attained only when the system is in the instantaneous thermal state at all times.

Another essential quantity in nonequilibrium thermodynamics is dynamical activity, quantified by the amplitude of transitions between states as
\begin{equation}
a_t\coloneqq\sum_{x\neq y}a_{xy}(t).
\end{equation}
The average number of jumps during period $\tau$ can be calculated as
\begin{equation}
\mca{A}_\tau\coloneqq\int_0^\tau a_t\dd{t}.
\end{equation}
The time average of dynamical activity characterizes the timescale of thermodynamic processes.
The higher the dynamical activity, the stronger the thermalization.
Entropy production and dynamical activity are the time-antisymmetric and time-symmetric parts, respectively, of the path-integral action with respect to a time-reversed process \cite{Maes.2020.PR}.
Both quantities constrain the fluctuation of currents according to the thermodynamic and kinetic uncertainty relations \cite{Gingrich.2016.PRL,Terlizzi.2019.JPA}.

\subsection{Dynamical state mobility}
\subsubsection{Definition}
Here, we introduce a new quantity called dynamical state mobility, which plays a crucial role in our results.
Before getting into the details, let us briefly recall the linear response relations \cite{Onsager.1931.PR.I,Onsager.1931.PR.II}, which express the equalities between currents and forces in near-equilibrium systems.
Consider an irreversible transport process driven by thermodynamic forces $F=[F_x]^\top$, such as affinities in temperatures or chemical potentials.
Let $J=[J_x]^\top$ be the thermodynamic currents that characterize the response of the system to the applied forces.
In a linear-response regime, the currents depend only on the thermodynamic forces and can be expressed by the following linear relations:
\begin{equation}\label{eq:lin.res.eq}
J_x=\sum_y \mu_{xy} F_y~\text{or}~{J}=\msf{L}{F}.
\end{equation}
These relations \eqref{eq:lin.res.eq} are referred to as linear response equations, where the coefficients $\mu_{xy}$ are known as Onsager kinetic coefficients, and $\msf{L}=[\mu_{xy}]$ is called the Onsager matrix.
Onsager reciprocal relations imply that in the case of time-reversal symmetry, the Onsager matrix is symmetric (i.e., $\mu_{xy}=\mu_{yx}$). 
In addition, the entropy production rate can be expressed in a quadratic form of the forces as 
\begin{equation}\label{eq:ent.lin.res}
\sigma=\sum_x J_x F_x=\sum_{x,y} \mu_{xy}F_x F_y~\text{or}~\sigma=F^\top\msf{L}{F} .
\end{equation}
The non-negativity of the entropy production rate immediately derives that $\msf{L}$ is positive semi-definite.

The Onsager coefficients characterize the response of dynamics close to equilibrium at the {\it macroscopic} level.
Nevertheless, they can be mimicked to dynamics far from equilibrium at the {\it microscopic} level.
To show this, let us focus on local transitions between states.
The generalized thermodynamic force associated with each transition from $y$ to $x$ is defined as \cite{Gingrich.2016.PRL}
\begin{equation}
f_{xy}(t)\coloneqq\ln\frac{a_{xy}(t)}{a_{yx}(t)},
\end{equation}
which is the sum of the entropy changes in the system and environment derived from the jump.
Since $j_{xy}(t)$ is the probability current associated with the transition from $y$ to $x$, defining the following coefficient is logical:
\begin{equation}\label{eq:lmn.def}
m_{xy}(t)\coloneqq\frac{a_{xy}(t)-a_{yx}(t)}{\ln a_{xy}(t)-\ln a_{yx}(t)}=\frac{j_{xy}(t)}{f_{xy}(t)}.
\end{equation}
Intuitively, $\{m_{xy}(t)\}$ characterize the responses of the probability currents against the thermodynamic forces at the transition level.
Notice that $\{m_{xy}(t)\}$ are always nonnegative and symmetric (i.e., $m_{xy}(t)=m_{yx}(t)\ge 0$).
Equation \eqref{eq:lmn.def} can also be rewritten in the form of
\begin{equation}\label{eq:current.force}
j_{xy}(t)=m_{xy}(t)f_{xy}(t),
\end{equation}
which shows that the currents and thermodynamic forces can be linearly related in terms of the coefficients $\{m_{xy}(t)\}$.
Moreover, the entropy production rate can be expressed in a quadratic form of the thermodynamic forces $\{f_{xy}(t)\}$ as
\begin{equation}\label{eq:ent.force}
\sigma_t=\sum_{x>y}m_{xy}(t)f_{xy}(t)^2=\sum_{x>y}\sigma_{xy}(t).
\end{equation}
Here, we define the entropy production rate associated with each transition as $\sigma_{xy}(t)\coloneqq m_{xy}(t)f_{xy}(t)^2$.
Equations ~\eqref{eq:current.force} and \eqref{eq:ent.force} have the same algebraic forms as Eqs.~\eqref{eq:lin.res.eq} and \eqref{eq:ent.lin.res}, respectively, which suggests that the coefficients $\{m_{xy}(t)\}$ play similar roles with the Onsager coefficients for far-from-equilibrium systems.

In a weak-thermodynamic-force limit (i.e., $|f_{xy}(t)|\to 0$), the coefficient $m_{xy}(t)$ reduces to the average dynamical activity between states $x$ and $y$:
\begin{equation}\label{eq:einstein.like.rel}
m_{xy}(t)\to\frac{a_{xy}(t)+a_{yx}(t)}{2}.
\end{equation}
This is somewhat analogous to the Einstein relation on mobility in overdamped Langevin dynamics (see Table \ref{table:mobility.analog}). Note that the weak-thermodynamic-force limit is defined at the microscopic state level and can be achieved via two routes: the completely equilibrium limit in discrete systems and the continuous state limit (e.g., the limit from the discrete hopping particle system to the overdamped Fokker--Planck equation). In the continuous state limit, the difference between neighboring states $x$ and $y$ is infinitesimal, and thus, the force $f_{xy}(t)$ is also infinitesimal because $a_{xy}(t)/a_{yx}(t)\simeq 1$. We discuss the continuous state limit in the following subsection and show that the right-hand side in Eq.~\eqref{eq:einstein.like.rel} is proportional to the diffusion coefficient. In general, the following relation holds for the coefficient $m_{xy}(t)$:
\begin{align}
\sqrt{a_{xy}(t)a_{yx}(t)}\le m_{xy}(t)\le\frac{a_{xy}(t)+a_{yx}(t)}{2}.\label{eq:rel.lmn.amn}
\end{align}
It is thus natural to define the sum of $\{ m_{xy}(t)\}$ over all transitions:
\begin{equation}\label{eq:lt.def}
m_t\coloneqq\sum_{x>y}m_{xy}(t).
\end{equation}
For convenience, we refer to this term as dynamical state mobility throughout this paper.
This nomenclature derives from the analogy between microscopic coefficients $\{m_{xy}(t)\}$ and the macroscopic mobility (see Table \ref{table:mobility.analog}).

\begin{table}[t]
\centering
\caption{Analogy between the dynamical state mobility and macroscopic mobility.}
\label{table:mobility.analog}
\begin{tabular}{c|c}
\hline\hline
Microscopic level & Macroscopic level \\ [0.4ex]
\hline
$j_{xy}=m_{xy}f_{xy}$ & $J=\mu F$\\[0.4ex]
\hline
Einstein-like relation $|f_{xy}|\ll 1$ & Einstein relation $|F|\ll 1$\\
$m_{xy}=(a_{xy}+a_{yx})/2$ & $\mu=\beta D$\\[0.4ex]
\hline\hline
\end{tabular}
\end{table}

To clarify further the identity of $m_t$, let us consider a case in which the system is attached to a single reservoir and satisfies the global detailed balance condition.
In this case, the master equation \eqref{eq:Markov.mas.eq} can be written as \cite{Vu.2021.PRL}
\begin{equation}\label{eq:pt.Kf}
{\dot p_t}=\msf{K}_t{f_t},
\end{equation}
where $\msf{K}_t$ is a symmetric, positive semi-definite matrix given by
\begin{equation}
\msf{K}_t\coloneqq\frac{1}{2}\sum_{x\neq y}m_{xy}(t)\msf{E}_{xy}.
\end{equation}
Here, ${f_t}=[f_1(t),...,f_N(t)]^\top$ with $f_x(t)=-\ln p_x(t)+\ln\peq_x(t)$ and $\msf{E}_{xy}=[e_{uv}]\in\mbb{R}^{N\times N}$ is a matrix with $e_{yy}=e_{xx}=1$, $e_{xy}=e_{yx}=-1$, and zeros in all other elements.
The quantities $\{f_x(t)\}$ are identified as the entropic thermodynamic forces, which characterize how far the system is driven from the instantaneous equilibrium state.
Equation \eqref{eq:pt.Kf} represents the linear relations between the rates ${\dot{p}_t}$ and forces ${f_t}$ through the symmetric matrix $\msf{K}_t$.
Furthermore, the total entropy production rate can be written in a quadratic form as \cite{Vu.2021.PRL} 
\begin{equation}\label{eq:ent.Kf}
\sigma_t=f_t^\top{\msf{K}_t}{f_t}.
\end{equation}
Therefore, Eqs.~\eqref{eq:pt.Kf} and \eqref{eq:ent.Kf} can be viewed as far-from-equilibrium counterparts of Eqs.~\eqref{eq:lin.res.eq} and \eqref{eq:ent.lin.res}, respectively.
The matrix $\msf{K}_t$ is thus identified as the Onsager-like matrix.
Because $\{m_{xy}(t)\}$ are elements of $\msf{K}_t$, they can be regarded as the Onsager-like kinetic coefficients for out-of-equilibrium systems.
From the definition of $\msf{K}_t$, we can easily verify that $m_t$ is exactly the sum of diagonal elements of the Onsager-like matrix:
\begin{equation}
 m_t=\frac{1}{2}\tr{\msf{K}_t}.
\end{equation}
Therefore, $m_t$ can be identified as a kinetic term.
The time integral of the kinetic term $m_t$ can be considered as the kinetic cost of Markov jump processes, defined by
\begin{equation}
\mca{M}_\tau\coloneqq\int_0^\tau  m_t\dd{t}=\tau\ev{m}_\tau,
\end{equation}
where we define the time-averaged quantity for arbitrary time-dependent quantity $x_t$ as
\begin{equation}
\ev{x}_\tau\coloneqq \tau^{-1}\int_0^\tau x_t\dd{t}.
\end{equation}

By summing both sides of Eq.~\eqref{eq:rel.lmn.amn} for all $x>y$, we can prove that dynamical state mobility is upper bounded by dynamical activity as
\begin{equation}\label{eq:lt.ub}
 m_t\le\frac{a_t}{2}.
\end{equation}
The equality of Eq.~\eqref{eq:lt.ub} can be achieved in the equilibrium limit.
It is thus evident that $\mca{M}_\tau\le\mca{A}_\tau/2$.

\subsubsection{Continuous limit}
Next, we investigate $m_t$ in the continuous limit. 
To this end, we consider an overdamped Brownian particle trapped in a one-dimensional potential $V_t(x)$.
Let $x_t$ denote the position of the particle at time $t$. Then, its dynamics is governed by the Langevin equation:
\begin{equation}
\dot{x}_t=F_t(x_t)+\sqrt{2D}\xi_t,
\end{equation}
where $F_t(x)\coloneqq-\partial_xV_t(x)$ is the total force applied on the particle, $\xi_t$ is a zero-mean Gaussian white noise with variance $\ev{\xi_t\xi_{t'}}=\delta(t-t')$, and $D>0$ is the diffusion coefficient.
Let $p_t(x)$ be the probability distribution function of finding the particle in state $x$ at time $t$. Then, its time evolution can be described by the Fokker--Planck equation:
\begin{equation}\label{eq:FP}
\dot p_t(x)=-\partial_x\qty[F_t(x)p_t(x)-D\partial_xp_t(x)],
\end{equation}
where we set $\mu=1$ for simplicity.
We now consider the discretization of the Fokker--Planck equation \eqref{eq:FP} with space interval $\Delta x>0$ and define $x_n\coloneqq n\Delta x$.
By defining the probability distribution and transition rates as
\begin{align}
p_n(t)&\coloneqq p_t(x_n)\Delta x,\\
w_{n+1,n}(t)&\coloneqq\frac{D}{(\Delta x)^2}\exp\qty[\frac{V_t(x_n)-V_t(x_{n+1})}{2D}]\notag\\
&\simeq\frac{F_t(x_n)}{2\Delta x}+\frac{D}{(\Delta x)^2},\label{eq:1d.dis.tran.rate1}\\
w_{n-1,n}(t)&\coloneqq\frac{D}{(\Delta x)^2}\exp\qty[\frac{V_t(x_{n})-V_t(x_{n-1})}{2D}]\notag\\
&\simeq\frac{-F_t(x_n)}{2\Delta x}+\frac{D}{(\Delta x)^2},\label{eq:1d.dis.tran.rate2}
\end{align}
we readily obtain the master equation:
\begin{align}
\dot p_n(t)&=w_{n,n-1}(t)p_{n-1}(t)+w_{n,n+1}(t)p_{n+1}(t)\notag\\
&-[w_{n+1,n}(t)+w_{n-1,n}(t)]p_n(t).
\end{align}
We note here from Eqs.~\eqref{eq:1d.dis.tran.rate1} and \eqref{eq:1d.dis.tran.rate2} that
\begin{align}
a_{n+1,n}(t) &\simeq \qty[ \frac{F_t(x_n)}{2\Delta x}+\frac{D}{(\Delta x)^2 }] p_n(t), \\
{a_{n+1,n}(t)+a_{n,n+1}(t) \over 2} &= Dp_t(x_n)(\Delta x)^{-1} + O(1).
\end{align}
The probability currents and thermodynamic forces can be calculated as
\begin{align}
j_{n+1,n}(t)&=F_t(x_n)p_t(x_n)-D\partial_xp_t(x_n)+{O}(\Delta x),\\
f_{n+1,n}(t)&=\qty[ \frac{F_t(x_n)}{D}-\frac{\partial_xp_t(x_n)}{p_t(x_n)} ]\Delta x+{O}(\Delta x^2).
\end{align}
From these expressions, $m_{n+1,n}(t)$ can be calculated via the definition \eqref{eq:lmn.def}, which gives
\begin{equation}\label{eq:lmn.cont}
m_{n+1,n}(t)\to \frac{a_{n+1,n}(t)+a_{n,n+1}(t)}{2}=Dp_t(x_n)(\Delta x)^{-1}+{O}(1).
\end{equation}
Equation \eqref{eq:lmn.cont} indicates that $m_{n+1,n}(t)$ converges to the value $[a_{n+1,n}(t)+a_{n,n+1}(t)]/2$, and these correspond to the product of the diffusion coefficient and probability distribution (see Table \ref{table:mobility.analog}).
Summing both sides of Eq.~\eqref{eq:lmn.cont} for all $n$, we obtain
\begin{equation}
 m_t\to\frac{a_t}{2}\simeq\sum_nDp_t(x_n)(\Delta x)^{-1}=D(\Delta x)^{-2}.
\end{equation}
This implies that $m_t$ is proportional to $D$ as the scaling factor is ignored. Thus, $m_t$ should play the same role as that of the diffusion coefficient.
It is noteworthy that the diffusion coefficient is exactly the diagonal Onsager coefficient of overdamped Langevin processes in the linear-response regime.

\subsection{Thermodynamic uncertainty relation:\\ State mobility is crucial in nonequilibrium}
Here, we describe an improved thermodynamic uncertainty relation, showing that the kinetic cost of dynamical state mobility plays a critical role in constraining the precision of time-integrated currents.
For simplicity, we consider a steady-state Markov jump process.
The generalization to the case of an arbitrary initial state and arbitrary time-dependent driving is straightforward.

Let $\Gamma=\qty{x_0,(t_1,x_1),\dots,(t_K,x_K)}$ be a stochastic trajectory, in which the system is initially at state $x_0$ and subsequently jumps from state $x_{k-1}$ to $x_k$ at time $t_k$ for each $k=1,\dots,K$.
For each stochastic trajectory $\Gamma$, we consider a time-antisymmetric current $J$, defined as
\begin{equation}
J(\Gamma)\coloneqq\sum_{k=1}^{K}\Upsilon_{x_{k}x_{k-1}}.
\end{equation}
Here, $\{\Upsilon_{xy}\}$ are arbitrary real coefficients satisfying $\Upsilon_{xy}=-\Upsilon_{yx}$ for all $x$ and $y$.
Examples of currents include the entropy flux and heat flux to the environment by specific choices of $\{\Upsilon_{xy}\}$.
The precision of current $J$ can be quantified by the square of the current mean divided by its variance $\ev{J}^2/\Var{J}$.
The conventional thermodynamic uncertainty relation \cite{Barato.2015.PRL,Gingrich.2016.PRL} sets an upper bound on the precision in terms of the total entropy production, given by the following inequality:
\begin{equation}\label{eq:org.TUR}
\frac{\ev{J}^2}{\Var{J}}\le\frac{\Sigma_\tau}{2}.
\end{equation}
Numerous studies have generalized this relation to other dynamics, from classical to quantum \cite{Proesmans.2017.EPL,Brandner.2018.PRL,Hasegawa.2019.PRE,Vu.2019.PRE.UnderdampedTUR,Hasegawa.2019.PRL,Timpanaro.2019.PRL,Guarnieri.2019.PRR,Carollo.2019.PRL,Dechant.2020.PNAS,Hasegawa.2020.PRL,Vu.2020.PRR,Liu.2020.PRL,Koyuk.2020.PRL,Wolpert.2020.PRL,Miller.2021.PRL.TUR,Hasegawa.2021.PRL,Lee.2021.PRE.TUR,Pal.2021.PRR,Vu.2022.PRL.TUR}.

We improve the thermodynamic uncertainty relation by proving that the precision of currents is upper bounded by the product of the thermodynamic and kinetic costs, as follows:
\begin{equation}\label{eq:new.TUR}
\frac{\ev{J}^2}{\Var{J}}\le\eta\frac{\Sigma_\tau}{2},
\end{equation}
where $\eta=2\mca{M}_\tau/\mca{A}_\tau\le 1$ can be regarded as an efficiency of dynamical activity (see Appendix \ref{app:proof.TUR} for the proof).
The new relation \eqref{eq:new.TUR} is tighter than the conventional relation \eqref{eq:org.TUR} and can be saturated in the case of a one-dimensional random walk.
Although the conventional relation \eqref{eq:org.TUR} implies that increasing dissipation is necessary to achieve high precision, it does not ensure the converse; that is, increasing dissipation is not sufficient for obtaining high precision of currents.
This can be explained through our relation, where the kinetic contribution $\eta$ appears in the bound in addition to the thermodynamic contribution.
For systems far from equilibrium, it is tedious that $\eta\ll 1$, which equivalently indicates the unattainability of the conventional bound.

For a case in which the system is in an arbitrary initial state and is driven by a time-dependent protocol, the derived relation can be analogously generalized as
\begin{equation}
\frac{[(\tau\partial_\tau-v\partial_v)\ev{J}]^2}{\Var{J}}\le\eta\frac{\Sigma_\tau}{2},
\end{equation}
where $v$ is a speed parameter of the control protocol \cite{Koyuk.2020.PRL}.
\begin{figure}[!]
\centering
\includegraphics[width=1.0\linewidth]{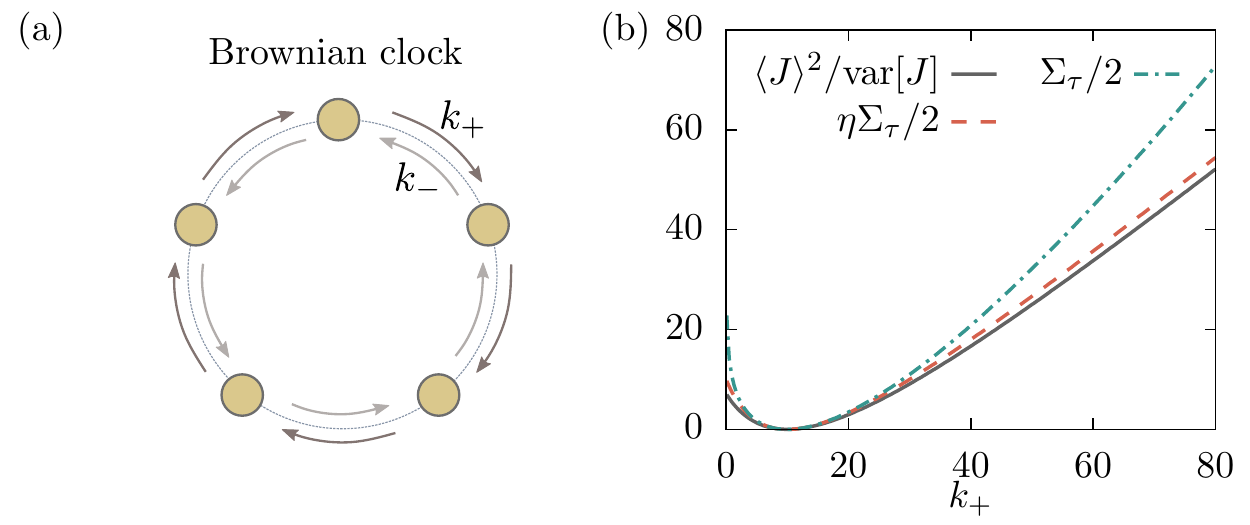}
\protect\caption{Numerical illustration of the thermodynamic uncertainty relations. (a) Schematic of the five-state Brownian clock and (b) numerical verification. The current precision $\ev{J}^2/\Var{J}$, new bound $\eta\Sigma_\tau/2$, and conventional bound $\Sigma_\tau/2$ are indicated by the solid, dashed, and dash-dotted lines, respectively. The forward rate $k_+$ is varied, whereas the backward rate $k_-$ is fixed at $k_-=10$. The operational time is fixed at $\tau=1$.}\label{fig:TUREx}
\end{figure}

In the following, we exemplify the derived thermodynamic uncertainty relation in a five-state Brownian clock \cite{Barato.2016.PRX}.
The Brownian clock is modeled as an inhomogeneous biased random walk on a ring with five states [see Fig.~\ref{fig:TUREx}(a)].
The clock's pointer transits from state $x$ to $x+1~(6\equiv 1)$ at rate $k_+>0$, whereas the backward rate is $k_->0$.
The net number of cycles completed by the pointer characterizes the clock's time.
In other words, time can be counted by a stochastic current $J$ that increases by $1$ for each transition from state $5$ to $1$ and decreases by $1$ for the reverse transition from state $1$ to $5$.
The stochastic current $J$ can be defined by setting $\Upsilon_{15}=1=-\Upsilon_{51}$ and $\Upsilon_{xy}=0$ for others.
Thus, the precision of the clock can be quantified by $\ev{J}^2/\Var{J}$.

We consider the clock operating in the stationary state.
To investigate the quality of the bounds, we fix the backward rate $k_-=10$ and vary the forward rate $k_+\in(0,80]$.
For each parameter setting, we calculate the precision of the clock and the bounds of the conventional and new thermodynamic uncertainty relations using the full counting statistics.
The numerical results are plotted in Fig.~\ref{fig:TUREx}(b), which verify that the new bound is always tighter than the conventional bound.
In particular, the new bound is tight even in the far-from-equilibrium regime.

\section{Results on discrete optimal transport}\label{sec:opt.tran.dis.sys}
Thus far, the problem of optimal transport has been discussed in terms of continuous spaces. In the following, we focus on the case of discrete spaces and explain the discrete Wasserstein distance. We then state our first theorem, which connects the discrete Wasserstein distance to stochastic thermodynamics of Markov jump processes.

\subsection{Optimal transport distance}
The optimal transport problem in the discrete case is analogous to that in the continuous case, except that we now deal with discrete $N$-dimensional distributions.
Given two discrete distributions $p^A=[p_x^A]$ and $p^B=[p_x^B]$, the optimal means of transporting distribution $p^A$ to $p^B$ with respect to a cost matrix $C=[c_{xy}]$ becomes the focus.
Here, $c_{xy}\ge 0$ denotes the cost of transporting a unit probability from $p_y^A$ to $p_x^B$.

The transport problem can be formulated using a coupling $\pi=[\pi_{xy}]$ between the probability distributions $p^A$ and $p^B$.
Specifically, $\pi$ is a joint probability distribution such that its marginal distributions coincide with $p^A$ and $p^B$ (i.e., the following conditions are satisfied for all $x$):
\begin{equation}
p_x^A=\sum_{y=1}^N \pi_{yx}~\text{and}~p_x^B=\sum_{y=1}^N\pi_{xy}.
\end{equation}
Each coupling thus defines a transport plan: for each $x$ and $y$, we transport an amount $\pi_{xy}$ from $p_y^A$ to $p_x^B$.
Thus, the discrete $L^1$-Wasserstein distance can be defined as the minimum transport cost over all admissible couplings:
\begin{equation}
W_1(p^A,p^B)\coloneqq\min_{\pi\in\Pi(p^A, p^B)}\sum_{x,y}c_{xy}\pi_{xy},
\end{equation}
where $\Pi(p^A, p^B)$ denotes a set of couplings between $p^A$ and $p^B$.
Once the cost matrix is provided, the discrete Wasserstein distance can be efficiently computed using a linear programming method.
In addition, as long as the cost matrix satisfies
\begin{equation}
c_{xy}+c_{yz}\ge c_{xz}
\end{equation}
for any $x$, $y$, and $z$, the resulting distance fulfills the triangle inequality:
\begin{equation}
W_1(p^A, p^B)+W_1(p^B, p^C)\ge W_1(p^A, p^C).
\end{equation}
We observe that the definition of the Wasserstein distance depends on the cost matrix. In other words, each matrix of transport costs induces a quantitatively different measure of distance.
Evidently, an infinite number of approaches can be used to choose the cost matrix.
In the following, we consider the cost matrix and corresponding Wasserstein distance defined on the basis of a graph.

Let $\mca{G}(V,E)$ denote an undirected graph, where $V$ and $E$ are the sets of vertices and unordered edges, respectively. Then, any microscopically reversible Markov jump process can be associated with an undirected graph, in which $V=\{1,\dots,N\}$ is the set of all states of the Markov jump process, and two vertices $x$ and $y$ are connected by an edge $(x,y)\in E$ if the transition between $x$ and $y$ is allowed. A jump process that has a unique steady state can be described by a connected graph; that is, for any unordered pair $(x,y)$, a sequence of vertices $P=[v_1,\dots,v_k]$ always exists such that $x=v_1$, $y=v_k$, and $(v_i,v_{i+1})\in E$ for all $1\le i\le k-1$. A subgraph $\tilde{\mca{G}}$ of a graph $\mca{G}$ is one whose edge set is a subset of that of $\mca{G}$. In other words, $\tilde{\mca{G}}$ can be obtained from $\mca{G}$ by removing some edges. This is equivalent to setting some transition rates of the Markov jump process to zero. For convenience, hereafter, the underlying graph structure of a Markov jump process is referred to as its {\it topology}.

\begin{figure}[t]
\centering
\includegraphics[width=1.0\linewidth]{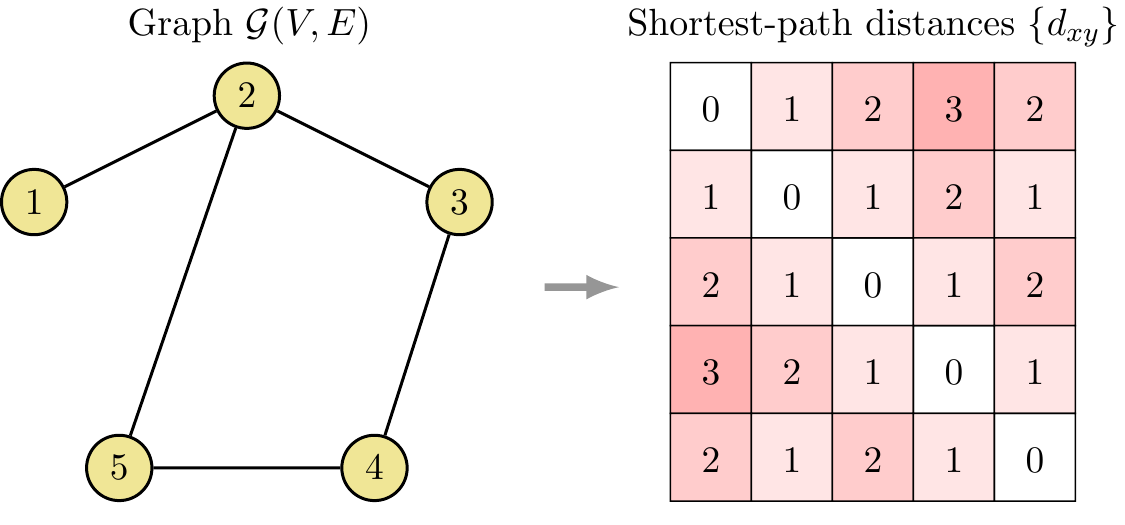}
\protect\caption{Example of the Wasserstein distance defined based on a graph consisting of $5$ vertices and $5$ edges. Given the topology $\mca{G}(V,E)$, the shortest-path distances $\{d_{xy}\}$ can be calculated, from which the Wasserstein distance can be defined.}\label{fig:WassDistIllust}
\end{figure}

Given the topology of a jump process, we now can define the transport cost matrix.
For each path $P$, let ${\rm len}(P)$ denote its length, which is the number of edges contained in the path.
The shortest-path distance from vertex $x$ to vertex $y$ can be defined as
\begin{equation}
d_{yx}\coloneqq\min_{P}\qty{{\rm len}(P)},
\end{equation}
where the minimum is over all paths that connect $x$ to $y$.
For undirected graphs, clearly $d_{xy}=d_{yx}$.
From the definition of the distances $\{d_{yx}\}$, we can easily verify that the triangle inequality is fulfilled, that is,
\begin{equation}
d_{xy}+d_{yz}\ge d_{xz}
\end{equation}
for arbitrary vertices $x$, $y$, and $z$. Employing these shortest-path distances as the transport costs (i.e., $c_{xy}=d_{xy}$), we hereafter exclusively focus on the following discrete Wasserstein distance:
\begin{equation}\label{eq:graph.Wass.dist}
\mca{W}_1(p^A, p^B)\coloneqq \min_{\pi\in\Pi(p^A, p^B)}\sum_{x,y}d_{xy}\pi_{xy}.
\end{equation}
It is noteworthy that only static information (i.e., graph connectivity) is required to define the Wasserstein distance at this time (see Fig.~\ref{fig:WassDistIllust} for illustration).

In a general case, since $d_{xy}\ge 1$ for all $x\neq y$, the Wasserstein distance is always lower bounded by the total variation distance:
\begin{equation}
\mca{W}_1(p^A, p^B)\ge \mca{T}(p^A,p^B)\coloneqq\frac{1}{2}\sum_x|p_x^A-p_x^B|.
\end{equation}
However, when the underlying graph is fully connected (i.e., the transition between any two states is admissible), the shortest-path distance becomes
\begin{equation}
d_{xy}=1-\delta_{xy},
\end{equation}
where $\delta_{xy}$ is the Kronecker delta of $x$ and $y$.
In this case, the Wasserstein distance coincides with the total variation distance (see Appendix \ref{app:Wass.tot.var.equiv} for the proof):
\begin{equation}\label{eq:Wc.tot.var.dist}
\mca{W}_1(p^A, p^B)=\mca{T}(p^A,p^B).
\end{equation}

\subsection{Thermodynamic interpretation}
With the above setup, we can now state the results.
Conventionally, the discrete Wasserstein distance is defined mathematically using the transport cost matrix based on only the shortest-path distances. 
However, in the following theorem, we explicitly show an intimate relationship between the discrete Wasserstein distance defined in Eq.~\eqref{eq:graph.Wass.dist} and the stochastic thermodynamics of Markov jump processes.
\begin{theorem}\label{thm:cla.dis.Wass.var}
The Wasserstein distance based on a topology $\mca{G}(V,E)$ can be written in variational forms as
\begin{align}
\mca{W}_1( p^A,p^B )&=\min_{\msf{W}_t}{\int_0^\tau\sqrt{\sigma_t m_t}\dd{t}}\label{eq:Wc.var.form1}\\
&=\min_{\msf{W}_t}\sqrt{\Sigma_\tau\mca{M}_\tau}.\label{eq:Wc.var.form2}
\end{align}
Here, the minimum is taken over all transition rate matrices $\{\msf{W}_t\}_{0\le t\le\tau}$ which satisfy the master equation \eqref{eq:Markov.mas.eq} with the boundary conditions ${p_0}={p^A}$ and ${p_\tau}={p^B}$ and induce subgraphs of $\mca{G}(V,E)$ for all times.
\end{theorem}
Theorem \ref{thm:cla.dis.Wass.var} is the first central result, and its sketch proof is given in the following.
Note that the minimization is over all transition rate matrices which are microscopically reversible and induce subgraphs of $\mca{G}(V,E)$ for all times.
This means that the transition rate between two states $x$ and $y$ must be fixed to zero for all times if no edge exists between vertices $x$ and $y$ in the graph $\mca{G}$.
Otherwise, as long as an edge exists between the vertices, the transition rate can be arbitrarily controlled.
Notably, the equality of Eq.~\eqref{eq:Wc.var.form1} can always be ensured with dynamics that satisfy the global detailed balance condition (see Appendix \ref{app:min.gdbc} for the proof).
\begin{proof}
Here, we provide an outline of the proof; see Appendix \ref{app:proof.thm1} for a detailed derivation.
The proof strategy can be mainly divided into the following two steps. We first prove that
\begin{equation}\label{eq:proof.thm1.1}
\mca{W}_1(p^A,p^B)\le\int_0^\tau\sqrt{\sigma_t m_t}\dd{t}\le\sqrt{\Sigma_\tau\mca{M}_\tau}
\end{equation}
holds for all admissible Markovian dynamics that transform $p^A$ into $p^B$, and we then construct a specific process that attains the equality.
Since the second inequality in Eq.~\eqref{eq:proof.thm1.1} is simply a consequence of the Cauchy--Schwarz inequality, Eq.~\eqref{eq:proof.thm1.1} can be proved after we verify the first inequality.
This can be done by proving the following relation:
\begin{equation}\label{eq:proof.thm1.2}
\mca{W}_1(p^A,p^B)\le\int_0^\tau\sum_{x>y}|j_{xy}(t)|\dd{t}\le\int_0^\tau\sqrt{\sigma_t m_t}\dd{t}.
\end{equation}
The second inequality in Eq.~\eqref{eq:proof.thm1.2} can be derived using the Cauchy--Schwarz inequality.
Thus, we need only show the first inequality in Eq.~\eqref{eq:proof.thm1.2}.
To this end, we map the Wasserstein distance to the minimum cost of the minimum cost flow problem in the field of graph theory.
For this problem, we can show that the Wasserstein distance is exactly the optimal flow cost.
Moreover, the Markov jump process also yields an admissible solution to the flow problem with the cost $\int_0^\tau\sum_{x>y}|j_{xy}(t)|\dd{t}$.
Consequently, the first inequality in Eq.~\eqref{eq:proof.thm1.2} is proved.
Finally, we inversely translate the optimal solution of the minimum cost flow problem to construct a Markov jump process that attains the equality of Eq.~\eqref{eq:proof.thm1.1}.
\end{proof}

Some remarks regarding Thm.~\ref{thm:cla.dis.Wass.var} are in order.
First, Eqs.~\eqref{eq:Wc.var.form1} and \eqref{eq:Wc.var.form2} provide a thermodynamic interpretation of the discrete Wasserstein distance; that is, $\mca{W}_1$ equals the minimum product of the thermodynamic and kinetic costs over all admissible Markovian dynamics that transform the source distribution into the target one.
From a different perspective, it can be regarded as a trade-off between irreversible entropy production and dynamical state mobility; that is, to transform a probability distribution into another one, both $\Sigma_\tau$ and $\mca{M}_\tau$ cannot be simultaneously small:
\begin{equation}
\Sigma_\tau\mca{M}_\tau\ge {\mca{W}_1(p^A,p^B)^2}.
\end{equation}
In other words, either the thermodynamic or kinetic cost must be sacrificed to achieve a feasible state transformation.

Second, we show that the discrete Wasserstein distance has analogous thermodynamic properties with the continuous $L^2$-Wasserstein distance. To this end, we rewrite Eq.~\eqref{eq:Wc.var.form2} in the following form:
\begin{equation}\label{eq:Wc.var.form.simpl}
\mca{W}_1(p^A,p^B)=\min_{\msf{W}_t}\sqrt{\bar{D}\tau\Sigma_\tau },
\end{equation}
where we define time-averaged state mobility $\bar{D}\coloneqq\ev{ m}_\tau$.
As previously shown, the kinetic term $m_t$ reduces to the diffusion coefficient in the continuous limit.
Therefore, its time-averaged quantity $\bar{D}$ here plays the same role as the diffusion coefficient $D$ does in the continuous case.
Consequently, Eq.~\eqref{eq:Wc.var.form.simpl} can be regarded as the discrete analog of the Benamou--Brenier formula \eqref{eq:BB.form2} known for the $L^2$-Wasserstein distance.
Equation \eqref{eq:Wc.var.form.simpl} immediately derives a lower bound on irreversible entropy production:
\begin{equation}\label{eq:Wass.bound.disc}
\Sigma_\tau\ge\frac{\mca{W}_1(p^A,p^B)^2}{\bar{D}\tau}.
\end{equation}
This bound is tight and can always be attained for an arbitrary pair of distributions.
In other words, the minimum entropy production among all feasible dynamics that have the same value of $\bar{D}$ is given by the Wasserstein distance:
\begin{equation}\label{eq:Wass.ent.eq.disc}
\min_{\ev{ m}_\tau=\bar{D}}\Sigma_\tau=\frac{\mca{W}_1(p^A,p^B)^2}{\bar{D}\tau}.
\end{equation}
Equations \eqref{eq:Wass.bound.disc} and \eqref{eq:Wass.ent.eq.disc} can be considered discrete analogs of Eqs.~\eqref{eq:Wass.bound.cont} and \eqref{eq:Wass.ent.eq.cont}, respectively.

\begin{figure}[t]
\centering
\includegraphics[width=0.85\linewidth]{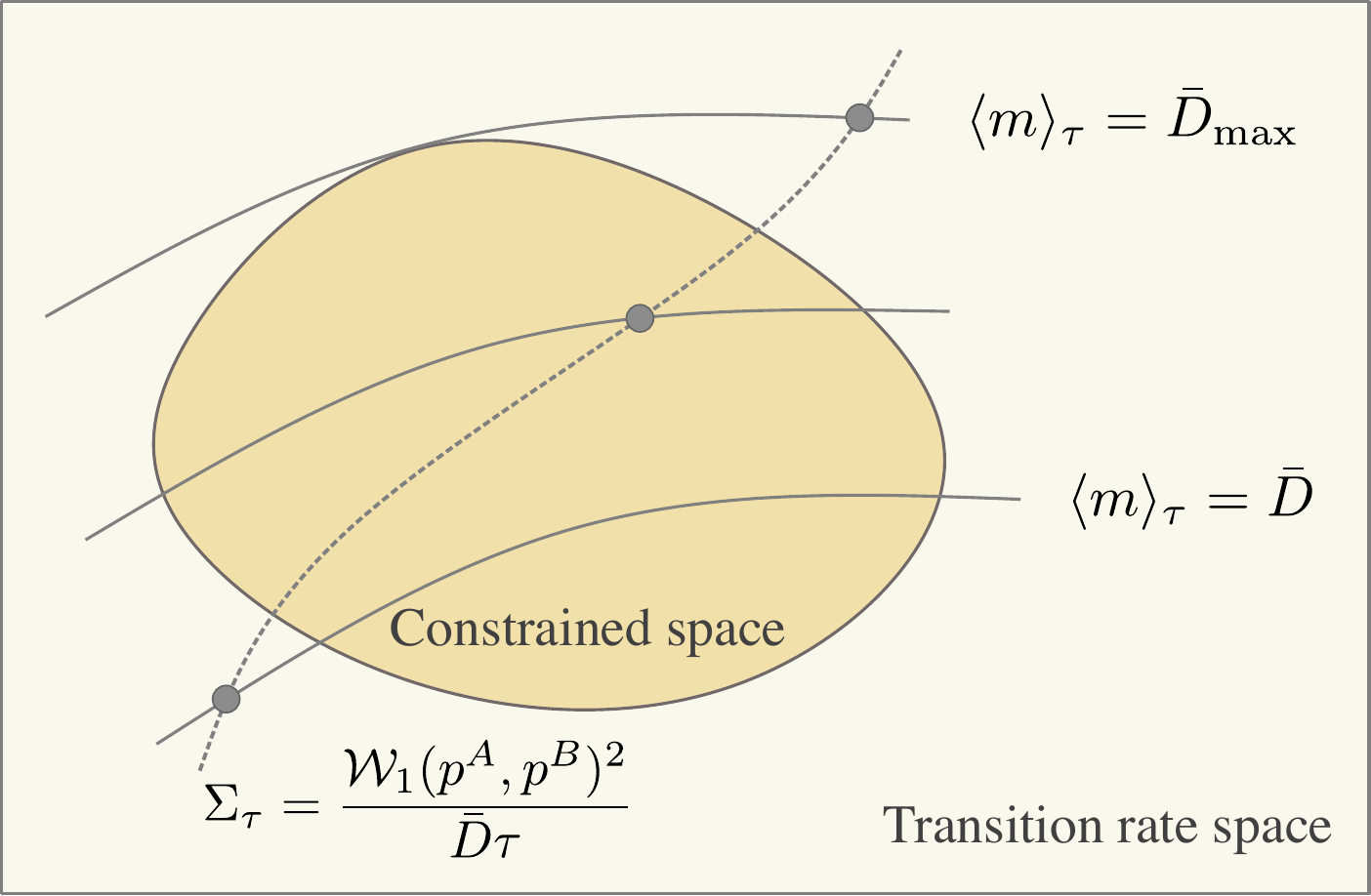}
\protect\caption{Schematic of the thermodynamic structure of minimum dissipation in the space of transition rates. Solid lines represent dynamics that have the same value of the average state mobility $\ev{m}_\tau$. Black circles depict extreme points in which the minimum dissipation is attained, provided that the average state mobility is fixed. In the presence of other constraints on transition rates, the minimum dissipation is lower bounded by the Wasserstein distance and maximum value of $\ev{m}_\tau$.}\label{fig:MinDiss}
\end{figure}

Third, Thm.~\ref{thm:cla.dis.Wass.var} provides insights into the problem of minimizing entropy production in discrete Markovian dynamics.
Previous studies have shown that the irreversible entropy production required to transform the initial into the final distribution can be arbitrarily small if no constraint is placed on the transition rates \cite{Vu.2021.PRL2,Remlein.2021.PRE,Dechant.2022.JPA}.
Equation~\eqref{eq:Wass.ent.eq.disc} also confirms this, where the minimum entropy production depends on $\bar{D}$ and can be arbitrarily adjusted.
Theorem \ref{thm:cla.dis.Wass.var} suggests that fixing $\bar{D}$ is a reasonable constraint under which the minimum entropy production is determined by the Wasserstein distance, as in the continuous case.
Notably, as Appendix \ref{app:mt.prop} shows, $\bar{D}$ can be fixed to an arbitrary positive value, indicating the flexibility of the optimization problem.
The optimal control protocol can also be constructed from optimal coupling, as shown in our proof of Thm.~\ref{thm:cla.dis.Wass.var}.
Moreover, for arbitrary $\bar{D}>0$, the minimum entropy production in Eq.~\eqref{eq:Wass.ent.eq.disc} can always be attained using a system with conservative forces (see Appendix \ref{app:min.ent.prod.cons.force} for the proof).
It is noteworthy that the discussion thus far has not imposed any other constraints on the transition rates, except for fixing $\ev{m}_\tau$. Therefore, if some additional constraints are placed on the transition rates, such as upper or lower bounds on the magnitude of transition rates, it may not be the case. Nevertheless, a lower bound can be derived for the minimum entropy production in this case. Let $\bar{D}_{\rm max}$ be the maximum of $\ev{m}_\tau$ among all processes that transform distribution $p^A$ into $p^B$ under these constraints. Then, the minimum entropy production is lower bounded by the Wasserstein distance and $\bar{D}_{\rm max}$ as
\begin{equation}
\Sigma_\tau\ge\frac{\mca{W}_1(p^A,p^B)^2}{\bar{D}_{\rm max}\tau}.
\end{equation}
The thermodynamic structure of minimum dissipation is illustrated in Fig.~\ref{fig:MinDiss}.

Finally, as shown in Sec.~\ref{sec:appl}, the variational formulas \eqref{eq:Wc.var.form1} and \eqref{eq:Wc.var.form2} have crucial implications for thermodynamic speed limits and thermodynamic cost of information erasure at arbitrary temperatures.

Note that each topology induces a different Wasserstein metric.
In the following, we consider a specific topology and discuss the relevance of the variational formulas \eqref{eq:Wc.var.form1} and \eqref{eq:Wc.var.form2}.
For other common topologies, see Appendix \ref{app:com.topo}.

One common topology is one-dimensional nearest-neighbor, in which a jump between states $x$ and $y$ is admitted if and only if $|x-y|=1$.
This topology is relevant to Brownian random walks and the discretization of a one-dimensional Langevin system.
The shortest-path distances in this topology can be readily calculated as
\begin{equation}
d_{xy}=|x-y|.
\end{equation}
Because this cost matrix is the discrete analog of the cost function $c(x,y)=|x-y|$ used in the definition of the continuous $L^1$-Wasserstein distance $W_1$, the discrete Wasserstein distance $\mca{W}_1$ should be reduced to $W_1$ in the continuous limit.
Let $\Delta x$ be the space interval. Then, $\mca{W}_1$ converges to $W_1$ as $\Delta x\to 0$ and $N\to\infty$:
\begin{equation}\label{eq:1dchain.W1.lim}
\mca{W}_1( p^A,p^B )\Delta x\xrightarrow[\Delta x\to 0]{N\to \infty} W_1(p^A,p^B).
\end{equation}
In addition, as shown in the proof of Thm.~\ref{thm:cla.dis.Wass.var}, $\mca{W}_1( p^A,p^B )$ can be expressed in terms of the probability currents as
\begin{equation}
\mca{W}_1( p^A,p^B )=\min_{\msf{W}_t}{\int_0^\tau\sum_{x>y}|j_{xy}(t)|\dd{t}}.\label{eq:Wc.var.form.cor}
\end{equation}
Equation \eqref{eq:Wc.var.form.cor} implies that the discrete Wasserstein distance is equal to the minimum sum of absolute probability currents.
In the case considered here, the equality \eqref{eq:Wc.var.form.cor} reads
\begin{equation}\label{eq:Wc.1dchain.var.form}
\mca{W}_1( p^A,p^B ) = \min_{\msf{W}_t}{\int_0^\tau\sum_{x=1}^{N-1}|j_{x+1,x}(t)|\dd{t}}.
\end{equation}
Noticing that $\sum_{x=1}^{N-1}|j_{x+1,x}(t)|\Delta x\to \int_{\mbb{R}}|j_t(x)|\dd{x}$ as $\Delta x\to 0$, we obtain the following limit:
\begin{equation}\label{eq:1dchain.W1.var.lim}
\mca{W}_1( p^A,p^B )\Delta x\xrightarrow[\Delta x\to 0]{N\to \infty}\min_{j_t}\int_0^\tau\int_{\mbb{R}}|j_t(x)|\dd{x}\dd{t}.
\end{equation}
Combining Eqs.~\eqref{eq:1dchain.W1.lim} and \eqref{eq:1dchain.W1.var.lim} gives the following relation:
\begin{align}
W_1(p^A,p^B)&=\min_{j_t}{\int_0^\tau\int_{\mbb{R}}|j_t(x)|\dd{x}\dd{t}},\label{eq:Wc.cont.limit}
\end{align}
where $j_t(x)$ is subject to the equation $\dot{p}_t(x)=-\partial_xj_t(x)$.
Notably, Eq.~\eqref{eq:Wc.cont.limit} is exactly the Benamou--Brenier formula for the continuous $L^1$-Wasserstein distance in the one-dimensional case \cite{Chen.2017.CSL}.
Therefore, we can conclude that Eq.~\eqref{eq:Wc.var.form.cor} provides a unified generalization of the Benamou--Brenier formula for the $L^1$-Wasserstein distance.

Let us now consider the discretization of one-dimensional Langevin dynamics, that is, Markov jump processes with transition rates specified as in Eqs.~\eqref{eq:1d.dis.tran.rate1} and \eqref{eq:1d.dis.tran.rate2}.
For these jump processes, the dynamical state mobility reduces to the diffusion coefficient in the continuous limit (i.e., $\mca{M}_\tau(\Delta x)^2\to D\tau$ as $\Delta x\to 0$).
In the continuous case, the Benamou--Brenier formula \eqref{eq:Wc.cont.limit} can be expressed as
\begin{align}
W_1(p^A,p^B)&=\min_{j_t}{\int_0^\tau\int_{\mbb{R}}\sqrt{\sigma_t(x)m_t(x)}\dd{x}\dd{t}}\notag\\
&\le\min_{F_t} \int_0^\tau\sqrt{D\sigma_t}\dd{t}=\min_{F_t} \sqrt{D\tau\Sigma_\tau},\label{eq:ring.tmp1}
\end{align}
where we define $m_t(x)\coloneqq Dp_t(x)$ and the local entropy production rate $\sigma_t(x)\coloneqq j_t(x)^2/[Dp_t(x)]$.
From Thm.~\ref{thm:cla.dis.Wass.var} and Eq.~\eqref{eq:ring.tmp1}, we can conclude that in the continuous limit, the equality in Thm.~\ref{thm:cla.dis.Wass.var} might not be achieved with Markov jump processes whose transition rates are expressed as in Eqs.~\eqref{eq:1d.dis.tran.rate1} and \eqref{eq:1d.dis.tran.rate2}.
This shows the difference between the discrete and continuous cases, where the discrete case has more degrees of freedom than the continuous case.

Theorem \ref{thm:cla.dis.Wass.var} characterizes the discrete Wasserstein distance $\mca{W}_1$ in terms of the thermodynamic and kinetic costs associated with Markovian dynamics.
In Appendix \ref{app:alt.Wc.exp}, we show that Thm.~\ref{thm:cla.dis.Wass.var} has some useful corollaries that not only provide alternative expressions for $\mca{W}_1$ but also lead to stringent bounds for thermodynamic speed limits.
Using other combinations of irreversible entropy production, pseudo entropy production, and dynamical activity, the discrete Wasserstein distance $\mca{W}_1$ can be expressed in similar variational forms.

\section{Quantum generalization}\label{sec:quan.gen}
We next generalize our framework to the quantum case.
We first briefly introduce quantum thermodynamics of Markovian open quantum dynamics described by the Lindblad equations and define a quantum analog of dynamical state mobility.
Then, we define a quantum Wasserstein distance and derive analogous variational formulas for the quantum Wasserstein distance in terms of thermodynamic cost.

\subsection{Markovian open quantum dynamics}

We consider a finite-dimensional open quantum system, which is attached to single or multiple thermal reservoirs.
In the weak-coupling limit, the time evolution of the reduced density matrix can be described by the Lindblad master equation \cite{Lindblad.1976.CMP},
\begin{equation}\label{eq:Lind.eq}
\dot{\varrho}_t=\mca{L}_t(\varrho_t)\coloneqq-i[H_t,\varrho_t]+\sum_k\mca{D}[L_k(t)]\varrho_t,
\end{equation}
where $H_t$ is the time-dependent Hamiltonian, $\mca{D}$ is the dissipator given by $\mca{D}[L]\varrho\coloneqq L\varrho L^\dagger-\qty{L^\dagger L,\varrho}/2$, and $L_k(t)$ are jump operators.
$[\circ,\star]$ and $\{\circ,\star\}$ denote the commutator and anticommutator of the two operators, respectively.
Hereafter, we set the Planck constant to unity $\hbar=1$.
To guarantee thermodynamically consistent dynamics, we assume that the jump operators satisfy the local detailed balance condition \cite{Horowitz.2013.NJP,Manzano.2018.PRX}; that is, they come in pairs $(k,k')$ such that
\begin{equation}
L_k(t)=e^{s_k(t)/2}L_{k'}(t)^\dagger,
\end{equation}
where $s_k(t)=-s_{k'}(t)$ denotes the entropy change in the environment due to the jump operator $L_k(t)$.
In the case of a single reservoir at inverse temperature $\beta$, we can write $s_k(t)=\beta \omega_k(t)$, where $\omega_k(t)$ is the energy change associated with the $k$th jump.

\subsection{Entropy production and dynamical activity}
Given the previous setup, we can now introduce quantum entropy production and dynamical activity.
Similar to the classical case, irreversible entropy production can be defined as the sum of entropy changes in the system and environment as
\begin{equation}\label{eq:qua.ent.prod.def}
\Sigma_\tau\coloneqq \Delta S_{\rm sys}+\Delta S_{\rm env},
\end{equation}
where $\Delta S_{\rm sys}\coloneqq S(\varrho_\tau)-S(\varrho_0)$ is the difference in the von Neumann entropy $S(\varrho)=-\tr{\varrho\ln\varrho}$ of the system and $\Delta S_{\rm env}$ denotes environmental entropy production, given by \cite{Horowitz.2013.NJP,Manzano.2018.PRX}
\begin{equation}
\Delta S_{\rm env}\coloneqq\int_0^\tau\sum_k\tr{L_k(t)\varrho_tL_k^\dagger(t)} s_k(t)\dd{t}.
\end{equation}
With this definition, we can prove that $\Sigma_\tau$ is always nonnegative, which implies the second law of thermodynamics.
For the case of a single reservoir and the jump operators that characterize transitions between energy eigenstates (i.e., $[L_k(t),H_t]=\omega_k(t)L_k(t)$), the entropy production of the environment reduces exactly to the conventional form \cite{Alicki.1979.JPA},
\begin{equation}
\Delta S_{\rm env}=-\beta\int_0^\tau\tr{H_t\dot{\varrho}_t}\dd{t}.
\end{equation}

Quantum dynamical activity can be analogously defined as in the classical case.
The frequency of jumps at time $t$ can be quantified as
\begin{equation}\label{eq:qua.dyn.act.def}
a_t\coloneqq\sum_k\tr{L_k(t)\varrho_tL_k^\dagger(t)},
\end{equation}
and the average total number of jumps can be calculated as $\mca{A}_\tau\coloneqq\int_0^\tau a_t\dd{t}$.
Quantum dynamical activity characterizes the thermalization rate of thermodynamic processes.
In addition, it has been shown that quantum dynamical activity constrains the precision of generic counting observables and their first passage time in quantum jump processes \cite{Hasegawa.2020.PRL,Vu.2022.PRL.TUR}.

It is convenient to alternatively express entropy production and dynamical activity defined in Eqs.~\eqref{eq:qua.ent.prod.def} and \eqref{eq:qua.dyn.act.def}, respectively. Let $\varrho_t=\sum_xp_x(t)\dyad{x_t}$ be the spectral decomposition of the density matrix $\varrho_t$. We then define transition rates between eigenbasis as $w_k^{xy}(t)\coloneqq|\mel{x_t}{L_k(t)}{y_t}|^2\ge 0$.
Notice that $w_k^{xy}(t)=e^{s_k(t)}w_{k'}^{yx}(t)$.
Taking the time derivative of $p_x(t)=\mel{x_t}{\varrho_t}{x_t}$, we obtain the following master equation for the distribution $\{p_x(t)\}$:
\begin{equation}\label{eq:eigen.master.eq}
\dot{p}_x(t)=\sum_k\sum_{y(\neq x)}\qty[w_{k}^{xy}(t)p_y(t)-w_{k}^{yx}(t)p_x(t)].
\end{equation}
Analogous to the classical case, we define
\begin{align}
a_k^{xy}(t)&\coloneqq w_k^{xy}(t)p_y(t),\\
j_k^{xy}(t)&\coloneqq w_k^{xy}(t)p_y(t)-w_{k'}^{yx}(t)p_x(t).
\end{align}
Using these probability currents, we can write the master equation as
\begin{equation}
\dot p_x(t)=\sum_k\sum_{y(\neq x)}j_k^{xy}(t).
\end{equation}
We emphasize that the classical-like master equation \eqref{eq:eigen.master.eq} is rigorously derived from Eq.~\eqref{eq:Lind.eq}. This equation is introduced only for the proof convenience of several properties that the dynamics \eqref{eq:Lind.eq} possesses.
After some simple manipulations, we can prove that the entropy production rate ${\sigma}_t\coloneqq\dot{\Sigma}_t$ can be analytically expressed as (see Appendix \ref{app:qua.ent.prod.exp} for the proof)
\begin{equation}\label{eq:qua.ent.prod.rate}
\sigma_t=\frac{1}{2}\sum_k\sum_{x,y}j_k^{xy}(t)\ln\frac{w_k^{xy}(t)p_y(t)}{w_{k'}^{yx}(t)p_x(t)}.
\end{equation}
Besides, plugging the spectral decomposition of $\varrho_t$ and inserting $\mbb{1}=\sum_{x}\dyad{x_t}$ into Eq.~\eqref{eq:qua.dyn.act.def}, the dynamical activity rate can also be expressed as
\begin{equation}
a_t=\sum_k\sum_{x,y}w_{k}^{xy}(t)p_y(t)=\sum_k\sum_{x,y}a_{k}^{xy}(t).
\end{equation}
Note that both $\sigma_t$ and $a_t$, which can be written in terms of the transition rates of the master equation \eqref{eq:eigen.master.eq}, are the entropy production and dynamical activity rates associated with the Lindblad dynamics \eqref{eq:Lind.eq}, respectively.

\subsection{Quantum dynamical state mobility}

Analogous with the classical case, the quantum analog of dynamical state mobility can be defined as
\begin{equation}\label{eq:quan.kinetic.cost}
 m_t\coloneqq\frac{1}{2}\sum_{k}e^{-s_k(t)/2}\ev{L_k(t)^\dagger,\sop{\varrho_t}_{s_k(t)}(\mca{P}_t[L_k(t)^\dagger])},
\end{equation}
where $\ev{X,Y}\coloneqq\tr{X^\dagger Y}$ denotes the scalar inner product, $\mca{P}_t$ is a super-operator given by $\mca{P}_t[X]\coloneqq X-\sum_x\mel{x_t}{X}{x_t}\dyad{x_t}$, and the tilted operator $\sop{\phi}_\theta(X)$ is defined for arbitrary density matrix $\phi$, real number $\theta$, and linear operator $X$ as
\begin{equation}
\sop{\phi}_\theta(X)\coloneqq e^{-\theta/2}\int_0^1e^{\theta u}\phi^uX\phi^{1-u}\dd{u}.
\end{equation}
The quantum kinetic cost can be analogously defined as
\begin{equation}
\mca{M}_\tau\coloneqq\int_0^\tau m_t\dd{t}.
\end{equation}
From the mathematical definition in Eq.~\eqref{eq:quan.kinetic.cost}, interpreting the term $m_t$ as a kinetic term may not be intuitive.
In the following, we provide the physical interpretations of $m_t$ from two perspectives.

First, by focusing on the master equation of the distribution $\{p_x(t)\}$, we can show that $m_t$ is equal to the dynamical state mobility associated with Markovian jump dynamics \eqref{eq:eigen.master.eq}:
\begin{equation}\label{eq:qua.mt.tmp1}
 m_t=\sum_{k}\sum_{x>y}\frac{a_k^{xy}(t)-a_{k'}^{yx}(t)}{\ln a_k^{xy}(t)-\ln a_{k'}^{yx}(t)}.
\end{equation}
Note that by applying the inequality \eqref{eq:rel.lmn.amn} to Eq.~\eqref{eq:qua.mt.tmp1}, we can readily prove that $m_t$ is upper bounded by the dynamical activity,
\begin{align}
 m_t &\le\sum_{k}\sum_{x>y}\frac{a_k^{xy}(t)+a_{k'}^{yx}(t)}{2}\notag\\
&=\frac{1}{2}\sum_{k}\sum_{x\neq y}a_k^{xy}(t)\notag\\
&\le \frac{a_t}{2}.\label{eq:lt.at.qua}
\end{align}

Second, let us consider the case of a single reservoir, in which the jump operators satisfy $[L_k(t),H_t]=\omega_k(t)L_k(t)$.
In this case, the thermal state $\varrho_t^{\rm eq}\coloneqq e^{-\beta H_t}/\tr e^{-\beta H_t}$ is always the instantaneous equilibrium state (i.e., $\mca{L}_t(\varrho_t^{\rm eq})=0$).
Note that the Lindblad master equation \eqref{eq:Lind.eq} can be rewritten as \cite{Vu.2021.PRL}
\begin{equation}
\dot{\varrho}_t=\mca{U}_{\varrho_t}(t,f_t)+\mca{O}_{\varrho_t}(t,f_t),\label{eq:Lind.sup.ope}
\end{equation}
where $f_t\coloneqq -\ln\varrho_t+\ln\varrho_t^{\rm eq}$ is the quantum thermodynamic force and $\mca{U}_\phi(t,X)$ and $\mca{O}_\phi(t,X)$ are time-dependent super-operators defined, respectively, as
\begin{align}
\mca{U}_{\phi}(t,X)&\coloneqq i\beta^{-1}[X,\phi],\\
\mca{O}_{\phi}(t,X)&\coloneqq\frac{1}{2}\sum_{k}e^{-s_k(t)/2}[L_k(t),\sop{\phi}_{s_k(t)}([L_k(t)^\dagger,X])].
\end{align}
The super-operators $\mca{U}$ and $\mca{O}$ characterize the unitary and dissipative parts of Lindblad dynamics, respectively.
They linearly relate the rate of the density matrix to the thermodynamic force. 
In addition, the entropy production rate can be written in a quadratic form of the thermodynamic force as \cite{Vu.2021.PRL}
\begin{equation}\label{eq:qua.ent.f}
\sigma_t=\ev{f_t,\mca{O}_{\varrho_t}(t,f_t)}.
\end{equation}
Since Eqs.~\eqref{eq:Lind.sup.ope} and \eqref{eq:qua.ent.f} are analogous to Eqs.~\eqref{eq:pt.Kf} and \eqref{eq:ent.Kf} in the classical case, respectively, the super-operator $\mca{O}$ can be regarded as a quantum Onsager-like super-operator.

We now investigate the relationship between $m_t$ and the Onsager-like super-operator $\mca{O}$.
To this end, we employ the vectorization of a linear operator $X$ as
\begin{equation}
X=\sum_{i,j}x_{ij}\dyad{i}{j}\to \dket{X}=\sum_{i,j}x_{ij}\ket{i}\otimes\ket{j}.
\end{equation}
Using this representation, we can rewrite the Lindblad master equation \eqref{eq:Lind.sup.ope} as
\begin{equation}
\dket{\dot\varrho_t}=\msf{U}_t\dket{f_t}+\msf{O}_t\dket{f_t},
\end{equation}
where the linear matrices $\msf{U}_t$ and $\msf{O}_t$ are defined as
\begin{align}
\msf{U}_t&\coloneqq i\beta^{-1}\qty(\mbb{1}\otimes\varrho_t^\top - \varrho_t\otimes\mbb{1}),\\
\msf{O}_t&\coloneqq \frac{1}{2}\sum_{k}e^{-s_k(t)}\int_0^1e^{s_k(t)u}\msf{O}_k(t,u)\dd{u}.
\end{align}
Here, $\top$ denotes the matrix transpose and $\msf{O}_k(t,x)$ is given by
\begin{align}
&\msf{O}_k(t,u)\\
&\coloneqq L_k(t)\varrho_t^uL_k(t)^\dagger\otimes(\varrho_t^{1-u})^\top+\varrho_t^u\otimes (L_k(t)^\dagger\varrho_t^{1-u}L_k(t))^\top\notag\\
&-L_k(t)\varrho_t^u\otimes (L_k(t)^\dagger\varrho_t^{1-u})^\top-\varrho_t^uL_k(t)^\dagger\otimes (\varrho_t^{1-u}L_k(t))^\top.\notag
\end{align}
Note that $\msf{U}_t$ and $\msf{O}_t$ are the matrix representations of the super-operators $\mca{U}$ and $\mca{O}$, respectively.
Simple algebraic calculations show that the term $m_t$ can be related to the diagonal elements of the Onsager-like matrix $\msf{O}_t$ as (see Appendix \ref{app:qua.mobi.exp} for the proof)
\begin{equation}\label{eq:Lt.Onsager.rel}
 m_t=\frac{1}{2}\sum_{x}\bra{x_t}\otimes\ket{x_t}^\top\msf{O}_t\ket{x_t}\otimes\bra{x_t}^\top.
\end{equation}
In this sense, $m_t$ can be regarded as a quantum kinetic term.

\subsection{Quantum optimal transport distance and thermodynamic interpretation}
Although the classical Wasserstein distance is well formulated and studied, its quantum version remains under development.
Several quantum generalizations of the Wasserstein distance have been proposed \cite{Carlen.2014.CMP,Chen.2017.CSL,Chen.2018.TAC,Duvenhage.2020.arxiv,DePalma.2021.TIT,Vu.2021.PRL,Friedland.2021.arxiv}.
However, defining the quantum $L^1$-Wasserstein distance unambiguously by directly generalizing the classical distance has been shown to be impossible \cite{Agredo.2017.STO}.

By a naive extension using quantum coupling, a quantum optimal transport distance can be defined as
\begin{equation}\label{eq:org.Wq.dist}
W_q(\varrho^A,\varrho^B)\coloneqq\min_{\varrho^{AB}\in\Pi(\varrho^A,\varrho^B)}\tr{C\varrho^{AB}},
\end{equation}
where the coupling $\Pi(\varrho^A,\varrho^B)$ denotes the set of density matrices $\varrho^{AB}$ defined over the Hilbert space $\mca{H}\otimes\mca{H}$ and satisfy $\tr_B{\varrho^{AB}}=\varrho^A$ and $\tr_A{\varrho^{AB}}=\varrho^B$, and $C$ is a cost matrix that must be properly chosen to guarantee that $W_q$ is a distance.
In the classical case, the total variation distance is a classical Wasserstein distance with an appropriate choice of the cost matrix $C$.
It is thus natural to ask whether a cost matrix $C$ exists such that the quantum version of the total variation distance (i.e., the trace distance) can be represented as a quantum Wasserstein distance defined in Eq.~\eqref{eq:org.Wq.dist}.
Unfortunately, Ref.~\cite{Yu.2018.arxiv} showed that the trace distance could not be expressed in terms of this type of Wasserstein distance. In other words, for any choice of the cost matrix $C$, density matrices $\varrho^A$ and $\varrho^B$ always exist such that the distance $W_q$ defined in Eq.~\eqref{eq:org.Wq.dist} differs from the trace distance:
\begin{equation}
W_q(\varrho^A,\varrho^B)\neq\frac{1}{2}\|\varrho^A-\varrho^B\|_1\eqqcolon \mca{T}(\varrho^A,\varrho^B).
\end{equation}

Our aim is to relate quantum optimal transport distances and dissipation in Lindblad dynamics.
Note that Lindblad dynamics consist of a non-dissipative unitary part and dissipative Lindblad part.
Both parts jointly contribute to the time evolution of the system's density matrix.
In the vanishing coupling limit, irreversible entropy production becomes zero, whereas the distance $W_q(\varrho_0,\varrho_\tau)$ may be positive since $\varrho_0\neq\varrho_\tau$.
Therefore, relating dissipation to the optimal transport distances defined in the current form \eqref{eq:org.Wq.dist} is impossible.
Inspired by the dissipative structure of Lindblad dynamics, we define the following distance:
\begin{equation}
\mca{W}_q(\varrho^A,\varrho^B)\coloneqq\frac{1}{2}\min_{V^\dagger V=\mbb{1}}\|V\varrho^AV^\dagger-\varrho^B\|_1.
\end{equation}
Here, the minimum is over all possible unitaries $V$.
Intuitively, the distance $\mca{W}_q$ characterizes the state difference induced by the dissipative Lindblad part. Thus, it is expected to be relevant to dissipation.
Note that in the zero-dissipation limit (i.e., the system is unitarily evolved), this distance also vanishes.
Although the distance $\mca{W}_q$ is defined in a variational form, it can be analytically calculated using the eigenvalues of the density matrices.
Interestingly, it becomes exactly the classical Wasserstein distance between the eigenvalue distributions:
\begin{equation}\label{eq:ana.exp.tra.nrm}
\mca{W}_q(\varrho^A,\varrho^B)=\frac{1}{2}\sum_x|p_x^A-p_x^B|=\mca{T}(p^A,p^B),
\end{equation}
where $\{p_x^A\}$ and $\{p_x^B\}$ are increasing eigenvalues of $\varrho^A$ and $\varrho^B$, respectively (see Appendix \ref{app:qua.Wass.dis.exp.proof} for the proof).
For this reason, hereafter, $\mca{W}_q$ is referred to as the quantum Wasserstein distance.
Evidently, this distance satisfies the triangle inequality.
However, it is a pseudo-metric (i.e., $\mca{W}_q(\varrho^A,\varrho^B)=0$ for $\varrho^A\neq\varrho^B$ is possible). This originates from our goal of relating the defined distance to dissipation in Lindblad dynamics.

For the quantum Wasserstein distance previously defined, we provide the following thermodynamic interpretation.
\begin{theorem}\label{thm:qua.dis.Wass.var}
The quantum Wasserstein distance can be written in the following variational form:
\begin{align}
\mca{W}_q(\varrho^A,\varrho^B)&=\min_{\mca{L}_t}{\int_0^\tau\sqrt{\sigma_t m_t}\dd{t}}\label{eq:Wq.var.form1}\\
&=\min_{\mca{L}_t}\sqrt{\Sigma_\tau\mca{M}_\tau}\label{eq:Wq.var.form2}.
\end{align}
Here, the minimum is taken over all super-operators $\{\mca{L}_t\}_{0\le t\le\tau}$ that satisfy the Lindblad master equation \eqref{eq:Lind.eq} with boundary conditions $\varrho_0 =\varrho^A$ and $\varrho_\tau=\varrho^B$.
\end{theorem}

Theorem \ref{thm:qua.dis.Wass.var} is the second central result, and its sketch proof is given in the following.
Interestingly, Thm.~\ref{thm:qua.dis.Wass.var} has the same structure as Thm.~\ref{thm:cla.dis.Wass.var} in the classical case.
This implies a universal relationship between the optimal transport distances and dissipation in classical and quantum discrete systems.
\begin{proof}
We briefly describe the proof strategy; for a detailed derivation, see Appendix \ref{app:proof.qua.thm}.
We first prove that the inequalities
\begin{equation}\label{eq:thm4.tmp1}
\mca{W}_q(\varrho^A,\varrho^B)\le\int_0^\tau\sqrt{\sigma_t m_t}\dd{t}\le\sqrt{\Sigma_\tau\mca{M}_\tau}
\end{equation}
hold for any Markovian open quantum dynamics and then construct a specific process that simultaneously attains all the equalities of Eq.~\eqref{eq:thm4.tmp1}.
The inequalities in Eq.~\eqref{eq:thm4.tmp1} can be proved similarly as in the classical case.
To construct the dynamics that can achieve the equalities, we first construct a classical Markov jump process that transforms distribution $p^A$ into $p^B$ and satisfies
\begin{equation}
\mca{T}(p^A,p^B)=\int_0^\tau\sqrt{\sigma_t m_t}\dd{t}=\sqrt{\Sigma_\tau\mca{M}_\tau}.
\end{equation}
Here, $\{p_x^A\}$ and $\{p_x^B\}$ are increasing eigenvalues of $\varrho^A$ and $\varrho^B$, respectively.
Subsequently, we construct Lindblad dynamics based on this classical jump process such that the dynamics transforms density matrix $\varrho^A$ into $\varrho^B$, and the quantities $\sigma_t$ and $m_t$ are identical to those in the classical jump process.
We can verify that this quantum dynamics attains the equalities of Eq.~\eqref{eq:thm4.tmp1}.
\end{proof}

Similar to the classical case, the quantum Wasserstein distance can also be determined through the entropy production and dynamical activity associated with Markovian quantum dynamics, which is stated in Cor.~\ref{cor:qua.dis.Wass.var1} in Appendix \ref{app:qua.cor.proof}.

\section{Applications for thermodynamic interpretation of optimal transport}\label{sec:appl}
In this section, we present applications for our central results, namely, Thms.~\ref{thm:cla.dis.Wass.var} and \ref{thm:qua.dis.Wass.var}.
Specifically, we show that these variational formulas lead to stringent bounds for thermodynamic speed limits and information erasure at arbitrary temperatures.

\subsection{Classical and quantum thermodynamic speed limits}
The speed of state transformation in any system cannot be made arbitrarily fast because of physical constraints.
This fact leads to a natural question: What is the ultimate limit for state transformation?
This question sparked a lot of research and gave rise to the concept of speed limits.

Precisely speaking, speed limits impose lower bounds on the operational time required for evolving a system from a given state to a target one.
Originally, speed limits were derived for closed quantum systems, inspired by the Heisenberg time-energy uncertainty principle \cite{Mandelstam.1945.JP}.
One of the celebrated results is the Mandelstam--Tamm bound, which applies to closed quantum systems and takes the following form:
\begin{equation}\label{eq:ML.bound}
\tau\ge \frac{\mca{B}(\varrho_0,\varrho_\tau)}{\ev{\Delta H}_\tau},
\end{equation}
where $\mca{B}(\varrho,\sigma)\coloneqq\arccos\tr{|\sqrt{\varrho}\sqrt{\sigma}|}$ is the Bures angle and $(\Delta H_t)^2\coloneqq\tr{H_t^2\varrho_t}-\tr{H_t\varrho_t}^2$ is the energy fluctuation.
Equation \eqref{eq:ML.bound} implies that the speed of state transformation in closed quantum systems is constrained by the fluctuation of energy.
Various types of speed limits were subsequently generalized for open quantum and classical systems \cite{Mandelstam.1945.JP,Margolus.1998.PD,Campo.2013.PRL,Taddei.2013.PRL,Deffner.2013.PRL,Pires.2016.PRX,Okuyama.2018.PRL,Campaioli.2018.PRL,Shanahan.2018.PRL,Sun.2021.PRL,Connor.2021.PRA,Hamazaki.2022.PRXQ,Nakajima.2022.arxiv,Pintos.2022.PRX,Hasegawa.2022.arxiv} (see Ref.~\cite{Deffner.2017.JPA} for a comprehensive review).

Although several versions of classical and quantum speed limits exist for open systems, here we aim to develop {\it thermodynamic} bounds that satisfy two conditions: (i) they should be tight (i.e., for generic initial and final states, a configuration of the system always exists that transforms these states and saturates the bounds) and (ii) they should be physically interpretable (i.e., all quantities appearing in the bound are physically meaningful).
In the following, we derive these thermodynamic speed limits from the variational formulas for both classical and quantum cases.

\subsubsection{Classical case}
We consider a discrete classical system modeled by a Markov jump process [Eq.~\eqref{eq:Markov.mas.eq}].
The system is driven by thermodynamic forces and evolves according to the laws of thermodynamics.
Intuitively, to achieve fast transformation, we must pay some costs.
In the following, we derive fundamental bounds on the operational time that is required to evolve the system's distribution to the target one.

Let $\mca{G}(V,E)$ be the underlying topology of the jump process (i.e., the graph connectivity that determines whether the transition between two states is allowed).
Then, we can define the corresponding Wasserstein distance based on the graph $\mca{G}$.
According to Thm.~\ref{thm:cla.dis.Wass.var}, we have
\begin{equation}\label{eq:thm1.app}
\mca{W}_1(p_0,p_\tau)=\min_{\msf{W}_t}{ \int_0^\tau\sqrt{\sigma_t m_t}\dd{t} }=\min_{\msf{W}_t}\sqrt{\Sigma_\tau\mca{M}_\tau}.
\end{equation}
Since the system dynamics considered here is one of the admissible dynamics that transform ${p_0}$ into ${p_\tau}$, the following inequalities follow immediately from the equality \eqref{eq:thm1.app}:
\begin{align}
\mca{W}_1(p_0,p_\tau)&\le \int_0^\tau\sqrt{\sigma_t m_t}\dd{t}=\tau\ev{\sqrt{\sigma m}}_\tau\\
&\le\tau\sqrt{\ev{\sigma}_\tau\ev{ m}_\tau} .
\end{align}
Consequently, we obtain lower bounds on the operational time in terms of the Wasserstein distance, thermodynamic cost, and kinetic cost as follows:
\begin{equation}\label{eq:speed.limit}
\tau\ge\frac{\mca{W}_1(p_0,p_\tau)}{\ev{\sqrt{\sigma m}}_\tau}\ge\frac{\mca{W}_1(p_0,p_\tau)}{\sqrt{\ev{\sigma}_\tau\ev{ m}_\tau}}.
\end{equation}
Equation \eqref{eq:speed.limit} implies that both irreversible entropy production and state mobility jointly constrain the speed of state transformation.
Using Cor.~\ref{cor:cla.dis.Wass.var1} in Appendix \ref{app:alt.Wc.exp} and following the same procedure, we also obtain similar but tighter bounds in terms of time-averaged entropy production and dynamical activity as
\begin{align}
\tau\ge\frac{2\mca{W}_1(p_0,p_\tau)}{\ev{\sigma\Phi(\sigma/2a)^{-1}}_\tau}\ge \frac{2\mca{W}_1(p_0,p_\tau)}{\ev{\sigma}_\tau\Phi({\ev{\sigma}_\tau}/{2\ev{a}_\tau})^{-1}},\label{eq:speed.limit2}
\end{align}
where $\Phi(x)$ is the inverse function of $x\tanh(x)$.
Equations \eqref{eq:speed.limit} and \eqref{eq:speed.limit2} are our new thermodynamic speed limits for classical Markov jump processes.

Some remarks are in order.
First, the thermodynamic speed limits in Eqs.~\eqref{eq:speed.limit} and \eqref{eq:speed.limit2} are tight and saturable.
More specifically, for generic initial and final distributions, we can always construct dynamics that satisfy the global detailed balance condition and transform the initial distribution into the final one in a time duration equal to that of the lower bounds.

Second, our bounds are tight for arbitrary temperatures, even in the zero-temperature limit.
Since $\ev{\sigma}_\tau=O(\beta)$, irreversible entropy production becomes infinite as $\beta\to+\infty$, whereas dynamical activity remains finite [i.e., $\ev{a}_\tau=O(1)$].
Nevertheless, we show in the following that our bounds remain useful in this low-temperature limit.
Indeed, in the $\beta\to+\infty$ limit, bound \eqref{eq:speed.limit2} reduces to a nontrivial inequality $\tau\ge\mca{W}_1(p_0,p_\tau)/\ev{a}_\tau$.
In addition, bound \eqref{eq:speed.limit} also remains finite because we can prove that $\beta\ev{ m}_\tau$ does not diverge in general.
To this end, we assume that the energy levels are non-degenerate and the system is typically driven far from the instantaneous equilibrium.
Since $ m_{xy}(t)$ can be calculated as
\begin{align}
\beta  m_{xy}(t)&=\frac{\beta[a_{xy}(t)-a_{yx}(t)]}{\ln a_{xy}(t)-\ln a_{yx}(t)}\notag\\
&=\frac{a_{xy}(t)-a_{yx}(t)}{\beta^{-1}[\ln p_y(t)-\ln p_x(t)]+\varepsilon_y(t)-\varepsilon_x(t)},
\end{align}
we have
\begin{equation}
\beta  m_{xy}(t)\xrightarrow{\beta\to+\infty}\frac{a_{xy}(t)-a_{yx}(t)}{\varepsilon_y(t)-\varepsilon_x(t)},
\end{equation}
which remains finite.
Therefore, the term $\beta\ev{ m}_\tau$ does not diverge in the zero-temperature limit $\beta\to +\infty$.

Third, we compare our results with existing bounds in the literature.
In Ref.~\cite{Shiraishi.2018.PRL}, a classical speed limit was obtained for Markov jump processes, which reads
\begin{equation}\label{eq:org.speed.limit}
\tau\ge\frac{\mca{T}(p_0,p_\tau)}{\sqrt{\ev{\sigma}_\tau\ev{a}_\tau/2}}.
\end{equation}
Since $\mca{W}_1(p_0,p_\tau)\ge \mca{T}(p_0,p_\tau)$ and $ m_t\le a_t/2$ for all times, our speed limits in Eq.~\eqref{eq:speed.limit} are stronger than those in Eq.~\eqref{eq:org.speed.limit}.
Our bounds also suggest that the conventional bound \eqref{eq:org.speed.limit} can be asymptomatically saturated only when $\mca{W}_1(p_0,p_\tau)=\mca{T}(p_0,p_\tau)$ (e.g., when the underlying graph is fully connected) and the system is always near the instantaneous equilibrium.
In Refs.~\cite{Vo.2022.arxiv,Delvenne.2021.arxiv}, another thermodynamic speed limit, which is tighter than the conventional bound \eqref{eq:org.speed.limit}, was derived as
\begin{equation}\label{eq:org.speed.limit2}
\tau\ge\frac{2\mca{T}(p_0,p_\tau)}{\ev{\sigma}_\tau\Phi({\ev{\sigma}_\tau}/{2\ev{a}_\tau})^{-1}}.
\end{equation}
Since $\mca{W}_1(p_0,p_\tau)\ge\mca{T}(p_0,p_\tau)$, our bound \eqref{eq:speed.limit2} is stronger than bound \eqref{eq:org.speed.limit2}.
The essential difference is that our bounds consider the topology of the jump process, whereas the conventional bounds do not.

\subsubsection{Quantum case}

Next, we consider an open quantum system described by the Markovian Lindblad master equation [Eq.~\eqref{eq:Lind.eq}].
Following the same procedure as in the classical case, we derive stringent thermodynamic bounds on the operational time required to transform the initial density matrix into the final one.

From Thm.~\ref{thm:qua.dis.Wass.var}
\begin{align}
\mca{W}_q(\varrho_0,\varrho_\tau)=\min_{\mca{L}_t}\qty{ \int_0^\tau\sqrt{\sigma_t m_t}\dd{t}}=\min_{\mca{L}_t}\sqrt{\Sigma_\tau\mca{M}_\tau},
\end{align}
we analogously obtain the following inequalities:
\begin{align}
\mca{W}_q(\varrho_0,\varrho_\tau)&\le \int_0^\tau\sqrt{\sigma_t m_t}\dd{t}=\tau\ev{\sqrt{\sigma m}}_\tau\\
&\le\tau\sqrt{\ev{\sigma}_\tau\ev{ m}_\tau}.
\end{align}
Consequently, we arrive at the following bounds on the operational time:
\begin{align}\label{eq:qsl1}
\tau\ge\frac{\mca{W}_q(\varrho_0,\varrho_\tau)}{\ev{\sqrt{\sigma m}}_\tau}\ge\frac{\mca{W}_q(\varrho_0,\varrho_\tau)}{\sqrt{\ev{\sigma}_\tau\ev{ m}_\tau}}. 
\end{align}
Equation \eqref{eq:qsl1} implies that the speed of state transformation in open quantum systems is constrained by irreversible entropy production and dynamical state mobility.
Notably, it has the same form as the classical bound \eqref{eq:speed.limit}, indicating the unification of our results.
Exploiting Cor.~\ref{cor:qua.dis.Wass.var1} in Appendix \ref{app:qua.cor.proof} and repeating the same procedure yield other speed limits in terms of entropy production and dynamical activity, which read
\begin{align}\label{eq:qsl2}
\tau\ge\frac{2\mca{W}_q(\varrho_0,\varrho_\tau)}{\ev{{\sigma}\Phi\qty({\sigma/2a})^{-1}}_\tau}\ge\frac{2\mca{W}_q(\varrho_0,\varrho_\tau)}{\ev{\sigma}_\tau\Phi\qty(\ev{\sigma}_\tau/2\ev{a}_\tau)^{-1}}.
\end{align}
Equations \eqref{eq:qsl1} and \eqref{eq:qsl2} are the new quantum thermodynamic speed limits.
Remarkably, these thermodynamic speed limits are tight and saturable.
In other words, for any pair of generic initial and final states, a combination of Hamiltonian and jump operators always exists that attains the lower bound of the operational time.
Moreover, they are useful even in the zero-temperature limit.
The bounds imply that both the thermodynamic and kinetic costs play a crucial role in the change speed of open quantum systems.

We discuss the relevance of our results to previous studies.
In Ref.~\cite{Vu.2021.PRL}, a thermodynamic speed limit was derived for Markovian open quantum dynamics and is given by
\begin{align}\label{eq:org.qsl}
\tau\ge\frac{\mca{W}_q(\varrho_0,\varrho_\tau)}{\sqrt{\ev{\sigma}_\tau\ev{a}_\tau/2}}. 
\end{align}
According to Eq.~\eqref{eq:lt.at.qua}, we have $ m_t\le a_t/2$ for all $t$. Therefore, bound \eqref{eq:org.qsl} is looser than the new bound \eqref{eq:qsl1}.
In addition, since $\Phi(x)\ge\sqrt{x}$ for all $x\ge 0$, the new bound \eqref{eq:qsl2} is also stronger than the conventional bound \eqref{eq:org.qsl}.
In Ref.~\cite{Funo.2019.NJP}, another thermodynamic speed limit in terms of trace distance was derived for open quantum systems.
Since the metrics used to measure the distance between quantum states in these bounds are different (i.e., the Wasserstein distance in our study and the trace distance in Ref.~\cite{Funo.2019.NJP}), they cannot be directly compared.
Nonetheless, by exploiting the dynamical state mobility introduced in this study, we can derive a similar but tighter speed limit in terms of the trace distance. The detailed form of this speed limit is presented in Appendix \ref{app:qsl.trace.dist}. However, it is worth noting that the attainability of this bound is unclear.

\subsection{Finite-time Landauer principle}

The Landauer principle \cite{Landauer.1961.JRD} implies that erasing information is always accompanied by a thermodynamic cost.
More specifically, the thermodynamic cost required to erase a classical bit is at least $T\ln 2$, where $T$ is the environment temperature.
The lower bound $T\ln 2$ (referred to as the Landauer bound) plays not only a fundamental role in the thermodynamics of information and computation \cite{Bennett.1982.IJTP,Sagawa.2012.PTP,Parrondo.2015.NP,Goold.2016.JPA,Wolpert.2019.JPA}, it also resolves the paradox of Maxwell's demon \cite{Maruyama.2009.RMP}.

Various classical and quantum platforms \cite{Brut.2012.N,Jun.2014.PRL,Yan.2018.PRL,Hong.2016.SA,Saira.2020.PRR,Dago.2021.PRL} have experimentally confirmed that the Landauer bound can be achieved in the slow quasistatic limit.
However, practical computing requires fast memory erasure in a short time and thus, in general, consumes a thermodynamic cost far beyond the Landauer bound.
This background strongly motivates researchers to develop finite-time generalizations of the Landauer bound, which capture finite-time corrections and can better predict the erasure cost.
Although several finite-time bounds have been developed for both classical and quantum discrete systems \cite{Zhen.2021.PRL,Vu.2022.PRL,Lee.2022.arxiv}, the attainability of these bounds remains unclear.
Moreover, these bounds have looser predictive power in the low-temperature regime.
In the following, we attempt to derive finite-time bounds that are tight for arbitrary temperatures.

Before presenting the new bounds, we first describe the generic setup of information erasure for both classical and quantum cases.
We consider a finite-dimensional discrete system attached to a thermal reservoir at temperature $T$.
Information is encoded in the system state and subsequently erased by controlling the classical energy levels or the quantum Hamiltonian and driving the system toward its ground state.
The erasure protocol should work for an arbitrary initial state; that is, any initial state should be reset close to the ground state in a finite time $\tau$.
This erasure process leads to a change in system entropy, which must be compensated for by the heat dissipated into the reservoir.
Because we are interested in the average thermodynamic cost associated with the erasure protocol, considering the maximally mixed state as the initial state is convenient.
Roughly speaking, the reasons for this are that the maximally mixed state is sufficient to understand the average dissipated heat of the erasure process for all initial states and that if a protocol can reliably reset the system from the maximally mixed state, then it does so also for an arbitrary state.
A detailed discussion will be given in the following.

\subsubsection{Classical case}

We consider an information erasure process using a $d$-state classical system, the dynamics of which is governed by the master equation.
The transitions between states are mediated by a single thermal reservoir at temperature $T$.
The system state is characterized by the probability distribution, which encodes information we want to erase.
The energy levels are controlled according to a fixed protocol such that the system is always driven toward the ground state ${p_{*}}=[1,0,\dots,0]^\top$, irrespective of the initial state.

Here, we explain why the initial state should be set to the uniform distribution ${\overline{p}}=[1/d,\dots,1/d]^\top$.
First, let $\mca{Q}(p_0)$ be the heat dissipation of erasure for the initial distribution ${p_0}$.
Then, due to the linearity of the master equation and $\mca{Q}(\cdot)$, the average dissipation can be calculated as
\begin{equation}\label{eq:avgQ}
\mbb{E}[\mca{Q}(p_0)]=\mca{Q}(\mbb{E}[p_0])=\mca{Q}(\overline{p}),
\end{equation}
where $\mbb{E}[\cdot]$ denotes the average over all possible initial distributions.
Equation \eqref{eq:avgQ} implies that investigating the case with initial distribution ${p_0}={\overline{p}}$ is sufficient to understand the average dissipation.
Second, let $\Lambda_\tau=\vec{T}\exp\qty( \int_0^\tau \msf{W}_t\dd{t})$ be the map that represents the erasure process, that is, $\Lambda_\tau{p_0}={p_\tau}$.
We can then prove that if the uniform distribution can be erased within error $\delta>0$ (i.e., $\|\Lambda_\tau{\overline{p}} - {p_{*}}\|_F\le\delta$), the following inequality holds for arbitrary initial distribution ${p_0}$ (see Appendix \ref{app:cla.erase.guar} for the proof):
\begin{equation}\label{eq:erasure.error.bound}
\|\Lambda_\tau{p_0} - {p_{*}}\|_F\le \sqrt{2d\delta}.
\end{equation}
Equation \eqref{eq:erasure.error.bound} indicates that if a protocol can erase the uniform distribution, it can reliably do so also for arbitrary initial states.

We can now present the new bound.
Let ${p_\tau}$ be the final distribution for the case ${p_0}={\overline{p}}$ and let $\epsilon\coloneqq\mca{T}(p_\tau,p_*)=|1-p_1(\tau)|$ be the erasure error, which should be sufficiently small.
From Eq.~\eqref{eq:Wass.bound.disc}, the heat dissipation is lower bounded by system entropy production and a finite-time correction term as
\begin{equation}\label{eq:cla.Landauer.prin}
Q\ge -T\Delta S_{\rm sys}+\frac{\mca{W}_1(p_0,p_\tau)^2}{\tau\beta\ev{ m}_\tau}.
\end{equation}
Equation \eqref{eq:cla.Landauer.prin} is regarded as the finite-time Landauer principle for classical systems.
The bound is tight and can be saturated for arbitrary temperatures, even in the zero-temperature limit.
As shown in the previous section, the term $\beta\ev{ m}_\tau$ remains finite even when $\beta\to\infty$ (i.e., $T\to 0$).
Therefore, bound \eqref{eq:cla.Landauer.prin} is useful for arbitrary temperatures.
By contrast, the conventional Landauer bound becomes trivial in the low-temperature regime (i.e., $Q\ge 0$).
We also note that bound \eqref{eq:cla.Landauer.prin} is tighter than the following bound:
\begin{equation}
Q\ge -T\Delta S_{\rm sys}+\frac{\mca{T}(p_0,p_\tau)^2}{\tau\beta\ev{a}_\tau/2},
\end{equation}
which is obtained from the conventional speed limit \eqref{eq:org.speed.limit}.

Bound \eqref{eq:cla.Landauer.prin} can be simplified by including the erasure error.
To this end, we further bound the terms in Eq.~\eqref{eq:cla.Landauer.prin} from below as
\begin{align}
-\Delta S_{\rm sys}&=\ln d-S(p_\tau)\ge \ln d-h(\epsilon),\\
\mca{W}_1(p_0,p_\tau)&\ge\mca{T}(p_0,p_\tau)\ge \qty|1-1/d-\epsilon|,
\end{align}
where $h(\epsilon)\coloneqq - \epsilon\ln[\epsilon/(d-1)] - (1-\epsilon)\ln(1-\epsilon)\ge 0$ is a function of $\epsilon$ that vanishes as $\epsilon\to 0$.
Consequently, we obtain the following bound on the average heat dissipation:
\begin{equation}\label{eq:fin.time.Landauer.error}
Q\ge T\qty[\ln d - h(\epsilon)] +\frac{(1-1/d-\epsilon)^2}{\tau\beta\ev{ m}_\tau}.
\end{equation}
Equation \eqref{eq:fin.time.Landauer.error} imposes a lower bound on heat dissipation in terms of the operational time and erasure error.
In the limit of perfect erasure (i.e., $\epsilon\to 0$), a simple bound can be derived:
\begin{equation}\label{eq:simple.Landauer.error}
Q\ge T\ln d+\frac{(1-1/d)^2}{\tau\beta\ev{ m}_\tau}.
\end{equation}
For slow erasure (i.e., $\tau\beta\ev{ m}_\tau\gg 1$), the second term in the lower bound vanishes. Thus, Eq.~\eqref{eq:simple.Landauer.error} recovers the conventional Landauer bound for the $d=2$ case.
By contrast, in the fast-erasure limit (i.e., $\tau\beta\ev{ m}_\tau\ll 1$), this correction term becomes dominant, implying that fast erasure is accompanied by a thermodynamic cost far beyond the Landauer cost.

If we consider dynamical activity instead of dynamical state mobility, we can obtain another finite-time Landauer bound. By transforming the speed limit \eqref{eq:speed.limit2}, we can show that heat dissipation is lower bounded by the Wasserstein distance and dynamical activity as
\begin{equation}\label{eq:cla.Landauer.prin2}
\frac{Q}{T}\ge -\Delta S_{\rm sys} + 2\mca{W}_1(p_0,p_\tau)\tanh^{-1}\qty(\frac{\mca{W}_1(p_0,p_\tau)}{\tau\ev{a}_\tau}).
\end{equation}
This new bound is always tighter than the bound reported in Ref.~\cite{Lee.2022.arxiv}, which uses the total variation distance to quantify the distance between probability distributions. In general, bound \eqref{eq:cla.Landauer.prin2} can be either stronger or looser than bound \eqref{eq:cla.Landauer.prin}.

\subsubsection{Quantum case}

Here, we consider a quantum process of erasing information.
The erasure process is implemented using a controllable $d$-dimensional qudit system, which is attached to a thermal reservoir at temperature $T$. 
The density matrix of the qudit encodes the information we want to erase and then is driven toward the ground state $\varrho_{*}=\dyad{0}$ by controlling the Hamiltonian.

Analogous to the classical case, the initial state is conveniently set to the maximally mixed state $\overline{\varrho}=\mbb{1}/d$.
This is because assigning the maximally mixed state to the initial state is sufficient to understand the average heat dissipation in the quantum case.
In addition, we can show that if an erasure protocol can erase the maximally mixed state, it can reliably do so also for arbitrary initial states.
More specifically, let $\Lambda_\tau(\cdot)=\vec{T}\exp\qty( \int_0^\tau \mca{L}_t\dd{t} )(\cdot)$ be the quantum map that describes the erasure process [i.e., $\Lambda_\tau(\varrho_0)=\varrho_\tau$].
Then, if the maximally mixed state can be erased within error $\delta>0$ (i.e., $\|\Lambda_\tau(\overline{\varrho}) - \varrho_{*}\|_F\le\delta$), the erasure error for arbitrary initial state $\varrho_0$ can be upper bounded as follows \cite{Vu.2022.PRL}:
\begin{equation}\label{eq:erasure.error.bound2}
\|\Lambda_\tau(\varrho_0) - \varrho_{*} \|_{F}\le \sqrt{2d\delta}.
\end{equation}
Equation \eqref{eq:erasure.error.bound2} provides insight into the reliability of the erasure protocol by verifying the $\varrho_0=\overline{\varrho}$ case.

Next, we present the finite-time bound for quantum information erasure.
Because $\varrho_0=\mbb{1}/d$, the Wasserstein distance coincides with the trace distance [i.e., $\mca{W}_q(\varrho_0,\varrho_\tau)=\mca{T}(\varrho_0,\varrho_\tau)$].
Consequently, a finite-time bound on heat dissipation can be obtained as
\begin{equation}\label{eq:qua.Landauer.prin}
Q\ge -T\Delta S_{\rm sys}+\frac{\mca{T}(\varrho_0,\varrho_\tau)^2}{\tau\beta\ev{ m}_\tau}.
\end{equation}
Equation \eqref{eq:qua.Landauer.prin} is the finite-time quantum Landauer principle, which is tight and can be saturated for arbitrary temperatures.
In addition to the conventional Landauer term $-T\Delta S_{\rm sys}$, a finite-time correction term exists in the lower bound, which does not vanish even in the zero-temperature limit. 
Therefore, the inequality \eqref{eq:qua.Landauer.prin} provides a stringent bound on heat dissipation for information erasure in both fast-driving and low-temperature regimes.

We compare bound \eqref{eq:qua.Landauer.prin} with an existing bound derived in Ref.~\cite{Vu.2022.PRL}, which reads
\begin{equation}\label{eq:qua.Landauer.prin.old}
Q\ge -T\Delta S_{\rm sys}+\frac{\mca{T}(\varrho_0,\varrho_\tau)^2}{\tau\beta\ev{a}_\tau/2}.
\end{equation}
Since $ m_t\le a_t/2$ for all times, it is immediately clear that bound \eqref{eq:qua.Landauer.prin} is always stronger than bound \eqref{eq:qua.Landauer.prin.old}.

Next, we derive a simplified bound which includes the erasure error.
Let $\epsilon\coloneqq\mca{T}(\varrho_\tau,\varrho_*)$ be the erasure error, which should be small.
We can then analogously bound the terms in Eq.~\eqref{eq:qua.Landauer.prin} from below as
\begin{align}
-\Delta S_{\rm sys}&=\ln d-S(\varrho_\tau)\ge \ln d-h(\epsilon),\label{eq:q.sysent.lb} \\
\mca{T}(\varrho_0,\varrho_\tau)&\ge\qty|\mca{T}(\varrho_0,\varrho_*)-\mca{T}(\varrho_\tau,\varrho_*)|=\qty|1-1/d-\epsilon|,\label{eq:q.trdist.lb}
\end{align}
where we use an inequality relating the entropy difference between two quantum states to their trace distance in the first line \cite{Audenaert.2007.JPA}.
Inserting Eqs.~\eqref{eq:q.sysent.lb} and \eqref{eq:q.trdist.lb} into Eq.~\eqref{eq:qua.Landauer.prin}, we arrive at the following simple bound:
\begin{equation}\label{eq:qua.Landauer.prin.sim}
Q\ge T[\ln d-h(\epsilon)]+\frac{(1-1/d-\epsilon)^2}{\tau\beta\ev{ m}_\tau}.
\end{equation}
Equation \eqref{eq:qua.Landauer.prin.sim} is the simplified Landauer bound that includes finite-time and finite-error corrections.
Remarkably, it has the same structure as the classical bound \eqref{eq:cla.Landauer.prin}.
In the limit of perfect and slow erasure, bound \eqref{eq:qua.Landauer.prin.sim} reduces to the conventional Landauer bound as $d=2$.

In analogy to the classical case, a finite-time Landauer bound in terms of the quantum Wasserstein distance and dynamical activity can also be obtained. By rearranging the speed limit \eqref{eq:qsl2}, we can prove that
\begin{equation}
\frac{Q}{T}\ge -\Delta S_{\rm sys} + 2\mca{W}_q(\varrho_0,\varrho_\tau)\tanh^{-1}\qty(\frac{\mca{W}_q(\varrho_0,\varrho_\tau)}{\tau\ev{a}_\tau}).
\end{equation}
This finite-time bound can be considered a quantum analog of the classical bound \eqref{eq:cla.Landauer.prin2}.

\section{Numerical demonstrations}\label{sec:num.demon}

Next, we numerically illustrate the applications of our results, the thermodynamic uncertainty relation, speed limits, and finite-time Landauer principles in several classical and quantum systems.

\begin{figure*}[t]
\centering
\includegraphics[width=1.0\linewidth]{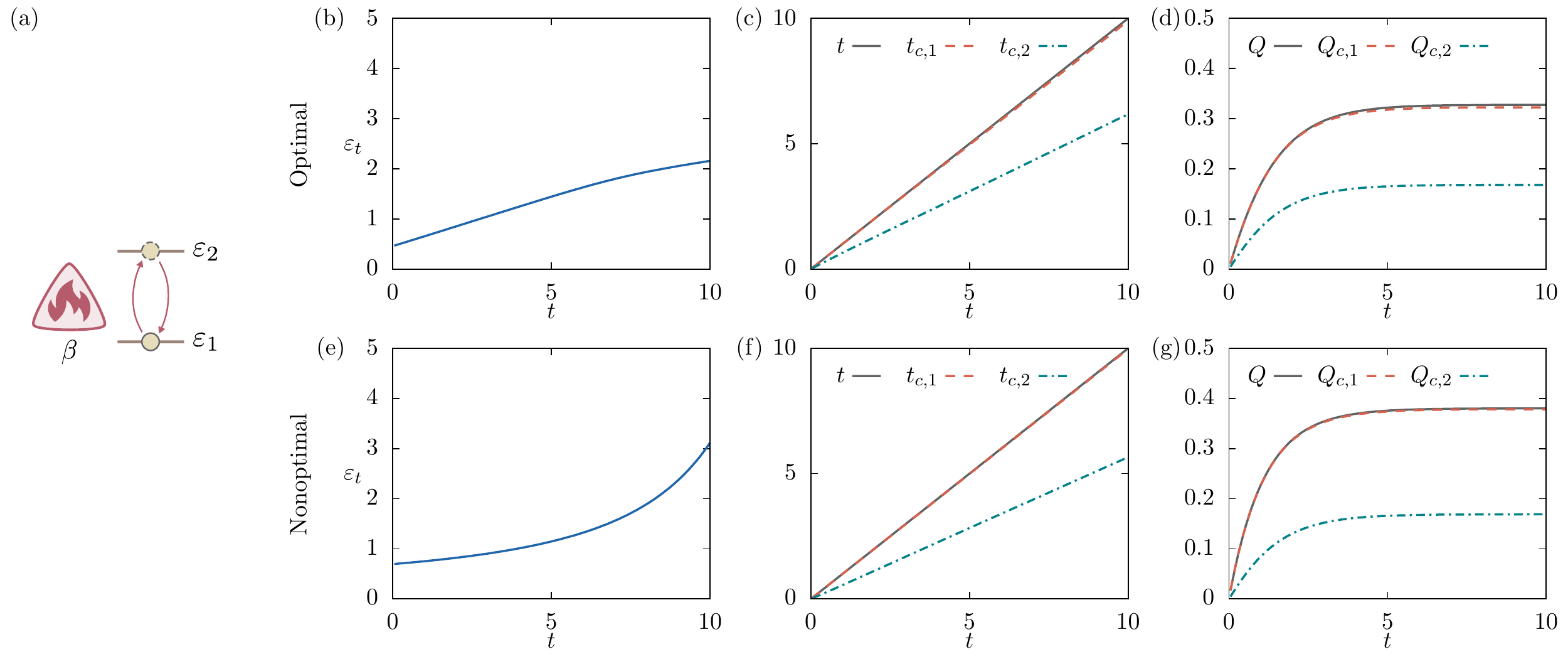}
\protect\caption{Numerical illustration of the classical thermodynamic speed limits and finite-time Landauer principle for both optimal and nonoptimal protocols. (a) Schematic of the two-level system. (b)(e) Time variations of the control parameters of the optimal and nonoptimal protocols. (c)(f) Numerical verification of the speed limits $t\ge t_{c,1}\ge t_{c,2}$ and (d)(g) finite-time Landauer principle $Q\ge Q_{c,1}\ge Q_{c,2}$. The other parameters are fixed as $\beta=10$ and $\tau=10$. The weighting factor $\lambda$ is set such that $(1-\lambda)/\lambda=10^5$.}\label{fig:CSLEx}
\end{figure*}

\subsection{Illustration of classical thermodynamic speed limits and finite-time Landauer principle}
First, we illustrate the classical speed limits and finite-time Landauer principle in a two-level system.
The system dynamics can be described by a Markov jump process with transition rates given by
\begin{equation}
w_{xy}(t)=\gamma_{xy}(t)\frac{e^{\beta\varepsilon_y(t)}}{e^{\beta\varepsilon_y(t)}+e^{\beta \varepsilon_x(t)}},
\end{equation}
where $\gamma_{xy}(t)=\gamma_{yx}(t)$ are tunable parameters and $\varepsilon_x(t)$ denotes the instantaneous energy level of state $x$.
For simplicity, we set $\varepsilon_t\coloneqq\varepsilon_2(t)-\varepsilon_1(t)$ and $\gamma_{xy}(t)=1$ for all transitions.
The parameter $\{\varepsilon_t\}$ thus defines a time-dependent control protocol.

To illustrate the bounds, we consider an information erasure process in which an arbitrary initial distribution is always reset to the ground state within a finite error.
We examine two control protocols, namely, optimal and nonoptimal.
The optimality here refers to the dissipation of the least amount of heat while achieving the predetermined error.

The optimal protocol is numerically obtained by solving the minimization problem with the following objective functional:
\begin{equation}\label{eq:cla.obj.func}
\mca{F}_c[\{\varepsilon_t\}]\coloneqq \lambda Q+(1-\lambda)\mca{T}(p_\tau,p_*),
\end{equation}
where $\lambda\in[0,1]$ is a weighting factor.
The functional $\mca{F}_c$ consists of two incompatible objectives, namely, heat dissipation and erasure error, which cannot be simultaneously small.
To reduce the erasure error, we must pay the price of dissipation. Conversely, reducing dissipation could enhance the error between the final and ground states.
As $\lambda$ is fixed, the solution of the optimization problem corresponds to a Pareto-optimal protocol, in which heat dissipation cannot be further minimized without increasing the error.
Imposing constraints on the control parameters is physically reasonable.
Hereafter, we consider the constraint $\varepsilon_t\in[10^{-2},10^1]$.
To solve the problem \eqref{eq:cla.obj.func} under both equality and inequality constraints, we discretize the control parameters into $1000$ points and optimize the functional $\mca{F}$ with the aid of nonlinear programming solvers.
We determine the weighting factor $\lambda$ such that both the optimal and nonoptimal protocols reset the uniform distribution $\overline{p}=[1/2,1/2]^\top$ to the ground state within the same error.
The time variation of the protocol is plotted in Fig.~\ref{fig:CSLEx}(b).
Notably, the increase in the energy gap between the two levels $\varepsilon_1$ and $\varepsilon_2$ is constant in the intermediate period but tends to slow down in the late period.

The nonoptimal protocol simply lifts the energy level $\varepsilon_2(t)$, forcing the system to descend to the ground state.
The time-dependent control parameter $\varepsilon_t$ is specified as
\begin{equation}
\varepsilon_t=0.422\times\exp\qty( \frac{\tau+t}{2\tau -t} ),
\end{equation}
which is illustrated in Fig.~\ref{fig:CSLEx}(e).
We can observe that, unlike the optimal protocol, the energy gap in the nonoptimal protocol rapidly increases in the late period.
In the final time, this naive protocol should dissipate more heat than the optimal protocol.

The process of information erasure is performed within the period $\tau$. At each time $t\le\tau$, according to Eqs.~\eqref{eq:speed.limit} and \eqref{eq:org.speed.limit}, the operational time is lower bounded as
\begin{equation}
t\ge \frac{\mca{W}_1(p_0,p_t)}{\sqrt{\ev{\sigma}_t\ev{ m}_t}}\eqqcolon t_{c,1}\ge\frac{\mca{T}(p_0,p_t)}{\sqrt{\ev{\sigma}_t\ev{a}_t/2}} \eqqcolon t_{c,2}.
\end{equation}
These bounds are numerically verified for the optimal and nonoptimal protocols in Figs.~\ref{fig:CSLEx}(c) and \ref{fig:CSLEx}(f), respectively.
As shown, the derived bound $t\ge t_{c,1}$ is tight and stronger than the existing bound $t\ge t_{c,2}$ for all times.

Likewise, the dissipated heat is lower bounded by the entropy change and finite-time correction term as
\begin{align}
Q&\ge -T\Delta S_{\rm sys} + \frac{\mca{W}_1(p_0,p_t)^2}{t\beta\ev{ m}_t}\eqqcolon Q_{c,1}\notag\\
&\ge -T\Delta S_{\rm sys} + \frac{\mca{T}(p_0,p_t)^2}{t\beta\ev{a}_t/2} \eqqcolon Q_{c,2}.
\end{align}
The numerical results are plotted in Figs.~\ref{fig:CSLEx}(d) and \ref{fig:CSLEx}(g) for the optimal and nonoptimal protocols, respectively.
As seen, the new lower bound $Q_{c,1}$ tightly bounds the dissipated heat $Q$ in both protocols, whereas the existing lower bound $Q_{c,2}$ is loose and does not provide a good prediction for heat dissipation.
Notice that the average heat dissipation at the final time $\tau=10$ is approximately $0.33$, which is far beyond the conventional Landauer bound $\beta^{-1}\ln 2\approx 0.069$. This implies that the finite-time correction is dominant over the entropy change in this case.

\subsection{Illustration of quantum thermodynamic speed limits and finite-time Landauer principle}
We next exemplify the quantum speed limits and finite-time Landauer principle with a simple model of information erasure using a spin-$1/2$ qubit.
The qubit is weakly attached to a heat bath at inverse temperature $\beta$.
The time evolution of the reduced density matrix can be described by the Lindblad equation with the Hamiltonian
\begin{equation}
H_t=\frac{\varepsilon_t}{2}\qty[\cos(\theta_t)\sigma_z+\sin(\theta_t)\sigma_x]
\end{equation}
and jump operators
\begin{align}
L_1(t)&=\sqrt{\gamma \varepsilon_t(n_t+1)}\dyad{0_t}{1_t},\\
L_2(t)&=\sqrt{\gamma \varepsilon_tn_t}\dyad{1_t}{0_t}.
\end{align}
Here, $\{\ket{0_t},\ket{1_t}\}$ are the instantaneous energy eigenstates, $\sigma_{x,y,z}$ are the Pauli matrices, $\gamma$ is the coupling strength, $n_t\coloneqq 1/(e^{\beta \varepsilon_t}-1)$, and $\varepsilon_t$ and $\theta_t$ are time-dependent control parameters.
More specifically, $E_t$ characterizes the energy gap between the energy eigenstates, whereas $\theta_t$ quantifies the relative strength of coherent tunneling to energy bias \cite{Leggett.1987.RMP}.
The qubit is initially prepared in the state $\varrho_0=\mbb{1}/2$ and subsequently driven toward the ground state $\varrho_*=\dyad{0}$ of $\sigma_z$.
If $\theta_t$ is time invariant, quantum coherence in the energy eigenstates cannot be created, and the protocol is thus classical. Otherwise, it becomes a quantum protocol.

\begin{figure*}[t]
\centering
\includegraphics[width=1.0\linewidth]{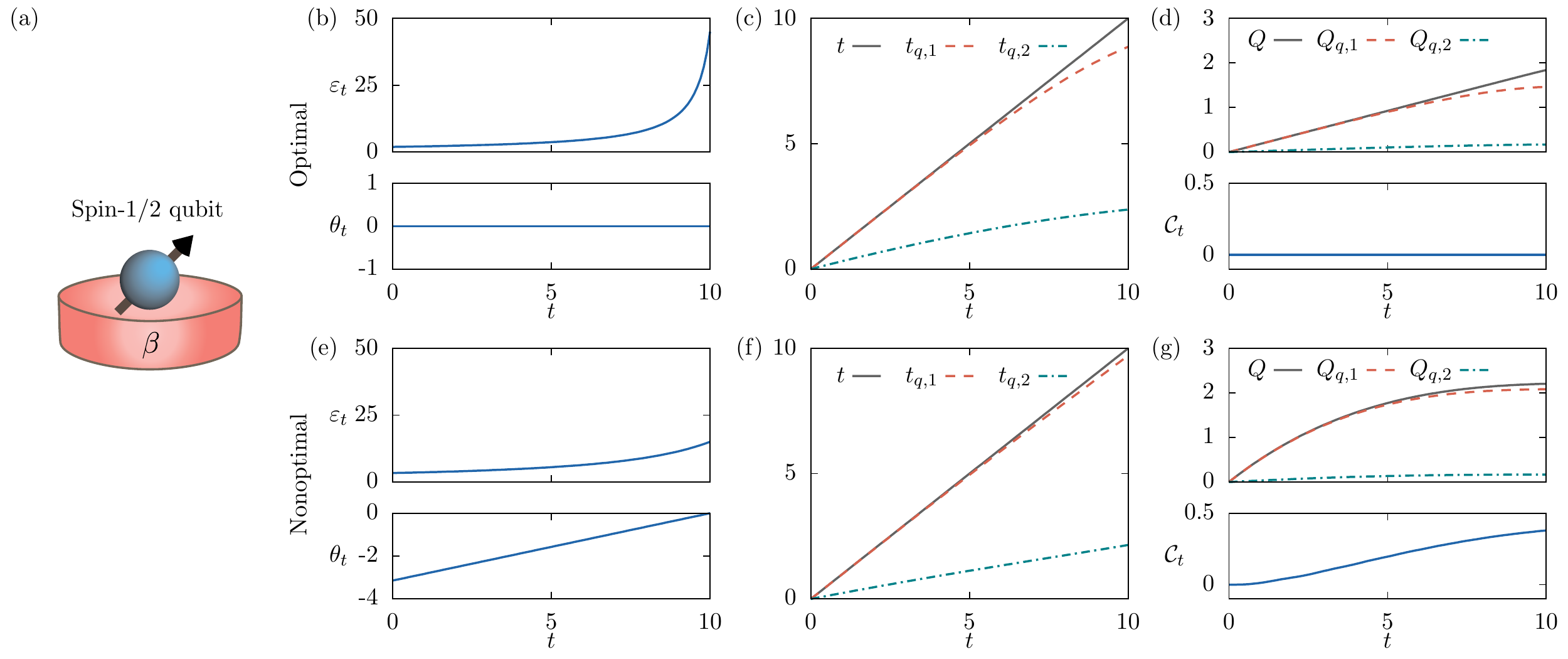}
\protect\caption{Numerical illustration of the quantum thermodynamic speed limits and finite-time Landauer principle for both optimal and nonoptimal protocols. (a) Schematic of the spin-$1/2$ qubit, which is coupled to a heat bath at inverse temperature $\beta$. (b)(e) Time variations of the control parameters of the optimal and nonoptimal protocols. (c)(f) Numerical verification of the speed limits $t\ge t_{q,1}\ge t_{q,2}$ and (d)(g) finite-time Landauer principle $Q\ge Q_{q,1}\ge Q_{q,2}$ as well as the amount of quantum coherence $\mca{C}_t$. The other parameters are fixed as $\beta=10$, $\gamma=0.1$, and $\tau=10$. The weighting factor $\lambda$ is set such that $(1-\lambda)/\lambda=10^3$.}\label{fig:QSLEx}
\end{figure*}

An infinite number of approaches can be used to reset the qubit with a probability close to $1$.
As in the classical case, two protocols are considered, namely, Pareto-optimal and nonoptimal.
Both protocols are designed to erase information with the same error.

The optimal protocol minimizes two incompatible objectives, the average dissipated heat and erasure error.
Specifically, the protocol can be achieved by solving the minimization problem with the following multi-objective functional \cite{Vu.2022.PRL}:
\begin{equation}
\mca{F}_q[\{\varepsilon_t,\theta_t\}]\coloneqq \lambda Q-(1-\lambda)F(\varrho_\tau,\varrho_*),\label{eq:qua.opt.func}
\end{equation}
where $\lambda\in[0,1)$ is a weighting factor and $F(\varrho,\sigma)\coloneqq(\tr |\sqrt{\varrho}\sqrt{\sigma}|)^2$ is the fidelity of the two quantum states $\varrho$ and $\sigma$ \cite{Jozsa.1994.JMO}.
Because of the physical limitations, placing constraints on the control parameters is natural.
Hereafter, we impose the following lower and upper bounds on the parameters: $\varepsilon_t\in [10^{-1},10^2]$ and $\theta_t\in[-\pi,\pi]$.
By numerically solving the nonlinear optimization problem \eqref{eq:qua.opt.func}, we can obtain the optimal protocol, as plotted in Fig.~\ref{fig:QSLEx}(b).
As seen, the parameter $\theta_t$ is fixed to $0$ for all times, implying that the optimal protocol is classical and does not generate any amount of energetic coherence.
Furthermore, the energy gap increases gradually in the intermediate period and changes rapidly in the final stage.

The nonoptimal protocol is defined as
\begin{equation}
\varepsilon_t=2.04\times\exp\qty(\frac{\tau+t}{2\tau-t}),~\theta_t=\pi\qty(\frac{t}{\tau}-1),
\end{equation}
which is plotted in Fig.~\ref{fig:QSLEx}(e).
This protocol naively increases the energy gap while varying the coherent parameter.
We also observe that the increase in the energy gap is different from that of the optimal case.

We first demonstrate the quantum thermodynamic speed limits.
According to Eqs.~\eqref{eq:qsl1} and \eqref{eq:org.qsl}, the operational time is lower bounded as follows:
\begin{equation}
t\ge \frac{\mca{W}_q(\varrho_0,\varrho_t)}{\sqrt{\ev{\sigma}_t\ev{ m}_t}}\eqqcolon t_{q,1}\ge\frac{\mca{W}_q(\varrho_0,\varrho_t)}{\sqrt{\ev{\sigma}_t\ev{a}_t/2}} \eqqcolon t_{q,2}.
\end{equation}
We illustrate these bounds for both the optimal and nonoptimal protocols in Figs.~\ref{fig:QSLEx}(c) and \ref{fig:QSLEx}(f), respectively.
As shown, the derived bound $t\ge t_{q,1}$ is tight for all times and is stronger than the existing bound $t\ge t_{q,2}$.

Next, we verify the finite-time quantum Landauer principle.
The lower bounds on the average heat dissipation are given by Eqs.~\eqref{eq:qua.Landauer.prin} and \eqref{eq:qua.Landauer.prin.old} as
\begin{align}
Q&\ge -T\Delta S_{\rm sys} + \frac{\mca{T}(\varrho_0,\varrho_t)^2}{t\beta\ev{ m}_t}\eqqcolon Q_{q,1}\notag\\
&\ge -T\Delta S_{\rm sys} + \frac{\mca{T}(\varrho_0,\varrho_t)^2}{t\beta\ev{a}_t/2} \eqqcolon Q_{q,2}.
\end{align}
The numerical results are plotted in Figs.~\ref{fig:QSLEx}(d) and \ref{fig:QSLEx}(g) for the optimal and nonoptimal protocols, respectively.
As shown, the new bound $Q\ge Q_{q,1}$ is tight for all times, whereas the existing bound $Q\ge Q_{q,2}$ is loose.
The optimal protocol also clearly dissipates less heat than does the nonoptimal protocol at the final time $\tau=10$.
In addition, note that the average heat dissipation in both protocols is approximately $2$, which is significantly greater than the conventional Landauer bound $\beta^{-1}\ln 2\approx 0.069$.

We discuss the effect of quantum coherence in the finite-time erasure process. For convenience, we quantify the amount of energetic coherence using the $\ell_1$-norm, which is one of the most general coherence monotones in the literature \cite{Baumgratz.2014.PRL}:
\begin{equation}
	\mca{C}_t\coloneqq\int_0^t[|\mel{0_s}{\varrho_s}{1_s}|+|\mel{1_s}{\varrho_s}{0_s}|]\dd{s}.
\end{equation}
That is, $\mca{C}_t$ is the time integral of the sum of absolute off-diagonal elements of quantum states in the basis of energy eigenstates.
Since $\theta_t$ is invariant in the optimal protocol, the instantaneous energy eigenstates remain unchanged, and the density matrix is always diagonal in the eigenstates. Therefore, the total amount of quantum coherence generated is always zero (i.e., $\mca{C}_t=0$), which is plotted in Fig.~\ref{fig:QSLEx}(d).
On the other hand, for the nonoptimal protocol, $\theta_t$ varies over time and quantum coherence is generally generated. The positive finite value of $\mca{C}_t$ can be confirmed from Fig.~\ref{fig:QSLEx}(g).
We can observe that the nonoptimal protocol that generates coherence is more dissipative than the optimal protocol that does not create coherence.
This is consistent with the fact that the creation of quantum coherence leads to unavoidable dissipation \cite{Vu.2022.PRL}.
This also suggests that a lower bound of heat dissipation that can capture the effect of coherence is desirable.

It is therefore worthwhile discussing how quantum coherence effects can be captured by the bounds.
According to the definition \eqref{eq:qua.dyn.act.def} of quantum dynamical activity, $a_t$ can be explicitly expressed in terms of the diagonal elements of $\varrho_t$ in the basis of energy eigenstates as
\begin{equation}
	a_t=\gamma\varepsilon_t\qty[ n_t\mel{0_t}{\varrho_t}{0_t} + (n_t+1)\mel{1_t}{\varrho_t}{1_t} ].
\end{equation}
As seen, $a_t$ has no coherent contribution from the off-diagonal part of quantum state $\varrho_t$.
On the other hand, quantum dynamical state mobility $m_t$ implicitly includes a coherent contribution.
This can be validated from Eq.~\eqref{eq:quan.kinetic.cost} by noticing that the eigenbasis $\{\ket{x_t}\}$ of the spectral decomposition $\varrho_t=\sum_xp_x(t)\dyad{x_t}$ is generally different from the energy eigenstates $\{\ket{0_t},\ket{1_t}\}$.
Although the coherence present in the quantum state $\varrho_t$ may contribute to the bounds through the boundary terms $\Delta S_{\rm sys}$ and $\mca{T}(\varrho_0,\varrho_t)$, such contributions are negligibly small as compared to the time-extensive contribution $t\ev{m}_t$.
Therefore, the derived bound $Q\ge Q_{q,1}$ can capture the effect of quantum coherence occurring during the process and precisely predict heat dissipation even in the presence of quantum coherence, whereas the existing bound $Q\ge Q_{q,2}$ cannot.

\section{Conclusion and outlook}\label{sec:conc.disc}
In this study, we elucidated an intimate relationship between thermodynamics and discrete optimal transport for both classical and quantum cases.
To this end, we introduced a novel physical term, namely, dynamical state mobility, which characterizes a complementary aspect of irreversible entropy production in the time evolution of a system.
By deriving an improved thermodynamic uncertainty relation, we showed that dynamical state mobility plays a critical role in constraining the fluctuation of time-antisymmetric currents, thus providing insight into the precision of currents in Markov jump processes.
Exploiting this term, we derived variational formulas that express the discrete Wasserstein distance in terms of the thermodynamic cost associated with Markovian dynamics.
These formulas not only unify the relationship between thermodynamics and optimal transport for both discrete and continuous cases but also generalize to the Markovian quantum dynamics.
From the variational formulas, we derived stringent thermodynamic speed limits and the finite-time Landauer principle.
The obtained bounds are tight and can be saturated for an arbitrary pair of initial and final states and arbitrary temperatures.

Our theoretical frameworks also shed light on the minimization problem of entropy production in discrete Markov dynamics.
Recent studies \cite{Vu.2021.PRL2,Remlein.2021.PRE,Dechant.2022.JPA} have shown that entropy production can be optimized to be arbitrarily small if there are no constraints on the transition rates.
Our results suggest that dynamical state mobility may be a reasonable constraint because once it is fixed, minimum entropy production is immediately determined by the discrete Wasserstein distance.
In addition, the optimal protocol that attains the minimum entropy production can be constructed from the optimal coupling between the initial and final distributions, which can be numerically computed in an efficient manner.

Although not explicitly stated in this study, our framework is also applicable to bipartite systems \cite{Horowitz.2014.PRX}, in which two subsystems exchange information.
In this case, the Wasserstein distance between the initial and final distributions of a subsystem can be expressed in terms of the entropy production of that subsystem and the information flow with another subsystem.

Our study opens several possible directions for future research, which are as follows.
\begin{itemize}
	\item[(1)] \emph{Generalizing the formulations to include measurement and feedback control}.---Measurement and feedback control are ubiquitous in physics and biology. The thermodynamics of feedback control \cite{Sagawa.2012.PRE} has been intensively developed in recent years. In this study, we focused exclusively on discrete Markovian systems subjected to deterministic control protocols. Extending our framework to include the effects of measurement and feedback control would be significant, as it would provide a better understanding of the role of information in nonequilibrium systems. Because information obtained from measurements can enhance the precision of observables \cite{Vu.2020.JPA,Potts.2019.PRE} and would violate the second law of thermodynamics \cite{Sagawa.2008.PRL}, in addition to entropy production, information would be expected to play a crucial role in the speed of state transformation and heat dissipation of finite-time information erasure.
	\item[(2)] \emph{Decomposition of entropy production}.---Decomposing entropy production is theoretically appealing because it provides insight into the dissipative structure of thermodynamic processes. Previous studies have shown that irreversible entropy production of Markovian dynamics could be split into an adiabatic and non-adiabatic contribution in both discrete and continuous cases, which originates from the breaking of detailed balance \cite{Hatano.2001.PRL,Esposito.2010.PRL}. For overdamped Langevin dynamics, recent studies \cite{Maes.2013.JSP,Dechant.2022.PRR} have introduced a new decomposition of the entropy production rate in terms of the continuous Wasserstein distance and a housekeeping entropy production rate as
	\begin{equation}
		\sigma_t = \frac{1}{D}\qty(\lim_{dt\to 0}\frac{W_2(p_t,p_{t+dt})}{dt})^2 + {\sigma}_t^{\rm hk}.
	\end{equation}
	The term $\sigma_t^{\rm hk}$ vanishes as the system is driven by a conservative force.
	Inspired by this decomposition, the entropy production rate of Markov jump processes can be split in a similar manner as
	\begin{equation}
		\sigma_t = \frac{1}{ m_t}\qty(\lim_{dt\to 0}\frac{\mca{W}_1(p_t,p_{t+dt})}{dt})^2 + \tilde{\sigma}_t.
	\end{equation}
	The term $\tilde{\sigma}_t$ is nonnegative and vanishes only when the system is driven by an optimal protocol, provided that $m_t$ is fixed.
	Investigating the properties of the contribution $\tilde{\sigma}_t$ would be an interesting direction and may lead to a deep understanding of dissipation in Markov jump processes.
	\item[(3)] \emph{Application to deterministic biochemical reaction networks}.---Although our framework deals with stochastic dynamics, generalizing the formulas to cases of deterministic dynamics such as biochemical reaction networks \cite{Rao.2016.PRX} would be interesting. This is feasible because our results are derived from the master equation, which is similar to the deterministic rate equation characterizing the time evolution of biochemical reaction networks.
	\item[(4)] \emph{Thermodynamic interpretation of the discrete $L^2$-Wasserstein distance}.---Thus far, we have investigated the connection between thermodynamics and optimal transport through the discrete $L^1$-Wasserstein distance. Although we showed that the discrete $L^1$-Wasserstein distance has aspects similar to the continuous $L^2$-Wasserstein distance, it remains an open question whether a thermodynamic interpretation exists for the discrete $L^2$-Wasserstein distance. Clarification of this interpretation is desirable and could lead to new fundamental thermodynamic bounds.
	\item[(5)] \emph{Formulation under constrained control protocols}.---In this study, we thermodynamically interpreted the discrete Wasserstein distance using Markov jump processes whose transition rates can be arbitrarily controlled without any constraint. However, in practice, some constraints may be imposed on the transition rates and protocols \cite{Kolchinsky.2021.PRX,Remlein.2021.PRE,Dechant.2022.JPA,Abiuso.2022.JPC}. Developing analogous formulas for these settings would be highly relevant and broaden the range of applications. The specific form of transition rates given in Eqs.~\eqref{eq:1d.dis.tran.rate1} and \eqref{eq:1d.dis.tran.rate2} also suggests that investigating this direction may reveal the thermodynamic role of the discrete $L^2$-Wasserstein distance.
\end{itemize}

\begin{acknowledgements}
We thank Shin-ichi Sasa and Andreas Dechant for the fruitful discussion. We are also grateful to Amos Maritan for valuable comments.
This work was supported by Grants-in-Aid for Scientific Research (JP19H05603 and JP19H05791).
\end{acknowledgements}

\appendix

\section{Geometric property of continuous $L^2$-Wasserstein distance}\label{app:W2.geo.prop}

Here, we discuss a geometric interpretation of the Wasserstein distance.
Specifically, we show that the $L^2$-Wasserstein distance can be interpreted as a Riemannian distance on the infinite-dimensional manifold $M$ of probability distribution functions.
For each distribution function $p(x)$, the tangent velocity space $\Tan_pM$ at point $p$ can be defined as \cite{Ambrosio.2008}
\begin{equation}\label{eq:tan.vel.spa}
\Tan_pM\coloneqq \qty{w(x)\,\big|\,\nabla\vdot[w(x)p(x)]=0}^\perp.
\end{equation}
Here, $\Omega^\perp$ denotes the orthogonal complement of a subspace $\Omega$.
In other words, $\Tan_pM$ contains all velocity fields $v(x)$ that satisfy
\begin{equation}
\int_{\mbb{R}^d}[v(x)\vdot w(x)]p(x)\dd{x}=0~\forall w~\text{s.t.}~\nabla\vdot(wp)=0.
\end{equation}
The tangent space can be indirectly defined via the tangent velocity space as
\begin{equation}
T_pM\coloneqq \qty{u\,\big|\,\exists v\in\Tan_pM~{\rm s.t.}~u+\nabla\vdot(vp)=0}.
\end{equation}
Consequently, a Riemannian metric $g_p:T_pM\times T_pM\to\mbb{R}$ can be defined on the tangent space as
\begin{equation}
g_p(u_1,u_2)\coloneqq\int_{\mbb{R}^d}[v_1(x)\vdot v_2(x)]p(x)\dd{x},\label{eq:Wass.Riemannian.metric}
\end{equation}
where $v_i$ is the velocity field corresponding to the tangent vector $u_i$ ($i=1,2$).
We can then show that $W_2(p^A,p^B)$ is exactly the geodesic distance between $p^A$ and $p^B$ induced by the defined metric:
\begin{align}
W_2(p^A,p^B)&=\min_{p_t}\qty{\tau\int_0^\tau g_{p_t}(\dot p_t,\dot p_t)\dd{t}}^{1/2}\notag\\
&=\min_{p_t}{\int_0^\tau \sqrt{g_{p_t}(\dot p_t,\dot p_t)}\dd{t}}.\label{eq:Wdiss2.metric}
\end{align}
Here, we consider the fact that the geodesic distance between two points is equal to the minimum square root of the divergence taken over all possible paths connecting those points.

In general, obtaining a closed form for $W_2(p^A,p^B)$ is difficult, except in the case in which $p^A$ and $p^B$ are normal distributions.
Therefore, a lower bound on $W_2$ is often considered.
It has been previously proved that $W_2(p^A,p^B)$ can be bounded from below by the means and covariances of distributions $p^A$ and $p^B$ as \cite{Gelbrich.1990.MN}
\begin{align}
W_2(p^A,p^B)^2&\ge\|\mu_A-\mu_B\|^2\notag\\
&+\tr{\Xi_A+\Xi_B-2\sqrt{\sqrt{\Xi_A}\Xi_B\sqrt{\Xi_A}}},
\end{align}
where $\mu_X$ and $\Xi_X$ are the mean and covariance matrices, respectively, of the probability distribution $p^X$ for $X\in\{A,B\}$.

\section{Useful propositions}

\begin{proposition}\label{prop:arra.ine}
Let $x=[x_1,\dots,x_n]^\top$ and $y=[y_1,\dots,y_n]^\top$ be vectors of real numbers and $\sigma$ be a permutation of $\{1,\dots,n\}$ such that if $x_i>x_j$, then $y_{\sigma(i)}\ge y_{\sigma(j)}$ for any $i$ and $j$.
Then, the following inequality holds:
\begin{equation}\label{eq:arra.ine1}
\sum_{j=1}^n|x_j-y_j|\ge\sum_{j=1}^n|x_j-y_{\sigma(j)}|.
\end{equation}
\end{proposition}
\begin{proof}
Without loss of generality, we assume $x_1\le x_2 \le \dots \le x_n$.
Let $\nu$ be a permutation of $\{1,\dots,n\}$ such that $\sum_{j=1}^n|x_j-y_{\nu(j)}|$ is minimum among all possible permutations and that the number of inversion pairs (i.e., $i > j$ and $y_{\nu(i)} < y_{\nu(j)}$) is minimum.
Assume that two indices $i$ and $j=i-1$ exist such that $x_i\ge x_j$ and $y_{\nu(i)}< y_{\nu(j)}$.
We then consider a new permutation $\nu'$ obtained from $\nu$ by swapping $\nu(i)$ and $\nu(j)$, that is, $\nu'(i)=\nu(j)$, $\nu'(j)=\nu(i)$, and $\nu'(k)=\nu(k)$ for all $k\neq i,j$. 
In this case, we can easily prove that
\begin{equation}
|x_i-y_{\nu(i)}|+|x_j-y_{\nu(j)}|\ge |x_i-y_{\nu'(i)}|+|x_j-y_{\nu'(j)}|.
\end{equation}
This means that 
\begin{equation}
\sum_{j=1}^n|x_j-y_{\nu(j)}|\ge\sum_{j=1}^n|x_j-y_{\nu'(j)}|,
\end{equation}
and the permutation $\nu'$ has fewer inversion pairs than the permutation $\nu$, which contradicts the optimality of the permutation $\nu$.
Therefore, we have $y_{\nu(i)}\ge y_{\nu(j)}$ for any $x_i>x_j$.
Consequently, the permutations $\sigma$ and $\nu$ satisfy
\begin{equation}
\sum_{j=1}^n|x_j-y_{\sigma(j)}|=\sum_{j=1}^n|x_j-y_{\nu(j)}|,
\end{equation}
from which Eq.~\eqref{eq:arra.ine1} is immediately proved because of the optimality of the permutation $\nu$.
\end{proof}

\begin{proposition}\label{prop:pre.mat}
Let $x=[x_1,\dots,x_n]^\top$ and $y=[y_1,\dots,y_{n'}]^\top$ be vectors of nonnegative numbers. If $\sum_{i=1}^nx_i=\sum_{j=1}^{n'}y_j$, then a matrix $\msf{Z}=[z_{ij}]\in\mbb{R}^{n\times n'}$ exists with nonnegative elements such that
\begin{equation}\label{eq:cond}
\sum_{j=1}^{n'}z_{ij}=x_i~\text{and}~\sum_{i=1}^nz_{ij}=y_j.
\end{equation}
\end{proposition}
\begin{proof}
We prove by induction on $k=n+n'\ge 2$.
The $k=2$ case is evident since $n=n'=1$ and $x_1=y_1$; therefore, we can choose $z_{11}=x_1$.
Supposing that it holds for all $k\le\bar{k}$, we can consider an arbitrary case with $k=\bar{k}+1$.
Let $v=\min\qty{x_1,y_1}$ and set $z_{11}=v$.
Without loss of generality, we can assume that $v=x_1$. Then, $z_{1i}=0$ for all $i\ge 2$.
Consider two vectors $x'=[x_2,\dots,x_n]^\top$ and $y'=[y_1-x_1,\dots,y_{n'}]^\top$ with $k'=n+n'-1=\bar{k}$.
A matrix $\msf{Z}'=[z_{ij}']\in\mbb{R}^{(n-1)\times n'}$ exists such that 
\begin{equation}
\sum_{j=1}^{n'}z_{ij}'=x_i'~\text{and}~\sum_{i=1}^{n-1}z_{ij}'=y_j'.
\end{equation}
Set $z_{ij}=z_{(i-1)j}'$ for all $i\ge 2$. Then, the matrix $\msf{Z}$ satisfies Eq.~\eqref{eq:cond}.
\end{proof}

\begin{proposition}\label{prop:add.ine0}
For arbitrary real numbers $\{x_i\}$ and $\{y_i\}$ that satisfy $x_iy_i\ge 0$ for all $i$, the following inequality holds:
\begin{equation}\label{eq:add.ine1}
\sum_i\frac{x_i}{y_i}\sum_i x_iy_i\ge\sum_i\frac{x_i}{\coth(y_i/2)}\sum_ix_i\coth(y_i/2).
\end{equation}
\end{proposition}
\begin{proof}
The inequality \eqref{eq:add.ine1} is equivalent to
\begin{align}
\sum_{i>j}\bigg[ \frac{x_i}{y_i}x_jy_j &+ \frac{x_j}{y_j}x_iy_i - \frac{x_i}{\coth(y_i/2)}x_j\coth(y_j/2)\notag\\
& - \frac{x_j}{\coth(y_j/2)}x_i\coth(y_i/2) \bigg] \ge 0.
\end{align}
It suffices to prove that each term in the above summation is nonnegative, that is,
\begin{equation}\label{eq:add.ine2}
x_ix_j\qty[\frac{y_j}{y_i} + \frac{y_i}{y_j} - \frac{\coth(y_j/2)}{\coth(y_i/2)} - \frac{\coth(y_i/2)}{\coth(y_j/2)}] \ge 0.
\end{equation}
Since $x_ix_jy_iy_j\ge 0$, Eq.~\eqref{eq:add.ine2} is equivalent to
\begin{equation}\label{eq:add.ine3}
y_i^2+y_j^2 - y_iy_j\qty[\frac{\coth(y_j/2)}{\coth(y_i/2)} + \frac{\coth(y_i/2)}{\coth(y_i/2)}] \ge 0.
\end{equation}
Since $y\coth(y)$ and $y/\coth(y)$ are even functions, we can assume that $y_i\ge y_j\ge 0$ without loss of generality.
The inequality \eqref{eq:add.ine3} can be rewritten as
\begin{align}
&\frac{y_i}{y_j}+\frac{y_j}{y_i}-\qty[\frac{\coth(y_j/2)}{\coth(y_i/2)} + \frac{\coth(y_i/2)}{\coth(y_j/2)}]\notag \\
&= \qty[\frac{y_i}{y_j} - \frac{\coth(y_j/2)}{\coth(y_i/2)}]\qty[1 - \qty{\frac{y_i}{y_j}\frac{\coth(y_j/2)}{\coth(y_i/2)}}^{-1}] \ge 0.\label{eq:add.ine4}
\end{align}
Since $\coth(y)$ is a strictly decreasing function over $[0,+\infty)$, we have $y_i\coth(y_j/2)\ge y_j\coth(y_i/2)$. Therefore, Eq.~\eqref{eq:add.ine4} is equivalent to
\begin{equation}\label{eq:add.ine5}
\frac{y_i}{y_j}\ge\frac{\coth(y_j/2)}{\coth(y_i/2)}.
\end{equation}
The inequality \eqref{eq:add.ine5} is always valid since $y\coth(y/2)$ is an increasing function over $[0,+\infty)$.
Therefore, Eq.~\eqref{eq:add.ine1} is proved.

\end{proof}

\section{Derivation of calculations in Sec.~\ref{sec:sto.the.dis.sys}}

\subsection{Property of $m_t$}\label{app:mt.prop}
Here, we show a relevant property of dynamical state mobility in terms of optimizing irreversible entropy production.
\begin{lemma}\label{lem:mobi.ult}
For any Markov jump process $\{p_t;\msf{W}_t\}_{0\le t\le\tau}$ and arbitrary positive constant $\bar{D}$, a Markov process $\{\tilde{p}_t;\tilde{\msf{W}}_t\}_{0\le t\le\tau}$ exists that simultaneously satisfies the following conditions:
\begin{itemize}
\item[(i)] The time evolution of probability distribution is the same (i.e., $\tilde{p}_t=p_t$ for all times).
\item[(ii)] The time-averaged state mobility is equal to $\bar{D}$ (i.e., $\ev{\tilde{ m}}_\tau=\bar{D}$).
\item[(iii)] The associated product of entropy production and dynamical state mobility is smaller than that of the original process (i.e., $\tilde{\Sigma}_\tau\tilde{\mca{M}}_\tau\le \Sigma_\tau\mca{M}_\tau$).
\end{itemize}
\end{lemma}
\begin{proof}
We consider another Markov jump process with the transition rate matrix $\tilde{\msf{W}}_t$ defined in the following manner. For each transition rate $w_{xy}(t)>0$, we define
\begin{equation}
\tilde{w}_{xy}(t)=w_{xy}(t)+\frac{\alpha_{xy}(t)}{p_y(t)}\quad (x\neq y),
\end{equation}
where $\alpha_{xy}(t)=\alpha_{yx}(t)$ are real coefficients to be later determined.
We can easily verify that $\tilde{j}_{xy}(t)=j_{xy}(t)$ for all $x\neq y$.
Therefore, given that the initial distribution is the same (i.e., $\tilde{p}_0=p_0$), we immediately obtain $\tilde{p}_t=p_t$ for all $t$, which fulfills condition (i).
In addition, for any $\lambda>0$, $\alpha_{xy}(t)$ always exists such that
\begin{equation}\label{eq:tilde.lmn}
\tilde{m}_{xy}(t)=\frac{\tilde{j}_{xy}(t)}{\ln[{\tilde{w}_{xy}(t)\tilde{p}_y(t)}]-\ln[{\tilde{w}_{yx}(t)\tilde{p}_x(t)}]}=\lambda|j_{xy}(t)|.
\end{equation}
This is because the following quantity can take an arbitrary positive value depending on the manner in which $\alpha_{xy}(t)$ is chosen:
\begin{equation}
\frac{w_{xy}(t)p_y(t)-w_{yx}(t)p_x(t)}{\ln\dfrac{w_{xy}(t)p_y(t)+\alpha_{xy}(t)}{w_{yx}(t)p_x(t)+\alpha_{yx}(t)}}.
\end{equation}
Choosing $\alpha_{xy}(t)$ such that Eq.~\eqref{eq:tilde.lmn} is satisfied, and setting $\lambda=\bar{D}\tau\qty[\int_0^\tau\sum_{x>y}|j_{xy}(t)|\dd{t}]^{-1}$, we can calculate
\begin{align}
\ev{\tilde{m}}_\tau&=\tau^{-1}\int_0^\tau\sum_{x>y}\tilde{m}_{xy}(t)\dd{t}\notag\\
&=\tau^{-1}\lambda\int_0^\tau\sum_{x>y}|j_{xy}(t)|\dd{t}\notag\\
&=\bar{D},
\end{align}
which fulfills condition (ii).
Finally, we prove that condition (iii) is also satisfied.
To this end, we first note that
\begin{align}
\tilde{\sigma}_t&=\sum_{x>y}\tilde{j}_{xy}(t)\ln\frac{\tilde{w}_{xy}(t)\tilde{p}_y(t)}{\tilde{w}_{yx}(t)\tilde{p}_x(t)}\notag\\
&=\lambda^{-1}\sum_{x>y}|j_{xy}(t)|.
\end{align}
Consequently, condition (iii) can be verified as follows:
\begin{align}
\tilde{\Sigma}_\tau\tilde{\mca{M}}_\tau&=\int_0^\tau \tilde{\sigma}_t\dd{t}\int_0^\tau\tilde{m}_t\dd{t}\notag\\
&=\qty( \int_0^\tau\sum_{x>y} |j_{xy}(t)| \dd{t} )^2\notag\\
&\le\qty( \int_0^\tau\sum_{x>y} \sigma_{xy}(t) \dd{t} ) \qty( \int_0^\tau\sum_{x>y}  m_{xy}(t) \dd{t} )\notag\\
&=\Sigma_\tau\mca{M}_\tau.
\end{align}
It is noteworthy that if we choose $\bar{D}=\ev{m}_\tau$, condition (iii) implies $\tilde{\Sigma}_\tau\le\Sigma_\tau$.

\end{proof}

\subsection{Lower bound of dynamical state mobility}
Here, we provide a lower bound of dynamical state mobility in terms of entropy production and dynamical activity.
By performing algebraic calculations, we can show that the kinetic coefficients $\{m_{xy}(t)\}$ can be expressed in terms of entropy production and dynamical activity rates at the transition level as
\begin{equation}
m_{xy}(t)=\frac{\sigma_{xy}(t)}{4}\Phi\qty(\frac{\sigma_{xy}(t)}{2[a_{xy}(t)+a_{yx}(t)]})^{-2},
\end{equation}
where $\Phi(x)$ is the inverse function of $x\tanh(x)$.
Since $x\Phi(x/y)^{-2}$ is a convex function over $(0,+\infty)\times(0,+\infty)$, we can derive a lower bound for $m_t$ as
\begin{equation}\label{eq:lt.lb}
\frac{\sigma_t}{4}\Phi\qty(\frac{\sigma_t}{2a_t})^{-2}\le m_t.
\end{equation}
The inequality \eqref{eq:lt.lb} indicates that $m_t$ can be lower bounded by both the entropy production and dynamical activity rates.
Exploiting the convexity of $x\Phi(x/y)^{-2}$ also yields the following inequality:
\begin{equation}\label{eq:lt.mag.rel}
\frac{\Sigma_\tau}{4}\Phi\qty(\frac{\Sigma_\tau}{2\mca{A}_\tau})^{-2}\le\mca{M}_\tau.
\end{equation}

\subsection{Proof of Eq.~\eqref{eq:new.TUR}}\label{app:proof.TUR}
Through the Cram{\'e}r-Rao inequality \cite{Hasegawa.2019.PRE}, the precision of time-antisymmetric currents can be upper bounded by pseudo entropy production as \cite{Shiraishi.2021.JSP}
\begin{equation}\label{eq:raw.TUR}
\frac{\ev{J}^2}{\Var{J}}\le \Sigma_\tau^{\rm ps},
\end{equation}
where $\Sigma_\tau^{\rm ps}$ denotes pseudo entropy production given by
\begin{equation}
\Sigma_\tau^{\rm ps}=\int_0^\tau\sum_{x>y}\frac{j_{xy}(t)^2}{a_{xy}(t)+a_{yx}(t)}\dd{t}.
\end{equation}
$\Sigma_\tau^{\rm ps}$ is an empirical quantity that quantifies the degree of irreversibility.
Unlike irreversible entropy production $\Sigma_\tau$, which diverges in the presence of unidirectional transitions, pseudo entropy production always remains finite.
However, it cannot be related directly to heat dissipation in thermodynamic processes.
It is noteworthy that the magnitude relation $\Sigma_\tau^{\rm ps}\le\Sigma_\tau/2$ holds for all times.

Next, we prove the following inequality:
\begin{equation}\label{eq:add.ine6}
\Sigma_\tau^{\rm ps}\le \frac{\Sigma_\tau\mca{M}_\tau}{\mca{A}_\tau}.
\end{equation}
Noting that $j_{xy}(t)=a_{xy}(t)-a_{yx}(t)$ and $f_{xy}(t)=\ln[a_{xy}(t)/a_{yx}(t)]$, we can show that
\begin{align}
a_{xy}(t)+a_{yx}(t)&=j_{xy}(t)\frac{e^{f_{xy}(t)}+1}{e^{f_{xy}(t)}-1}\notag\\
&=j_{xy}(t)\coth[f_{xy}(t)/2].
\end{align}
Applying Prop.~\ref{prop:add.ine0}, we can prove the inequality \eqref{eq:add.ine6} as follows:
\begin{align}
\Sigma_\tau^{\rm ps}\mca{A}_\tau &=\int_0^\tau\sum_{x>y}\frac{j_{xy}(t)^2}{a_{xy}(t)+a_{yx}(t)}\dd{t}\notag\\
&\times\int_0^\tau\sum_{x>y}[a_{xy}(t)+a_{yx}(t)]\dd{t}\notag\\
&=\int_0^\tau\sum_{x>y}\frac{j_{xy}(t)}{\coth[f_{xy}(t)/2]}\dd{t}\notag\\
&\times\int_0^\tau\sum_{x>y}j_{xy}(t)\coth[f_{xy}(t)/2]\dd{t}\notag\\
&\le \int_0^\tau\sum_{x>y}\frac{j_{xy}(t)}{f_{xy}(t)}\dd{t}\int_0^\tau\sum_{x>y}j_{xy}(t)f_{xy}(t)\dd{t}\notag\\
&=\Sigma_\tau\mca{M}_\tau.
\end{align}
Combining Eqs.~\eqref{eq:raw.TUR} and \eqref{eq:add.ine6} gives the following thermodynamic uncertainty relation:
\begin{equation}
\frac{\ev{J}^2}{\Var{J}}\le\frac{\Sigma_\tau\mca{M}_\tau}{\mca{A}_\tau}=\eta\frac{\Sigma_\tau}{2}.
\end{equation}

\subsection{Additional illustration of Eq.~\eqref{eq:new.TUR}}
Here, we numerically demonstrate the improved thermodynamic uncertainty relation in a thermoelectric device \cite{Rutten.2009.PRB}.
A thermoelectric device is an engine that transports electrons from a low- to a high-potential lead through a two-level quantum dot [see Fig.~\ref{fig:TUREx2}(a)].
Each energy level $\varepsilon_i$ of the quantum dot is coupled to a lead with chemical potential $\mu_i~(\mu_2>\mu_1)$ and temperature $T_c$.
Electrons enter and exit the quantum dot due to interactions with the leads.
Because of the Coulomb repulsion between electrons, we can assume that at most one electron always exists in the quantum dot.
The transitions between the two levels of the quantum dot are mediated by two heat baths, namely, cold and hot baths at temperatures $T_c$ and $T_h~(>T_c)$, respectively.
From a thermodynamic perspective, the device can be considered a heat engine that converts some of the heat absorbed from the hot heat bath into work in the form of transporting electrons from a low to a high potential.
\begin{figure}[!]
\centering
\includegraphics[width=1.0\linewidth]{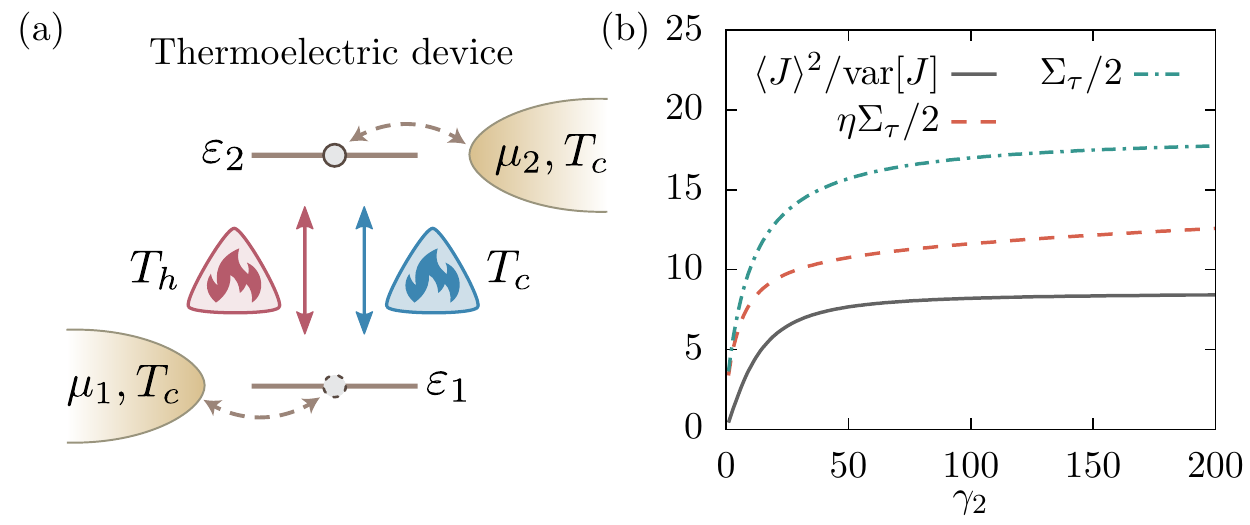}
\protect\caption{Numerical illustration of the thermodynamic uncertainty relations. (a) Schematic of the thermoelectric engine that transports electrons from the left to the right lead through the two-level quantum dot, and (b) numerical verification. The current precision $\ev{J}^2/\Var{J}$, new bound $\eta\Sigma_\tau/2$, and conventional bound $\Sigma_\tau/2$ are depicted by the solid, dashed, and dash-dotted lines, respectively. $\gamma_2$ is varied, whereas other parameters are fixed as $\beta_c=10$, $\beta_h=0.1$, $\gamma_1=10$, $\gamma_c=\gamma_h=1$, $\varepsilon_1=0$, $\varepsilon_2=1$, $\mu_1=0.4$, $\mu_2=0.6$, and $\tau=1$.}\label{fig:TUREx2}
\end{figure}

The thermoelectric device can be described by a Markov jump process with three states; that is, the quantum dot is either 1) an empty (state $0$), 2) contains one electron in energy level $\varepsilon_1$ (state $1$), or 3) contains one electron in energy level $\varepsilon_2~(>\varepsilon_1)$ (state $2$).
Electrons are exchanged with the leads at the following rates:
\begin{align}
w_{i0} = \gamma_i/(1+e^{x_i}),~ w_{0i} = \gamma_ie^{x_i}/(1+e^{x_i}),
\end{align}
where $\gamma_i>0$ denotes the coupling strength to lead $i$ and $x_i:=(\varepsilon_i-\mu_i)/T_c$.
The transition rates between the two energy levels of the quantum dot are given by
\begin{align}
w_{12}=w_{12}^c+w_{12}^h,~w_{21}=w_{21}^c+w_{21}^h,
\end{align}
where $w_{21}^a=\gamma_a/(e^{x_a}-1)$, $w_{12}^a=\gamma_ae^{x_a}/(e^{x_a}-1)$, $x_a:=(\varepsilon_2-\varepsilon_1)/T_a$ for $a\in\{c,h\}$, and $\gamma_c$ and $\gamma_h$ denote the coupling strengths to the heat baths.
Here, the symbols $c$ and $h$ correspond to the cold and hot heat baths, respectively.

We consider the thermoelectric device operating in a stationary state.
The current of interest is the net number of electrons transported between the leads.
The stochastic current can be defined by setting $\Upsilon_{10}=1=-\Upsilon_{01}$ and $\Upsilon_{xy}=0$ for others.
The precision of the current over a finite period $\tau$ can be numerically calculated using full counting statistics.

We vary $\gamma_2\in(0,200]$ while fixing the remaining parameters.
For each parameter setting, we calculate the precision of the electron current and the bounds of the conventional and new relations.
As Fig.~\ref{fig:TUREx2}(b) shows, the new bound is always tighter than the conventional bound and more effectively predicts the current precision.

\section{Derivation of calculations in Sec.~\ref{sec:opt.tran.dis.sys}}

\subsection{Proof of Eq.~\eqref{eq:Wc.tot.var.dist}}\label{app:Wass.tot.var.equiv}
Here, we prove that $\mca{W}_1(p,q)=\mca{T}(p,q)$ in the case of $d_{xy}=1-\delta_{xy}$.
First, we prove that $\mca{W}_1(p,q)\ge \mca{T}(p,q)$.
Let $S_+=\qty{x\,|\,p_x\ge q_x}$ and $S_-=\qty{x\,|\,p_x< q_x}$. Evidently, $S_+\cup S_-=\qty{1,2,\dots,N}$.
Moreover, since $\sum_xp_x=\sum_xq_x=1$, we have
\begin{equation}
\sum_{x\in S_+}(p_x-q_x)=\sum_{x\in S_-}(q_x-p_x).
\end{equation}
Consequently, $\sum_x|p_x-q_x|=2\sum_{x\in S_-}(q_x-p_x)$.
Exploiting the positivity of $d_{xy}$ and $\pi_{xy}$, we can bound $\mca{W}_1$ from below as follows:
\begin{align}
\mca{W}_1(p,q)&=\min_{\pi\in\Pi(p,q)}\sum_{x,y}d_{xy}\pi_{xy}\notag\\
&\ge\min_{\pi\in\Pi(p,q)}\sum_{x\in S_-}\sum_{y}d_{xy}\pi_{xy}\notag\\
&\ge\min_{\pi\in\Pi(p,q)}\sum_{x\in S_-}\sum_{y}d_{xy}(\pi_{xy}-\pi_{yx})\notag\\
&=\min_{\pi\in\Pi(p,q)}\sum_{x\in S_-}\sum_{y}(\pi_{xy}-\pi_{yx})\notag\\
&=\min_{\pi\in\Pi(p,q)}\sum_{x\in S_-}(q_x-p_x)\notag\\
&=\frac{1}{2}\sum_x|p_x-q_x|\notag\\
&=\mca{T}(p,q).\label{eq:tv.w1.1}
\end{align}
Next, we show that this inequality can be attained with a specific coupling.
Since $\sum_{x\in S_+}(p_x-q_x)=\sum_{x\in S_-}(q_x-p_x)$, according to Prop.~\ref{prop:pre.mat}, nonnegative coefficients $\{z_{xy}\}$ defined over $S_-\times S_+$ always exist such that
\begin{align}
\sum_{y\in S_+}z_{xy}&=q_x-p_x,~\forall x\in S_-,\\
\sum_{y\in S_-}z_{yx}&=p_x-q_x,~\forall x\in S_+.
\end{align}
We now construct a coupling $\pi=[\pi_{xy}]$ as follows:
\begin{align}
\pi_{xx}&=p_x,~\forall x\in S_-,\\
\pi_{xx}&=q_x,~\forall x\in S_+,\\
\pi_{xy}&=0,~\forall x\in S_+~\text{and}~y\neq x,\\
\pi_{xy}&=0,~\forall y\in S_-~\text{and}~x\neq y,\\
\pi_{xy}&=z_{xy},~\text{otherwise}.
\end{align}
We can verify that $\pi\in\Pi(p,q)$ and $\sum_{x,y}d_{xy}\pi_{xy}=\mca{T}(p,q)$.
From the definition of the Wasserstein distance, we have
\begin{equation}\label{eq:tv.w1.2}
\mca{W}_1(p,q)\le \sum_{x,y}d_{xy}\pi_{xy}=\mca{T}(p,q).
\end{equation}
Combining Eqs.~\eqref{eq:tv.w1.1} and \eqref{eq:tv.w1.2} yields $\mca{W}_1(p,q)=\mca{T}(p,q)$.

\subsection{Proof of Thm.~\ref{thm:cla.dis.Wass.var}}\label{app:proof.thm1}
Here, we prove Thm.~\ref{thm:cla.dis.Wass.var}, which can be restated as
\begin{equation}
\mca{W}_1(p^A,p^B)=\min_{\msf{W}_t}{\int_0^\tau\sqrt{\sigma_t m_t}\dd{t}}=\min_{\msf{W}_t}\sqrt{\Sigma_\tau\mca{M}_\tau}.
\end{equation}
To this end, we prove that $\text{RHS}\ge\text{LHS}$ and $\text{RHS}\le\text{LHS}$. First, we prove the former.
According to the Cauchy--Schwarz inequality, we have
\begin{align}
\sqrt{\Sigma_\tau\mca{M}_\tau}&\ge\int_0^\tau\sqrt{\sigma_t m_t}\dd{t},\label{eq:thm1.tmp1} \\
\sqrt{\sigma_t m_t}&=\qty(\sum_{x>y} m_{xy}(t)f_{xy}(t)^2\sum_{x>y} m_{xy}(t))^{1/2}\notag \\
&\ge \sum_{x>y}\sqrt{ m_{xy}(t)f_{xy}(t)^2}\sqrt{ m_{xy}(t)}\notag \\
&=\sum_{x>y}|j_{xy}(t)|.\label{eq:thm1.tmp2}
\end{align}
We then need only prove that
\begin{equation}\label{eq:app.thm1.tmp3}
\int_0^\tau\sum_{x>y}|j_{xy}(t)|\dd{t}\ge\mca{W}_1(p^A,p^B).
\end{equation}
For this purpose, we map the optimal transport problem to a minimum cost flow problem.
Let $\mca{G}(V,E)$ be the topology of Markov jump processes, from which the Wasserstein distance is defined.
We consider a directed graph of $N+2$ vertices: source vertex, target vertex, and $N$ intermediate vertices $\{1,\dots,N\}$ (see Fig.~\ref{fig:WassFlow} for illustration).
Each edge $e$ of the graph is associated with a cost $c(e)\ge 0$ and capacity $a(e)>0$ (i.e., the maximum flow that can be sent along this edge).
The cost of sending a flow $f$ along an edge $e$ is thus $f\times c(e)$.
The set of directed edges is as
\begin{align}
\text{source}\to x&: (c=0,~a=p_x^A),\\
x\to\text{target}&: (c=0,~a=p_x^B),\\
x\leftrightarrow y&: (c=1,~a=+\infty)~\text{if}~(x,y)\in E.
\end{align}
Consider a case in which an amount of flow $1$ is sent from the source vertex to the target vertex.
We can then prove that the minimum cost $\mca{C}$ of this flow problem is exactly the discrete Wasserstein distance.
To this end, we first show that $\mca{C}\ge\mca{W}_1(p^A,p^B)$.
Assume that $\mca{C}$ is attained by effectively sending a flow $\pi_{xy}$ from $\text{source}\to y\to x\to \text{target}$ for each $x$ and $y$. Since the shortest-path distance from $y$ to $x$ is $d_{xy}$, the total cost must be greater than or equal to $\sum_{x,y}d_{xy}\pi_{xy}$.
Notice that $\{\pi_{xy}\}$ is a valid coupling. Therefore, we obtain $\mca{C}\ge\mca{W}_1(p^A,p^B)$ from the definition of the Wasserstein distance.
We now need only prove the reverse statement $\mca{C}\le\mca{W}_1(p^A,p^B)$.
Assume that $\mca{W}_1(p^A,p^B)$ is achieved by an optimal transport plan $\pi^*=[\pi_{xy}^*]$ [i.e, for any pair $(x,y)$, we move a probability $\pi_{xy}^*$ from state $y$ to state $x$ with the cost of $d_{xy}$ per unit probability].
For each $x$ and $y$, let $P=[v_1,\dots,v_k]$ be the shortest path of length $d_{xy}$ that connects $y$ to $x$; that is, $y=v_1$, $x=v_k$, $k-1=d_{xy}$, and $(v_i,v_{i+1})\in E$ for all $1\le i<k$.
We can then send an amount of flow $\pi_{xy}^*$ along the path $(\text{source}\to v_1\to \dots \to v_k\to\text{target})$.
The total flow cost is exactly $\sum_{x,y}\pi_{xy}^*d_{xy}=\mca{W}_1(p^A,p^B)$; thus, $\mca{C}\le\mca{W}_1(p^A,p^B)$.
Consequently, we arrive at the equality $\mca{C}=\mca{W}_1(p^A,p^B)$.
\begin{figure}[t]
\centering
\includegraphics[width=1\linewidth]{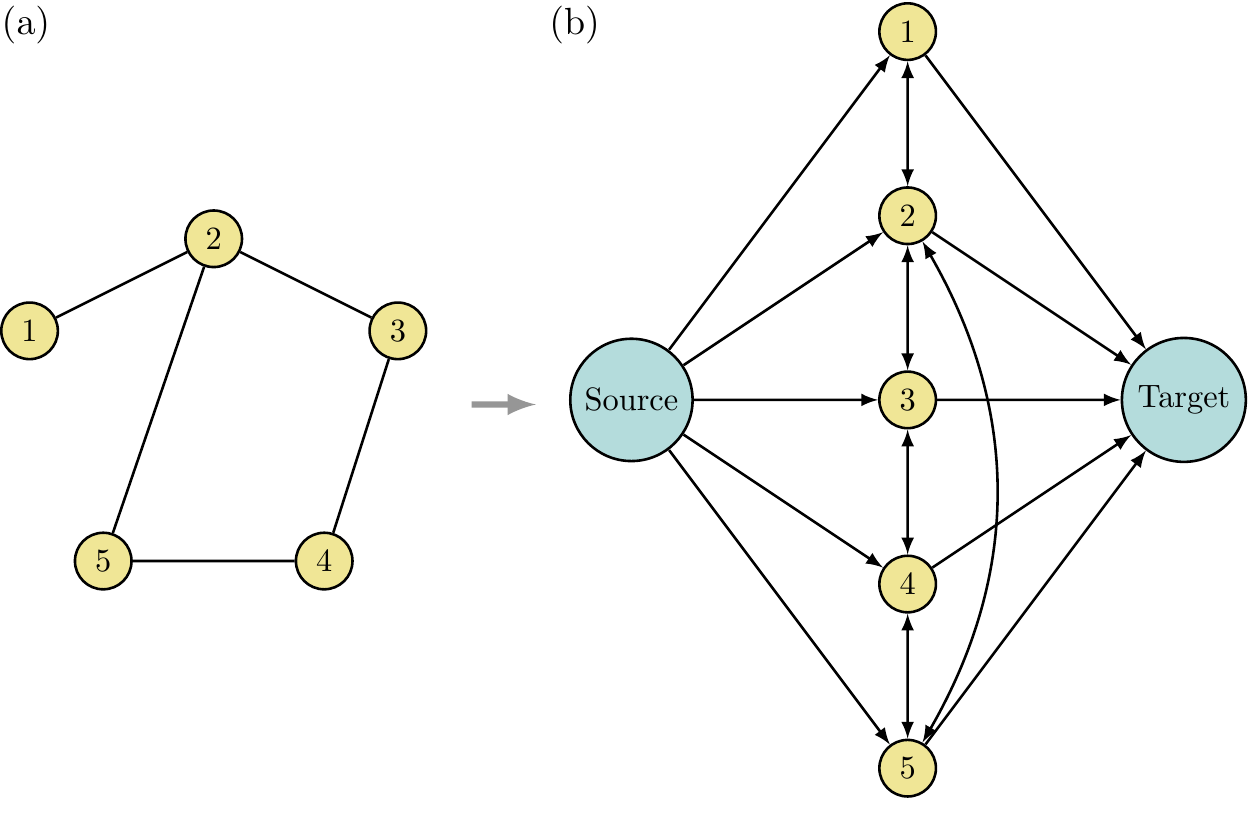}
\protect\caption{Mapping the Wasserstein distance defined based on topology to a minimum cost flow problem. (a) Topology $\mca{G}(V,E)$ with $5$ vertices and $5$ edges, from which the Wasserstein distance is defined. (b) Directed graph of the minimum cost flow problem, which is constructed using the given topology $\mca{G}$. The minimum cost of the flow problem is equal to the Wasserstein distance.}\label{fig:WassFlow}
\end{figure}

We next show that the Markov jump process gives an admissible solution of the minimum cost flow problem with the cost $\int_0^\tau\sum_{x>y}|j_{xy}(t)|\dd{t}$. Consider discretization of the master equation with the time interval $\delta t$, where $\tau=K\delta t$. For each $k=0,\dots,K-1$, we have
\begin{equation}
p_x((k+1)\delta t)=p_x(k\delta t)+\sum_{y(\neq x)}j_{xy}(k\delta t)\delta t.
\end{equation}
This means that we send an amount of flow $|j_{xy}(k\delta t)|\delta t$ from $y$ to $x$ if $j_{xy}(k\delta t)\ge 0$ and from $x$ to $y$ if $j_{xy}(k\delta t)<0$.
Since $j_{xy}(k\delta t)\neq 0$ only if $x$ and $y$ are directly connected by an edge, the cost of each transport is $|j_{xy}(k\delta t)|\delta t$.
Therefore, the total flow cost associated with the Markov jump process is
\begin{equation}
\sum_{k=0}^{K-1}\sum_{x>y}|j_{xy}(k\delta t)|\delta t\xrightarrow{\delta t\to 0} \int_0^\tau\sum_{x>y}|j_{xy}(t)|\dd{t}.
\end{equation}
Since the Markov jump process realizes an admissible manner of sending flow from $p^A$ to $p^B$, we obtain
\begin{equation}\label{eq:thm1.tmp3}
\int_0^\tau\sum_{x>y}|j_{xy}(t)|\dd{t}\ge\mca{C}=\mca{W}_1(p^A,p^B),
\end{equation}
which verifies Eq.~\eqref{eq:app.thm1.tmp3}.
Consequently, combining Eqs.~\eqref{eq:thm1.tmp1}, \eqref{eq:thm1.tmp2}, and \eqref{eq:app.thm1.tmp3} yields
\begin{align}
\text{RHS}&\ge\min_{\msf{W}_t}\qty{\int_0^\tau\sum_{x>y}|j_{xy}(t)|\dd{t}}\notag\\
&\ge\mca{W}_1(p^A,p^B)={\rm LHS}.
\end{align}

We next prove that ${\rm LHS}\ge {\rm RHS}$ by showing that the optimal cost $\mca{W}_1(p^A,p^B)$ can be achieved with a specific Markov jump process, the underlying graph of which is a subgraph of $\mca{G}(V,E)$ for all times.
Note that the optimal transport plan can be represented as a sequence of transportation between neighboring states.
Let $[(x_1,y_1,\chi_1),\dots,(x_I,y_I,\chi_I)]$ denote the optimal transport plan; that is, at each step $1\le i\le I$, we move a probability $\chi_i$ from state $x_i$ to $y_i~(\neq x_i)$. It is ensured that at each step the probability of state $x_i$ is always greater than or equal to $\chi_i$.
Since $x_i$ and $y_i$ are neighboring states, the total transport cost is $\sum_{i=1}^I\chi_i$. Thus,
\begin{equation}
\sum_{i=1}^I\chi_i=\mca{W}_1(p^A,p^B).
\end{equation}
We now construct a Markov jump process of time period $\tau$ such that for each $1\le i\le I$, a probability $\chi_i$ is moved from $x_i$ to $y_i$ after time $t=i\Delta$, where $\Delta\coloneqq\tau/I$.
Specifically, we construct transition rates such that the probability distribution evolves as follows:
\begin{align}
p_{x_i}((i-1)\Delta+s)&=p_{x_i}((i-1)\Delta)-\frac{s}{\Delta}\chi_i,\\
p_{y_i}((i-1)\Delta+s)&=p_{y_i}((i-1)\Delta)+\frac{s}{\Delta}\chi_i,\\
p_{x}((i-1)\Delta+s)&=p_{x}((i-1)\Delta),~\forall x\neq x_i,y_i.
\end{align}
Here, $0\le s\le \Delta$ is a time parameter.
This time evolution of the probability distribution is effectively a two-level system, which can be realized using the following transition rates:
\begin{align}
w_{y_ix_i}(t)&=\frac{1}{\Delta(1-e^{-\phi})}\frac{\chi_i}{ p_{x_i}(t)},\\
w_{x_iy_i}(t)&=\frac{e^{-\phi}}{\Delta(1-e^{-\phi})}\frac{\chi_i}{ p_{y_i}(t)},\\
w_{xy}(t)&=0,~\text{otherwise}.
\end{align}
Here, $\phi>0$ is an arbitrary constant.
During the time interval $[(i-1)\Delta,i\Delta]$, the underlying graph of this process has only one edge that connects vertices $x_i$ and $y_i$. Thus, it is always a subgraph of $\mca{G}$.
Using these transition rates, we can verify that
\begin{equation}
\dot{p}_x(t)=\sum_{y(\neq x)}j_{xy}(t),~\forall x.
\end{equation}
Moreover,
\begin{align}
\sigma_t&=\sum_{x>y}j_{xy}(t)\ln\frac{a_{xy}(t)}{a_{yx}(t)}={\phi}\sum_{x>y}|j_{xy}(t)|,\\
 m_t&=\sum_{x>y}\frac{j_{xy}(t)}{\ln[a_{xy}(t)/a_{yx}(t)]}=\frac{1}{\phi}\sum_{x>y}|j_{xy}(t)|.
\end{align}
In addition, note that 
\begin{equation}\label{eq:jint.i}
\sum_{x>y}|j_{xy}(t)|=\frac{\chi_i}{\Delta} \Rightarrow \int_{(i-1)\Delta}^{i\Delta}\sum_{x>y}|j_{xy}(t)|\dd{t}=\chi_i.
\end{equation}
By summing both sides of Eq.~\eqref{eq:jint.i} for all $i=1,\dots,I$, we obtain
\begin{equation}
\int_{0}^{\tau}\sum_{x>y}|j_{xy}(t)|\dd{t}=\sum_{i=1}^I\chi_i=\mca{W}_1(p^A,p^B).
\end{equation}
Consequently, we have
\begin{align}
\Sigma_\tau\mca{M}_\tau&=\qty(\int_0^\tau\sum_{x>y}|j_{xy}(t)|\dd{t})^2\notag\\
&=\mca{W}_1(p^A,p^B)^2,
\end{align}
which completes the proof.

\subsection{Equality in Thm.~\ref{thm:cla.dis.Wass.var} can be achieved with global detailed-balance systems}\label{app:min.gdbc}
Based on the previous construction of the dynamics that attains the equality in Thm.~\ref{thm:cla.dis.Wass.var}, we can further prove that the equality can be attained with global detailed-balance dynamics.
Here, we prove this fact using a different approach.

Minimizing the integral term in Thm.~\ref{thm:cla.dis.Wass.var} is equivalent to minimizing the cost function $\sigma_tm_t$ at each instance of time $t$.
Consider the Lagrangian function
\begin{equation}
L(\msf{W}_t,\lambda)=\sigma_t m_t+{\lambda^\top}({\dot p_t}-\msf{W}_t{p_t}),
\end{equation}
where $\{p_t\}_{0\le t\le\tau}$ is the probability distribution of a dynamics that attains the equality of Thm.~\ref{thm:cla.dis.Wass.var}.
For simplicity, the time notation $t$ is omitted hereafter.
Taking the derivative of $L$ with respect to $w_{xy}$, we have
\begin{align}
\frac{\partial L}{\partial w_{xy}}&=p_y\qty(f_{xy}+1-e^{-f_{xy}}) m+\frac{p_y\qty(f_{xy}-1+e^{-f_{xy}})}{f_{xy}^2}\sigma\notag\\
&+p_y(\lambda_y-\lambda_x)=0.
\end{align}
Recall that $f_{xy}=\ln(a_{xy}/a_{yx})$.
If $p_y=0$, then $w_{xy}$ can be arbitrarily determined. Therefore, we need only consider the nontrivial case $p_y\neq 0$.
This leads to
\begin{equation}\label{eq:par.deri.1}
\qty(f_{xy}+1-e^{-f_{xy}}) m+\frac{\qty(f_{xy}-1+e^{-f_{xy}})}{f_{xy}^2}\sigma+\lambda_y-\lambda_x=0.
\end{equation}
Likewise, taking the derivative of $L$ with respect to $w_{yx}$ yields
\begin{equation}\label{eq:par.deri.2}
\qty(f_{yx}+1-e^{-f_{yx}}) m+\frac{\qty(f_{yx}-1+e^{-f_{yx}})}{f_{yx}^2}\sigma+\lambda_x-\lambda_y=0.
\end{equation}
Notice that $f_{xy}=-f_{yx}$.
Adding Eqs.~\eqref{eq:par.deri.1} and \eqref{eq:par.deri.2} side by side, we obtain
\begin{equation}
\qty(e^{f_{xy}}+e^{-f_{xy}}-2)\qty(\frac{\sigma}{f_{xy}^2}- m)=0,
\end{equation}
which gives the solution $f_{xy}=0$ or $f_{xy}^2=\sigma/ m$.
Note that $f_{xy}=0$ is equivalent to $j_{xy}=0$, which implies that the transition between $x$ and $y$ does not contribute to the time evolution of the probability distribution. Therefore, such transitions need not be considered and can be eliminated by simply setting $w_{xy}=w_{yx}=0$.
Otherwise, if $f_{xy}^2=\sigma/ m$, then Eq.~\eqref{eq:par.deri.1} becomes
\begin{equation}
2f_{xy} m+\lambda_y-\lambda_x=0,
\end{equation}
or equivalently,
\begin{equation}
\ln\frac{w_{xy}}{w_{yx}}=\frac{\lambda_x-\lambda_y}{2 m}+\ln p_x-\ln p_y.
\end{equation}
By defining an instantaneous energy $\beta\varepsilon_x\coloneqq-\lambda_x/(2 m)-\ln p_x$, we can verify that the transition rates satisfy the global detailed balance condition:
\begin{equation}
\ln\frac{w_{xy}}{w_{yx}}=\beta(\varepsilon_y-\varepsilon_x).
\end{equation}

\subsection{Minimum entropy production can be achieved with global detailed-balance systems}\label{app:min.ent.prod.cons.force}
Here we show that given the time-averaged state mobility (i.e., $\ev{m}_\tau=\bar{D}$), there always exists a system that satisfies the global detailed balance and achieves the minimum entropy production:
\begin{equation}\label{eq:min.ent.prod.Dbar}
	\min_{\ev{m}_\tau=\bar{D}}\Sigma_\tau=\frac{\mca{W}_1(p_0,p_\tau)^2}{\bar{D}\tau}.
\end{equation}
According to the equality of Thm.~\ref{thm:cla.dis.Wass.var} and Lem.~\ref{lem:mobi.ult}, there exists a dynamics that satisfies Eq.~\eqref{eq:min.ent.prod.Dbar} with time-dependent probability distributions $\{p_t\}_{0\le t\le\tau}$.
We consider the following minimization problem:
\begin{equation}\label{eq:min.ent.prob}
	\min_{}\int_0^\tau\sigma_t\dd{t},
\end{equation}
given that $\dot{p}_t=\msf{W}_tp_t$ and $\ev{m}_\tau=\bar{D}$.
Notice that the minimum value for this problem is exactly $\mca{W}_1(p_0,p_\tau)^2/(\bar{D}\tau)$.
Consider the Lagrangian function 
\begin{align}
L(\msf{W}_t,\lambda_t,\kappa)&=\int_0^\tau\sigma_t\dd{t}+\int_0^\tau{\lambda_t^\top}({\dot p_t}-\msf{W}_t{p_t})\dd{t}\notag\\
&+\kappa\qty(\int_0^\tau m_t\dd{t}-\bar{D}\tau).
\end{align}
For simplicity, the time notation $t$ is omitted hereafter.
Taking the derivative of $L$ with respect to $w_{xy}$, we have
\begin{align}
\frac{\partial L}{\partial w_{xy}}&=p_y\qty(f_{xy}+1-e^{-f_{xy}})+\frac{p_y\qty(f_{xy}-1+e^{-f_{xy}})}{f_{xy}^2}\kappa\notag\\
&+p_y(\lambda_y-\lambda_x)=0.
\end{align}
Following the same procedure as in Sec.~\ref{app:min.gdbc}, we obtain the following relation:
\begin{equation}
	2f_{xy}+\lambda_y-\lambda_x=0,
\end{equation}
or equivalently,
\begin{equation}
\ln\frac{w_{xy}}{w_{yx}}=\frac{\lambda_x-\lambda_y}{2}+\ln p_x-\ln p_y.
\end{equation}
By defining an instantaneous energy $\beta\varepsilon_x\coloneqq-\lambda_x/2-\ln p_x$, the transition rates satisfy the global detailed balance condition:
\begin{equation}
\ln\frac{w_{xy}}{w_{yx}}=\beta(\varepsilon_y-\varepsilon_x).
\end{equation}
This means that the minimum entropy production \eqref{eq:min.ent.prob} can be achieved with conservative forces.

\subsection{Particular topologies}\label{app:com.topo}
\subsubsection{Ring topology}
Here, we consider a ring topology in which vertices $x$ and $x+1$ are connected for all $x$, where $N+1\equiv 1$.
This topology can be seen in a one-dimensional asymmetric simple exclusion process on a ring of $N$ sites and corresponds to a continuous-variable situation in which a single particle is driven in a periodic potential.
For each integer number $x$, given $x=kN+r$, where $0\le r\le N-1$ is the remainder, we define $[x]_N\coloneqq r$. 
Then, the shortest-path distance between states $x$ and $y$ can be calculated as
\begin{equation}
d_{xy}=\min\qty{[x-y]_N,N-[x-y]_N}.
\end{equation}
In this case, the discrete Wasserstein distance can be written as
\begin{equation}\label{eq:Wc.ring.var.form}
\mca{W}_1( p^A,p^B ) = \min_{\msf{W}_t}{\int_0^\tau\sum_{x=1}^{N}|j_{x+1,x}(t)|\dd{t}}.
\end{equation}
We now consider the continuous case in which the particle is driven in a ring with a diameter $L=N\Delta x$.
Taking the continuous limit of Eq.~\eqref{eq:Wc.ring.var.form}, namely, $N\to\infty$ and $\Delta x\to 0$, we obtain the following relation:
\begin{align}
&\min_{\pi}{\iint\min\qty{ |x-y|, L-|x-y| }\pi(x,y)\dd{x}\dd{y}}\notag\\
&=\min_{j_t}{\int_0^\tau \int_0^L|j_t(x)|\dd{x}\dd{t}},\label{eq:Wc.ring.var.form.1}
\end{align}
where $j_t(x)$ is subject to the continuity equation $\dot{p}_t(x)=-\partial_xj_t(x)$.
The term on the left-hand side of Eq.~\eqref{eq:Wc.ring.var.form.1} is exactly the $L^1$-Wasserstein distance between probability distributions defined periodically over $[0,L]$ with the cost function
\begin{equation}
c(x,y)=\min\qty{ |x-y|,L-|x-y| }.
\end{equation}
Equation \eqref{eq:Wc.ring.var.form.1} thus provides a variational formula for the periodic $L^1$-Wasserstein distance.

\subsubsection{Fully connected topology}
Another topology is the fully connected topology; that is, for an arbitrary pair of two vertices, an edge always exists that connects them.
In this case, the shortest-path distances become
\begin{equation}
d_{xy}=1-\delta_{xy},
\end{equation}
and the discrete Wasserstein distance equals the total variation distance.
Theorem \ref{thm:cla.dis.Wass.var} thus implies the following equality:
\begin{equation}\label{eq:tot.var.dist.var.form}
\mca{T}(p^A,p^B)=\min_{\msf{W}_t}{\int_0^\tau\sqrt{\sigma_t m_t}\dd{t}}=\min_{\msf{W}_t}\sqrt{\Sigma_\tau\mca{M}_\tau}.
\end{equation}
Here, the minimum is taken over all possible transition rate matrices; that is, the transition rate between any two states can be arbitrarily controlled.
Although the total variation distance is widely used in previous studies, its connection with thermodynamics has thus far been veiled.
Equation \eqref{eq:tot.var.dist.var.form} reveals a thermodynamic interpretation of this distance, showing that it equals the minimum product of the thermodynamic and kinetic costs given the full control of the transition rates.

\subsection{Alternative variational expressions of the discrete Wasserstein distance}\label{app:alt.Wc.exp}
\begin{corollary}\label{cor:cla.dis.Wass.var1}
The discrete Wasserstein distance can be expressed in terms of irreversible entropy production and dynamical activity as
\begin{align}
\mca{W}_1( p^A,p^B )&=\min_{\msf{W}_t}{\int_0^\tau\frac{\sigma_t}{2}\Phi\qty(\frac{\sigma_t}{2a_t})^{-1}\dd{t}}\label{eq:Wc.var.form3}\\
&=\min_{\msf{W}_t}\frac{\Sigma_\tau}{2}\Phi\qty(\frac{\Sigma_\tau}{2\mca{A}_\tau})^{-1}.\label{eq:Wc.var.form4}
\end{align}
\end{corollary}
\begin{proof}
We first prove that 
\begin{equation}\label{eq:lem11}
\int_0^\tau\frac{\sigma_t}{2}\Phi\qty(\frac{\sigma_t}{2a_t})^{-1}\dd{t}\ge \mca{W}_1(p^A,p^B).
\end{equation}
Note that $x\Phi(x/y)^{-1}$ is a concave function over $(0,+\infty)\times(0,+\infty)$. Applying Jensen's inequality yields
\begin{align}
\frac{\sigma_t}{2}\Phi\qty(\frac{\sigma_t}{2a_t})^{-1}&\ge\sum_{x>y}\frac{\sigma_{xy}(t)}{2}\Phi\qty(\frac{\sigma_{xy}(t)}{2[a_{xy}(t)+a_{yx}(t)]})^{-1}\notag\\
&=\sum_{x>y}|j_{xy}(t)|.\label{eq:ent.llbt}
\end{align}
By taking the time integration of Eq.~\eqref{eq:ent.llbt} and using Eq.~\eqref{eq:Wc.var.form.cor}, we immediately prove Eq.~\eqref{eq:lem11}.
Moreover, using the concavity of $x\Phi(x/y)^{-1}$ yields
\begin{align}
\int_0^\tau\frac{\sigma_t}{2}\Phi\qty(\frac{\sigma_t}{2a_t})^{-1}\dd{t}&\le\frac{\Sigma_\tau}{2}\Phi\qty(\frac{\Sigma_\tau}{2\mca{A}_\tau})^{-1}\notag\\
&\le \sqrt{\Sigma_\tau\mca{M}_\tau}.
\end{align}
Thus, we have
\begin{align}
\mca{W}_1(p^A,p^B)\le \int_0^\tau\frac{\sigma_t}{2}\Phi\qty(\frac{\sigma_t}{2a_t})^{-1}\dd{t}&\le \frac{\Sigma_\tau}{2}\Phi\qty(\frac{\Sigma_\tau}{2\mca{A}_\tau})^{-1}\notag\\
&\le \sqrt{\Sigma_\tau\mca{M}_\tau}.\label{eq:ent.llbt2}
\end{align}
The proof is completed by taking the minimum of the terms on the right-hand side of Eq.~\eqref{eq:ent.llbt2} over all admissible dynamics and applying Thm.~\ref{thm:cla.dis.Wass.var}.
\end{proof}
Equation \eqref{eq:Wc.var.form4} implies that the discrete Wasserstein distance can be expressed in terms of irreversible entropy production and dynamical activity as
\begin{equation}
\mca{W}_1(p^A,p^B)=\min{\frac{\Sigma_\tau}{2}\Phi\qty(\frac{\Sigma_\tau}{2\mca{A}_\tau})^{-1}},
\end{equation}
which recovers the result obtained in Ref.~\cite{Dechant.2022.JPA}.

\begin{corollary}\label{cor:cla.dis.Wass.var2}
The discrete Wasserstein distance can be expressed in terms of pseudo entropy production and dynamical activity as
\begin{align}
\mca{W}_1( p^A,p^B )&=\min_{\msf{W}_t}{\int_0^\tau\sqrt{\sigma_t^{\rm ps}a_t}\dd{t}}\label{eq:Wc.var2.form1}\\
&=\min_{\msf{W}_t}\sqrt{\Sigma_\tau^{\rm ps}\mca{A}_\tau},\label{eq:Wc.var2.form2}
\end{align}
where $\sigma_t^{\rm ps}\coloneqq\dot{\Sigma}_t^{\rm ps}$ denotes the pseudo entropy production rate.
\end{corollary}
\begin{proof}
The proof strategy is the same as in Cor.~\ref{cor:cla.dis.Wass.var1}. We first prove that 
\begin{equation}\label{eq:lem21}
\int_0^\tau\sqrt{\sigma_t^{\rm ps}a_t}\dd{t}\ge \mca{W}_1(p^A,p^B).
\end{equation}
Applying the Cauchy--Schwarz inequality, we obtain 
\begin{align}
\sum_{x>y}|j_{xy}(t)|&=\sum_{x>y}\frac{|j_{xy}(t)|}{\sqrt{a_{xy}(t)+a_{yx}(t)}}\sqrt{a_{xy}(t)+a_{yx}(t)}\notag\\
&\le \sqrt{\sigma_t^{\rm ps}a_t}.\label{eq:pent.dyn}
\end{align}
By taking the time integration of Eq.~\eqref{eq:pent.dyn} and using Eq.~\eqref{eq:Wc.var.form.cor}, we immediately prove Eq.~\eqref{eq:lem21}.
Since $\Sigma_\tau^{\rm ps}\mca{A}_\tau\le\Sigma_\tau\mca{M}_\tau$, the following relation holds:
\begin{equation}\label{eq:pent.dyn2}
\mca{W}_1(p^A,p^B)\le \int_0^\tau\sqrt{\sigma_t^{\rm ps}a_t}\dd{t}\le \sqrt{\Sigma_\tau^{\rm ps}\mca{A}_\tau}\le \sqrt{\Sigma_\tau\mca{M}_\tau}.
\end{equation}
Taking the minimum of the terms on the right-hand side of Eq.~\eqref{eq:pent.dyn2} over all admissible dynamics and using Thm.~\ref{thm:cla.dis.Wass.var} complete the proof.
\end{proof}
Equation \eqref{eq:Wc.var2.form2} has the following implication.
If dynamical activity $\mca{A}_\tau$ is fixed, then the minimum pseudo entropy production can be calculated using the Wasserstein distance as
\begin{equation}
\min \Sigma_\tau^{\rm ps}=\frac{\mca{W}_1(p^A,p^B)^2}{\mca{A}_\tau},
\end{equation}
which recovers the result reported in Ref.~\cite{Dechant.2022.JPA}.
From Thm.~\ref{thm:cla.dis.Wass.var} and Cor.~\ref{cor:cla.dis.Wass.var2}, we can observe that in the context of optimal transport, ($\Sigma_\tau$, $\mca{M}_\tau$) and ($\Sigma_\tau^{\rm ps}$, $\mca{A}_\tau$) are two thermodynamic-kinetic conjugate pairs.

\section{Derivation of calculations in Sec.~\ref{sec:quan.gen}}

\subsection{Proof of Eq.~\eqref{eq:qua.ent.prod.rate}}\label{app:qua.ent.prod.exp}
Here, we derive an analytical expression of the entropy production rate $\sigma_t$.
Taking the time derivative of irreversible entropy production, we can calculate the entropy production rate as
\begin{align}
\sigma_t&=-\tr{\dot\varrho_t\ln\varrho_t}+\sum_k\tr{L_k(t)^\dagger L_k(t)\varrho_t}s_k(t)\notag\\
&=\sum_k-\tr{(\mca{D}[L_k(t)]\varrho_t)\ln\varrho_t}+\tr{L_k(t)^\dagger L_k(t)\varrho_t}s_k(t)\notag\\
&=\sum_k\tr{L_k(t)\varrho_t(s_k(t)L_k(t)^\dagger-[L_k(t)^\dagger,\ln\varrho_t])}.
\end{align}
Notice that $w_k^{xy}(t)=e^{s_k(t)}w_{k'}^{yx}(t)$.
Since $\tr{A}=\sum_x\mel{x_t}{A}{x_t}$ for any operator $A$, the entropy production rate can be calculated further as
\begin{align}
\sigma_t&=\sum_k\sum_x\mel{x_t}{L_k(t)\varrho_t( s_k(t)L_k(t)^\dagger-[L_k(t)^\dagger,\ln\varrho_t])}{x_t}\notag\\
&=\sum_k\sum_{x,y}w_k^{xy}(t)p_y(t)\qty[ s_k(t)+\ln\frac{p_y(t)}{p_x(t)}]\notag\\
&=\frac{1}{2}\sum_k\sum_{x,y}j_k^{xy}(t)\ln\frac{w_k^{xy}(t)p_y(t)}{w_{k'}^{yx}(t)p_x(t)}.\label{eq:ent.prod.rate}
\end{align}
Since $(a-b)\ln(a/b)\ge 0$ for all $a,b\ge 0$, the positivity of $\sigma_t$ is immediately derived.

\subsection{Proof of Eq.~\eqref{eq:Lt.Onsager.rel}}\label{app:qua.mobi.exp}
First, quantum dynamical state mobility can be expressed in terms of eigenvalues $\{p_x(t)\}$ and transition rates $\{w^{xy}_k(t)\}$ as
\begin{align}
 m_t&=\frac{1}{2}\sum_{k}e^{-s_k(t)/2}\ev{L_k(t)^\dagger,\sop{\varrho_t}_{s_k(t)}(\mca{P}_t[L_k(t)^\dagger])}\notag\\
&=\frac{1}{2}\sum_{k}e^{-s_k(t)}\int_0^1e^{s_k(t)u}\ev{L_k(t)^\dagger,\varrho_t^{u}(\mca{P}_t[L_k(t)^\dagger])\varrho_t^{1-u}}\dd{u}\notag\\
&=\frac{1}{2}\sum_{k}e^{-s_k(t)}\sum_{x\neq y}\int_0^1e^{s_k(t)u}w_k^{xy}(t)p_y(t)^up_x(t)^{1-u}\dd{u}.\label{eq:lt.ons.tmp1}
\end{align}
Likewise, we can calculate
\begin{align}
&\sum_x\bra{x_t}\otimes\ket{x_t}^\top \msf{O}_k(t,u)\ket{x_t}\otimes\bra{x_t}^\top\notag\\
&=\sum_x\bra{x_t}\otimes\ket{x_t}^\top L_k(t)\varrho_t^uL_k(t)^\dagger\otimes(\varrho_t^{1-u})^\top\ket{x_t}\otimes\bra{x_t}^\top\notag\\
&+\sum_x\bra{x_t}\otimes\ket{x_t}^\top\varrho_t^u\otimes (L_k(t)^\dagger\varrho_t^{1-u}L_k(t))^\top\ket{x_t}\otimes\bra{x_t}^\top\notag\\
&-\sum_x\bra{x_t}\otimes\ket{x_t}^\top L_k(t)\varrho_t^u\otimes (L_k(t)^\dagger\varrho_t^{1-u})^\top\ket{x_t}\otimes\bra{x_t}^\top\notag\\
&-\sum_x\bra{x_t}\otimes\ket{x_t}^\top\varrho_t^uL_k(t)^\dagger\otimes (\varrho_t^{1-u}L_k(t))^\top\ket{x_t}\otimes\bra{x_t}^\top\notag\\
&=\sum_{x\neq y}\qty[w_k^{xy}(t)p_x(t)^{1-u}p_y(t)^{u} + w_k^{yx}(t)p_x(t)^{u}p_y(t)^{1-u}]\notag\\
&=2\sum_{x\neq y}w_k^{xy}(t)p_y(t)^{u}p_x(t)^{1-u}.\label{eq:lt.ons.tmp2}
\end{align}
Consequently, combining Eqs.~\eqref{eq:lt.ons.tmp1} and \eqref{eq:lt.ons.tmp2} yields the desired relation:
\begin{align}
&\frac{1}{2}\sum_{x}\bra{x_t}\otimes\ket{x_t}^\top\msf{O}_t\ket{x_t}\otimes\bra{x_t}^\top\notag\\
&=\frac{1}{4}\sum_{x} \sum_{k}e^{-s_k(t)}\int_0^1e^{s_k(t)u}\bra{x_t}\otimes\ket{x_t}^\top\msf{O}_k(t,u)\ket{x_t}\otimes\bra{x_t}^\top\dd{u}\notag\\
&=\frac{1}{2}\sum_ke^{-s_k(t)}\sum_{x\neq y}\int_0^1e^{s_k(t)u}w_k^{xy}(t)p_y(t)^{u}p_x(t)^{1-u}\dd{u}\notag\\
&= m_t.
\end{align}

\subsection{Proof of Eq.~\eqref{eq:ana.exp.tra.nrm}}\label{app:qua.Wass.dis.exp.proof}
First, we prove that $\mca{W}_q(\varrho^A,\varrho^B)\le \mca{T}(p^A,p^B)$.
Let $\varrho^A=\sum_xp_x^A\dyad{x^A}$ and $\varrho^B=\sum_xp_x^B\dyad{x^B}$ be the spectral decompositions of the density matrices.
Setting $V_*=\sum_x\dyad{x^B}{x^A}$, we can verify that $V_*$ is a unitary operator and $V_*\varrho^AV_*^\dagger=\sum_xp_x^A\dyad{x^B}$.
From the definition of $\mca{W}_q$, we have
\begin{align}
\mca{W}_q(\varrho^A,\varrho^B)&=\frac{1}{2}\min_{V^\dagger V=\mbb{1}}{ \|V\varrho^AV^\dagger-\varrho^B\|_1 }\notag\\
&\le\frac{1}{2} \|V_*\varrho^AV_*^\dagger-\varrho^B\|_1\notag\\
&=\frac{1}{2}\|\sum_x(p_x^A-p_x^B)\dyad{x^B}\|_1\notag\\
&=\frac{1}{2}\sum_x|p_x^A-p_x^B|=\mca{T}(p^A,p^B).
\end{align}
Note that $\{p_x^A\}$ are increasing eigenvalues of $V\varrho^A V^\dagger$ for an arbitrary unitary operator $V$.
Let $\varsigma_x(A)$ be the $x$-th singular value of operator $A$ in ascending order. Then,
\begin{equation}
\|A-B\|_1=\sum_x\varsigma_x(A-B)\ge \sum_x|\varsigma_x(A)-\varsigma_x(B)|
\end{equation}
holds for arbitrary Hermitian operators $A$ and $B$ \cite{Bhatia.1996}.
Applying the above inequality for $A=V\varrho^A V^\dagger$ and $B=\varrho^B$ yields $\|V\varrho^A V^\dagger-\varrho^B\|_1\ge \sum_x|p_x^A-p_x^B|$, from which we immediately obtain $\mca{W}_q(\varrho^A,\varrho^B)\ge \mca{T}(p^A,p^B)$.
Consequently, Eq.~\eqref{eq:ana.exp.tra.nrm} is proved.

\subsection{Proof of Thm.~\ref{thm:qua.dis.Wass.var}}\label{app:proof.qua.thm}
Here, we prove Thm.~\ref{thm:qua.dis.Wass.var}, which can be restated as
\begin{equation}
\mca{W}_q(\varrho^A,\varrho^B)=\min_{\mca{L}_t}{\int_0^\tau\sqrt{\sigma_t m_t} \dd{t}}=\min_{\mca{L}_t}\sqrt{\Sigma_\tau\mca{M}_\tau}.
\end{equation}
First, we prove that $\text{RHS}\ge \text{LHS}$.
Note that
\begin{align}
\sigma_t&=\frac{1}{2}\sum_k\sum_{x,y}j_k^{xy}(t)\ln\frac{a_k^{xy}(t)}{a_{k'}^{yx}(t)},\\
 m_t&=\frac{1}{2}\sum_k\sum_{x\neq y}\frac{j_k^{xy}(t)}{\ln[a_k^{xy}(t)/a_{k'}^{yx}(t)]}.
\end{align}
Applying the Cauchy--Schwarz inequality and triangle inequality, we obtain
\begin{align}
\sqrt{\sigma_t m_t}&\ge\frac{1}{2}\sum_k\sum_{x\neq y}|j_k^{xy}(t)|\notag\\
&\ge\frac{1}{2}\sum_{x}|\sum_k\sum_{y(\neq x)}j_k^{xy}(t)|\notag\\
&=\frac{1}{2}\sum_{x}|\dot p_x(t)|.\label{eq:proof.pnt.bound}
\end{align}
Consequently, taking the time integration and applying Prop.~\ref{prop:arra.ine} yield the following result:
\begin{align}
\min_{\mca{L}_t}\sqrt{\Sigma_\tau\mca{M}_\tau}&\ge\min_{\mca{L}_t}{\int_0^\tau\sqrt{\sigma_t m_t}\dd{t}}\notag\\
&\ge\frac{1}{2}\min_{\mca{L}_t}{\sum_{x}\int_0^\tau|\dot p_x(t)|\dd{t}}\notag\\
&\ge\frac{1}{2}\min_{\mca{L}_t}{\sum_{x}\big|\int_0^\tau\dot p_x(t)\dd{t}\big|}\notag\\
&=\frac{1}{2}\min_{\mca{L}_t}{\sum_{x}|p_x(\tau)-p_x(0)|}\notag\\
&\ge \mca{W}_q(\varrho^A,\varrho^B).\label{eq:Wq.bound.tmp1}
\end{align}
Next, we need only show that the equality in the inequality \eqref{eq:Wq.bound.tmp1} can be achieved with particular dynamics.
First, we construct a Markov jump process with the transition rate matrix $\{\msf{W}_t\}$ that transforms the initial distribution ${p^A}$ into the final distribution ${p^B}$.
Let us consider probability path ${p_t}={p^A}+t({p^B}-{p^A})/\tau$.
We then have $\dot{p}_x(t)=(p_x^B-p_x^A)/\tau$, which is invariant for all times $t$.
We next define $S_+\coloneqq\qty{x\,|\,p_x^B\ge p_x^A}$ and $S_-\coloneqq\qty{x\,|\,p_x^B<p_x^A}$.
Then, $\sum_{x\in S_+}\dot{p}_x(t)=-\sum_{x\in S_-}\dot{p}_x(t)$.
Let $\phi>0$ be an arbitrary real positive number.
According to Prop.~\ref{prop:pre.mat}, nonnegative coefficients $\{z_{xy}\}$ exist such that
\begin{align}
\sum_{y\in S_-}z_{xy}&=\frac{\dot{p}_x(t)}{1-e^{-\phi}},~\forall x\in S_+,\\
\sum_{x\in S_+}z_{xy}&=\frac{-\dot{p}_y(t)}{1-e^{-\phi}},~\forall y\in S_-.
\end{align}
Using these coefficients, we consider the following transition rates:
\begin{align}
w_{xy}(t)&=\frac{z_{xy}}{p_y(t)},~\forall x\in S_{+},\,y\in S_-,\\
w_{yx}(t)&=e^{-\phi}\frac{z_{xy}}{p_x(t)},~\forall x\in S_{+},\,y\in S_-,\\
w_{xy}(t)&=0,~\text{otherwise}.
\end{align}
With these transition rates, we can verify that
\begin{equation}
\dot{p}_x(t)=\sum_{y(\neq x)}j_{xy}(t),~\forall x.
\end{equation}
Moreover, the irreversible entropy production rate and dynamical state mobility associated with this Markov jump process can be calculated as
\begin{align}
\sigma_t&=\sum_{x\in S_+,y\in S_-}j_{xy}(t)\ln\frac{a_{xy}(t)}{a_{yx}(t)}=\phi\sum_{x>y}|j_{xy}(t)|,\\
 m_t&=\sum_{x\in S_+,y\in S_-}\frac{j_{xy}(t)}{\ln[a_{xy}(t)/a_{yx}(t)]}={\phi}^{-1}\sum_{x>y}|j_{xy}(t)|.
\end{align}
In addition, note that 
\begin{align}
\sum_{x>y}|j_{xy}(t)|&=(1-e^{-\phi})\sum_{x\in S_+}\sum_{y\in S_-}z_{xy}\notag\\
&=\sum_{x\in S_+}\dot{p}_x(t).
\end{align}
Consequently, we have
\begin{align}
\sqrt{\Sigma_\tau\mca{M}_\tau}&=\int_0^\tau\sum_{x>y}|j_{xy}(t)|\dd{t}\notag\\
&=\sum_{x\in S_+}\int_0^\tau\dot{p}_x(t)\dd{t}\notag\\
&=\sum_{x\in S_+}[p_x(\tau)-p_x(0)]\notag\\
&=\sum_{x\in S_+}(p_x^B-p_x^A)\notag\\
&=\mca{T}(p^A,p^B)\notag\\
&=\mca{W}_q(\varrho^A,\varrho^B).\label{eq:thm2.equality}
\end{align}

We next construct Lindblad dynamics that transforms $\varrho^A$ into $\varrho^B$ and simultaneously satisfies the equality \eqref{eq:thm2.equality}.
For each pair of positive transition rates $\{w_{xy}(t),w_{yx}(t)\}$, we define the corresponding jump operators $\tilde L_k(t)=\sqrt{w_{xy}(t)}\dyad{x^A}{y^A}$ and $\tilde{L}_{k'}(t)=\sqrt{w_{yx}(t)}\dyad{y^A}{x^A}$.
We consider the following Lindblad equation:
\begin{equation}\label{eq:tilde.Lindblad.eq}
\dot{\tilde\varrho}_t=\sum_kD[\tilde L_k(t)]\tilde\varrho_t.
\end{equation}
As the initial state is diagonal in the eigenbasis $\{\ket{x^A}\}$, Eq.~\eqref{eq:tilde.Lindblad.eq} is equivalent to the classical Markov jump process previously constructed.
Given the initial state $\tilde\varrho_0=\varrho^A$, we can easily see that $\tilde\varrho_t$ is always diagonal in the eigenbasis $\{\ket{x^A}\}$ [i.e., $\tilde\varrho_t=\sum_xp_x(t)\dyad{x^A}$].
Moreover, from Eq.~\eqref{eq:thm2.equality}, it is evident that
\begin{equation}\label{eq:tilde.dyn.equ.sat}
\sqrt{\tilde{\Sigma}_\tau\tilde{\mca{M}}_\tau}=\mca{W}_q(\varrho^A,\varrho^B).
\end{equation}
Now, consider the unitary operator $U_\tau=\sum_x\dyad{x^B}{x^A}$. A Hermitian Hamiltonian $H$ exists such that $U_\tau=e^{-iH\tau}$.
Using this Hamiltonian, we consider the following Lindblad dynamics:
\begin{equation}\label{eq:final.Lindblad.eq}
\dot\varrho_t=-i[H,\varrho_t]+\sum_kD[L_k(t)]\varrho_t,
\end{equation}
where jump operators are given by $L_k(t)=U_t\tilde L_k(t)U_t^\dagger$, and $U_t\coloneqq e^{-iHt}$.
The density matrix $\varrho_t$ is related to that in Eq.~\eqref{eq:tilde.Lindblad.eq} as $\varrho_t=U_t\tilde\varrho_tU_t^\dagger$.
We can confirm that the dynamics \eqref{eq:final.Lindblad.eq} transforms the density matrix $\varrho_0=\varrho^A$ into $\varrho_\tau=\varrho^B$, and irreversible entropy production and dynamical state mobility remain unchanged:
\begin{align}
\Sigma_\tau&=\tilde\Sigma_\tau,\\
\mca{M}_\tau&=\tilde{\mca{M}}_\tau.
\end{align}
Combining this with Eq.~\eqref{eq:tilde.dyn.equ.sat}, we can show that the inequality \eqref{eq:Wq.bound.tmp1} can be saturated as
\begin{equation}
\sqrt{\Sigma_\tau\mca{M}_\tau}=\mca{W}_q(\varrho^A,\varrho^B).
\end{equation}

\subsection{Quantum variational formula in terms of entropy production and dynamical activity}\label{app:qua.cor.proof}
\begin{corollary}\label{cor:qua.dis.Wass.var1}
The quantum Wasserstein distance can be expressed in terms of irreversible entropy production and dynamical activity as
\begin{align}
\mca{W}_q( \varrho^A,\varrho^B )&=\min_{\mca{L}_t}{\int_0^\tau\frac{\sigma_t}{2}\Phi\qty(\frac{\sigma_t}{2a_t})^{-1}\dd{t}}\label{eq:Wq.var.form3}\\
&=\min_{\mca{L}_t}{\frac{\Sigma_\tau}{2}\Phi\qty(\frac{\Sigma_\tau}{2\mca{A}_\tau})^{-1}}.\label{eq:Wq.var.form4}
\end{align}
\end{corollary}
\begin{proof}
First, we prove that
\begin{equation}\label{eq:app.qua.cor.tmp1}
\mca{W}_q( \varrho^A,\varrho^B )\le {\int_0^\tau\frac{\sigma_t}{2}\Phi\qty(\frac{\sigma_t}{2a_t})^{-1}\dd{t}}\le {\frac{\Sigma_\tau}{2}\Phi\qty(\frac{\Sigma_\tau}{2\mca{A}_\tau})^{-1}}.
\end{equation}
Noting that $x\Phi(x/y)^{-1}$ is a concave function and
\begin{align}
\sigma_t&=\frac{1}{2}\sum_{k,x,y}j_k^{xy}(t)\ln\frac{a_k^{xy}(t)}{a_{k'}^{yx}(t)}\eqqcolon\frac{1}{2}\sum_{k,x,y}\sigma_k^{xy}(t),\\
a_t&=\frac{1}{2}\sum_{k,x,y}[a_k^{xy}(t)+a_{k'}^{yx}(t)],
\end{align}
we obtain the following result from Jensen's inequality:
\begin{align}
\frac{\sigma_t}{2}\Phi\qty(\frac{\sigma_t}{2a_t})^{-1}&\ge\sum_{k,x,y}\frac{\sigma_k^{xy}(t)}{4}\Phi\qty(\frac{\sigma_k^{xy}(t)}{2[a_k^{xy}(t)+a_{k'}^{yx}(t)]})^{-1}\notag\\
&=\frac{1}{2}\sum_{k,x,y}|j_k^{xy}(t)|\notag\\
&\ge \frac{1}{2}\sum_x|\dot{p}_x(t)|.
\end{align}
Taking the time integration, we can immediately prove Eq.~\eqref{eq:app.qua.cor.tmp1}:
\begin{align}
\mca{W}_q(\varrho^A,\varrho^B)&\le \frac{1}{2}\int_0^\tau\sum_x|\dot{p}_x(t)|\dd{t}\notag\\
&\le \int_0^\tau\frac{\sigma_t}{2}\Phi\qty(\frac{\sigma_t}{2a_t})^{-1} \dd{t}\notag\\
&\le \frac{\Sigma_\tau}{2}\Phi\qty(\frac{\Sigma_\tau}{2\mca{A}_\tau})^{-1}.
\end{align}
Next, we show that the equalities in Eq.~\eqref{eq:app.qua.cor.tmp1} can be attained with the dynamics constructed in the proof of Thm.~\ref{thm:qua.dis.Wass.var}.
Notice that the density matrix $\tilde{\varrho}_t$ of Lindblad dynamics \eqref{eq:tilde.Lindblad.eq} can be expressed as $\tilde{\varrho}_t=\sum_xp_x(t)\dyad{x^A}$.
Therefore, the density matrix $\varrho_t$ of Lindblad dynamics \eqref{eq:final.Lindblad.eq} reads $\varrho_t=\sum_xp_x(t)U_t\dyad{x^A}U_t^\dagger$, the time-dependent eigenvectors of which are $\ket{x_t}=U_t\ket{x^A}$.
The jump operators are given by $L_k(t)=\sqrt{w_{xy}(t)}U_t\dyad{x^A}{y^A}U_t^\dagger$.
Using these quantities, we can calculate
\begin{align}
w_k^{x'y'}(t)&=\delta_{xx'}\delta_{yy'}w_{xy}(t).
\end{align}
In addition, the entropy production and dynamical activity rates can be calculated as
\begin{align}
\sigma_t&=\phi\sum_{x>y}|j_{xy}(t)|,\\
a_t&=\coth(\phi/2)\sum_{x>y}|j_{xy}(t)|,
\end{align}
where $\{j_{xy}(t)\}$ are probability currents in the classical Markov jump process.
Consequently, we obtain the following relations:
\begin{align}
\frac{\sigma_t}{2}\Phi\qty(\frac{\sigma_t}{2a_t})^{-1}&=\frac{1}{2}\phi\sum_{x>y}|j_{xy}(t)|\Phi\qty(\frac{\phi}{2\coth(\phi/2)})^{-1}\notag\\
&=\sum_{x>y}|j_{xy}(t)|,\label{eq:app.qua.cor.tmp2}\\
{\frac{\Sigma_\tau}{2}\Phi\qty(\frac{\Sigma_\tau}{2\mca{A}_\tau})^{-1}}&=\int_0^\tau\sum_{x>y}|j_{xy}(t)|\dd{t}.\label{eq:app.qua.cor.tmp3}
\end{align}
Combining Eqs.~\eqref{eq:app.qua.cor.tmp2}, \eqref{eq:app.qua.cor.tmp3}, and \eqref{eq:thm2.equality} verifies the equalities of Eq.~\eqref{eq:app.qua.cor.tmp1}.
\end{proof}

\section{Derivation of calculations in Sec.~\ref{sec:appl}}

\subsection{Thermodynamic speed limit in terms of the trace distance}\label{app:qsl.trace.dist}
The Wasserstein distance is used as a metric between quantum states in the speed limits in Eq.~\eqref{eq:qsl1}.
Here, we show that another thermodynamic speed limit with a different metric can also be obtained.
Specifically, we derive a speed limit using the trace distance in the following.
Let $\varrho_t=\sum_xp_x(t)\dyad{x_t}$ be the spectral decomposition of the density matrix $\varrho_t$. Then, as previously shown in Ref.~\cite{Funo.2019.NJP}, we have
\begin{equation}
\|\dot{\varrho}_t\|_1\le 2(\Delta H_t+\Delta H_t^D)+\sum_x|\dot{p}_x(t)|,
\end{equation}
where
\begin{align}
(\Delta H_t)^2&=\tr{H_t^2\varrho_t}-(\tr{H_t\varrho_t})^2,\\
(\Delta H_t^D)^2&=\tr{(H_t^D)^2\varrho_t}-(\tr{H_t^D\varrho_t})^2,\\
H_t^D&\coloneqq\sum_{x\neq y}\frac{i\mel{x_t}{\sum_k\mca{D}[L_k]\varrho_t}{y_t}}{p_y(t)-p_x(t)}\dyad{x_t}{y_t}.
\end{align}
In addition, as shown in Eq.~\eqref{eq:proof.pnt.bound}, we can prove that
\begin{equation}\label{eq:pnt.dot.bound}
\sum_x|\dot p_x(t)|\le 2\sqrt{\sigma_t m_t}.
\end{equation}
Taking the time integration and using the triangle inequality for the trace norm, we obtain
\begin{align}
\mca{T}(\varrho_0,\varrho_\tau)&\le\frac{1}{2}\int_0^\tau\|\dot\varrho_t\|_1\dd{t}\notag\\
&\le\tau\ev{\Delta H+\Delta H^D+\sqrt{\sigma m}}_\tau,
\end{align}
which yields the following speed limit:
\begin{align}
\tau&\ge\frac{\mca{T}(\varrho_0,\varrho_\tau)}{\ev{\Delta H+\Delta H^D+\sqrt{\sigma m}}_\tau}\notag\\
&\ge \frac{\mca{T}(\varrho_0,\varrho_\tau)}{\ev{\Delta H}_\tau+\ev{\Delta H^D}_\tau+{\sqrt{\ev{\sigma}_\tau\ev{m}_\tau}}}.\label{eq:trace.dist.qsl}
\end{align}
Since $\ev{m}_\tau\le\ev{a}_\tau/2$, this new speed limit is stronger than the bound reported in Ref.~\cite{Funo.2019.NJP}, which reads
\begin{equation}
\tau\ge \frac{\mca{T}(\varrho_0,\varrho_\tau)}{\ev{\Delta H}_\tau+\ev{\Delta H^D}_\tau+{\sqrt{\ev{\sigma}_\tau\ev{a}_\tau/2}}}
\end{equation}
In the classical limit, the speed limit \eqref{eq:trace.dist.qsl} reduces to the following bound:
\begin{align}
\tau\ge \frac{\mca{T}(p^A,p^B)}{\ev{\sqrt{{\sigma}{m}}}_\tau} \ge\frac{\mca{T}(p^A,p^B)}{\sqrt{\ev{\sigma}_\tau\ev{m}_\tau}}.
\end{align}

\subsection{Proof of Eq.~\eqref{eq:erasure.error.bound}}\label{app:cla.erase.guar}
Equation \eqref{eq:erasure.error.bound} is the consequence of the following lemma.
\begin{lemma}
If an erasure protocol satisfies $\|\Lambda_\tau\overline{p}-{p_*}\|_F\le\delta$, where $\delta>0$ is a sufficiently small number, then for an arbitrary probability distribution ${p_0}$, the following inequality holds:
\begin{equation}
\|\Lambda_\tau{p_0}-{p_*}\|_F\le \sqrt{2d\delta}\xrightarrow{\delta\to 0} 0.
\end{equation}
\end{lemma}
\begin{proof}
For any probability distribution ${p_0}$, a distribution ${p_0'}$ always exists such that ${p_0}+(d-1){p_0'}=\vb*{1}=d{p_*}$.
Indeed, the distribution ${p_0'}$ can be chosen as ${p_0'}=(\vb*{1}-{p_0})/(d-1)$.
Here, $\vb*{1}=[1,\dots,1]^\top$ is the all-one vector.
We then define ${p_\tau}\coloneqq \Lambda_\tau{p_0}$ and ${p_\tau'}\coloneqq \Lambda_\tau{p_0'}$ and obtain the following relation:
\begin{equation}
\Lambda_\tau\overline{p}=\frac{1}{d}\Lambda_\tau({p_0}+(d-1){p_0'})=\frac{1}{d}{p_\tau}+\frac{d-1}{d}{p_\tau'}.
\end{equation}
Therefore, the condition $\|\Lambda_\tau\overline{p}-{p_*}\|_F\le\delta$ is equivalent to
\begin{equation}
\|({p_\tau}+(d-1){p_\tau'})/d-{p_*}\|_F^2\le\delta^2.\label{eq:fnorm.tmp1}
\end{equation}
It suffices to prove that $\|{p_\tau}-{p_*}\|_F^2\le 2d\delta$.
From Eq.~\eqref{eq:fnorm.tmp1}, we have
\begin{align}
|[p_{\tau,1}+(d-1)p_{\tau,1}']/d-p_{*,1}|^2&\le \|[{p_\tau}+(d-1){p_\tau'}]/d-{p_*}\|_F^2\notag\\
&\le\delta^2.
\end{align}
Consequently,
\begin{align}
\frac{1-p_{\tau,1}}{d}+(d-1)\frac{1-p_{\tau,1}'}{d}\le{\delta}\Rightarrow 1-d{\delta}\le p_{\tau,1}\le 1.\label{eq:fnorm.tmp2}
\end{align}
The last inequality in Eq.~\eqref{eq:fnorm.tmp2} immediately derives $|1-p_{\tau,1}|\le d{\delta}$ and $|p_{\tau,1}|\ge 1-d{\delta}\ge 0$.
From the inequality $\sum_{n=1}^{d}|p_{\tau,n}|^2\le \sum_{n=1}^{d}p_{\tau,n}=1$, the partial sum $\sum_{n=2}^{d}|p_{\tau,n}|^2$ can be upper bounded as
\begin{equation}
\sum_{n=2}^{d}|p_{\tau,n}|^2\le 1-|p_{\tau,1}|^2\le 1-(1-d{\delta})^2=2d{\delta}-d^2\delta^2.
\end{equation}
Combining these inequalities, we obtain
\begin{equation}
\|{p_\tau}-{p_*}\|_F^2=|1-p_{\tau,1}|^2+\sum_{n=2}^{d}|p_{\tau,n}|^2\le d^2\delta^2+2d{\delta}-d^2\delta^2=2d{\delta},
\end{equation}
which completes the proof.
\end{proof}

\end{document}